\documentclass[a4paper,11pt]{article}

\usepackage{fullpage}%
\usepackage{graphicx,tikz,enumerate}
\usepackage[colorlinks=true]{hyperref}
\usepackage[title]{appendix}
\usepackage{amsmath, amssymb, amsfonts,mathtools,stmaryrd}
\usepackage{xcolor}
\usepackage{enumitem}
\usepackage{etoolbox}
\usepackage{ntheorem}
\usepackage{authblk}

\usepackage{dsfont}

\newtheorem{definition}{\textbf{Definition}}
\newtheorem{theorem}{\textbf{Theorem}}
\newtheorem{lemma}{\textbf{Lemma}}
\newtheorem{corollary}{\textbf{Corollary}}
\newtheorem{remark}{\textbf{Remark}}
\newtheorem{proposition}{\textbf{Proposition}}

\newtheorem*{proof}{\textbf{Proof}}
\AtEndEnvironment{proof}{$\square$}

\newcommand{\GF}{\mathsf{GF}}

\newcommand{\formNF}{\mathcal{F}}
\newcommand{\formN}{\mathcal{F}^\mathsf{nf}}
\newcommand{\gameNF}[2]{\langle #1,#2 \rangle}

\newcommand{\outComeNF}{\mathsf{O}}
\newcommand{\outCNF}{\varrho}

\newcommand{\Head}{\ensuremath{\mathsf{head}}}

\newcommand{\prefUniq}[1]{\mathsf{Pref}(#1)}

\newcommand{\A}{\mathsf{A} }
\newcommand{\B}{\mathsf{B} }

\newcommand{\setA}{A}
\newcommand{\setB}{B}

\newcommand{\colSet}{\mathsf{K} }

\newcommand{\s}{\ensuremath{\mathsf{s}}}

\newcommand{\Opt}{\ensuremath{\mathsf{Opt}}}

\newcommand{\colFunc}{\mathsf{col} }
\newcommand{\Dist}{\mathcal{D} }

\newcommand{\Supp}{\mathsf{Supp} }
\newcommand{\distribSet}{\mathsf{D} }
\newcommand{\distribFunc}{\mathsf{dist} }
\newcommand{\prob}[2]{\mathbb{P}^{#1}_{#2} }

\newcommand{\MarVal}[1]{\chi_{#1} }
\newcommand{\LiftVal}[1]{\nu_{#1} }

\newcommand{\G}{\mathcal{G} }
\newcommand{\Aconc}{\mathcal{C} }

\newcommand{\Games}[2]{\langle #1,#2 \rangle }
\newcommand{\AparamConc}{\ensuremath{\langle 
		Q,(\setA_q)_{q \in Q},(\setB_q)_{q \in Q},\distribSet,\delta,\distribFunc,\colSet,\colFunc \rangle}}

\newcommand{\SetStrat}[2]{\ensuremath{\mathsf{S}_{#1}^{#2} }}

\newcommand{\minit}{m_{\text{\upshape init}}}

\newcommand{\restr}[2]{\ensuremath{\left.#1\right|_{#2}}}

\newcommand{\cyl}{\ensuremath{\mathsf{Cyl}}}

\newcommand{\Borel}{\mathsf{Borel}}

\newcommand{\outM}{\ensuremath{\mathsf{out}}}
\newcommand{\va}{\ensuremath{\mathsf{val}}}
\newcommand{\N}{\ensuremath{\mathbb{N}}}

\date{Université Paris-Saclay, CNRS, ENS Paris-Saclay, LMF, 91190 Gif-sur-Yvette, France}

\begin{document}
	\author{Benjamin Bordais, Patricia Bouyer and Stéphane Le Roux}
	
	\title{Sub-game optimal strategies in concurrent games with prefix-independent objectives}
	\maketitle
	
	\begin{abstract}
		We investigate concurrent two-player win/lose stochastic games on
		finite graphs with prefix-independent objectives. We characterize
		subgame optimal strategies and use this characterization to show
		various memory transfer results: 1) For a given (prefix-independent)
		objective, if every game that has a subgame \emph{almost-surely
			winning} strategy also has a positional one, then every game that
		has a subgame \emph{optimal} strategy also has a positional one;
		2) Assume that the (prefix-independent) objective has a neutral
		color. If every \emph{turn-based} game that has a subgame
		almost-surely winning strategy also has a positional one, then every
		game that has a \emph{finite-choice} (notion to be defined)
		subgame optimal strategy also has a positional one.
		
		We collect or design examples to show that our results are tight in
		several ways. We also apply our results to Büchi, co-Büchi, parity,
		mean-payoff objectives, thus yielding simpler statements.
	\end{abstract}
	
	\section{Introduction}
	
	Turn-based two-player win/lose (stochastic) games on finite graphs have been
	intensively studied in the context of model checking in a broad
	sense~\cite{thomas02,BCJ18}.
	These games behave well regarding optimality in various settings. Most
	importantly for this paper, \cite{GIH10} proved the following results
	for finite turn-based stochastic games with prefix-independent
	objectives: (1) every game has deterministic optimal strategies; (2)
	from every value-$1$ state, there is an optimal, i.e. almost-surely
	winning, strategy; (3) if from every value-$1$ state of every game
	there is an optimal strategy using some fixed amount of memory, every
	game has an optimal strategy using this amount of memory. These
	results are of either of the following generic forms:
	\begin{itemize}
		\item In all games, (from all nice states) there is a nice strategy.
		
		\item If from all \emph{nice states} of all games there is a nice
		strategy, so it is from all \emph{states}.
	\end{itemize}
	
	%
	%
	The concurrent version of these turn-based (stochastic) games has a
	higher modeling power than the turn-based version: this is really
	useful in practice since real-world systems are intrinsically
	concurrent~\cite{KNPS21}. They are played on a finite graph as
	follows: at each player state, 
	the two players stochastically and independently choose one among
	finitely many actions. This yields a Nature state, which
	stochastically draws a next player state, from where each player
	chooses one action again, and so on.  Each player state is labelled by
	a color, and who wins depends on the infinite sequence of colors
	underlying the (stochastically) generated infinite sequence of player
	states.
	Unfortunately, these concurrent games do
	not behave well in general even for simple winning conditions and
	simple graph structures, like finite graphs:
	\begin{itemize}
		\item Reachability objectives: there is a game without optimal
		strategies \cite{everett57};
		\item B\"uchi objectives: there is a game with value $1$ while all
		finite-memory strategies have value $0$ \cite{AH00};
		\item Co-B\"uchi objectives: although there are always positional
		$\varepsilon$-optimal strategies \cite{AH06}, there is a game with optimal
		strategies but without finite-memory optimal strategies \cite{BBSArXivICALP}; 
		\item Parity \cite{AH00} and mean-payoff
		\cite{DBLP:journals/iandc/ChatterjeeI15} objectives: there is a game
		with subgame almost-surely-winning strategies, but where all
		finite-memory strategies have value $0$.
	\end{itemize}
	
	In this paper, we focus on concurrent stochastic finite
	games. Therefore, the generic forms of our results will be more
	complex, in order to take into account the above-mentioned
	discrepancies. They will somehow 
	be given as generic statements as follows:
	\begin{itemize}
		\item Every game that has a \emph{nice} strategy also has a
		\emph{nicer} one.		
		\item If all \emph{special games} that have a nice strategy have a
		nicer one, so it is for all \emph{games}.
	\end{itemize}
	Much of the difficulty consists in fine-tuning the strength of
	``nice'', ``nicer'' and ``special'' above. We present below our main
	contributions on finite two-player win/lose concurrent stochastic
	games with prefix-independent objectives:
	\begin{enumerate}
		\item\label{main1} We provide a characterization of subgame optimal
		strategies, which are strategies that are optimal after every
		history (Theorem~\ref{thm:subgame_optimal_arbitrary_strategy}): a Player $\A$
		strategy is subgame optimal iff 1) it is
		locally optimal and 2) for every Player $\B$ deterministic strategy,
		after every history, if the visited states have the same positive
		value, Player $\A$ wins with probability $1$. This characterization is 
		used to prove all the results below.
		
		\item\label{main2} We prove memory transfer results from subgame
		almost-surely winning strategies to subgame optimal strategies:
		\begin{enumerate}
			\item\label{main21} Theorem~\ref{thm:transfer_memory}: If every game that has
			a subgame
			\emph{almost-surely winning} strategy also has a positional one,
			then every game that has a subgame \emph{optimal} strategy also
			has a positional one.
			
			\item Corollary~\ref{coro:application_buchi_cobuchi}: every Büchi or co-Büchi
			game that has a subgame
			optimal strategy has a positional one. (Whereas parity games may
			require infinite memory \cite{AH00}.)
		\end{enumerate}
		Note that the transfer result~{\bfseries\sffamily
			\color{darkgray}\ref{main21}} is generalized from positional to
		finite memory in appendix.
		
		%
		
		\item We say that a strategy has finite-choice, if it uses only
		finitely many action distributions. Note that finite-memory
		(resp. deterministic) strategies clearly have finite choice.
		\begin{enumerate}
			\item\label{main31} Theorem~\ref{thm:fair_strategy_uniformly_optimal}: In a
			given game, if there is a finite-choice
			optimal strategy, there is a finite-choice \emph{subgame} optimal
			strategy.
			
			\item\label{main32} Theorem~\ref{thm:transfer_memory_finite_choice}: Assume
			that the objective has a neutral
			color. If every \emph{turn-based} game that has a subgame
			almost-surely winning strategy also has a positional one, then
			every game that has a \emph{finite-choice} subgame optimal
			strategy also has a positional one.
			
			\item\label{main33} Corollary~\ref{coro:transfer_memory_finite_choice}: every
			parity or mean-payoff game that
			has a finite-memory subgame optimal strategy also has a
			positional one.
		\end{enumerate}
		Note
		that {\bfseries\sffamily \color{darkgray}\ref{main31}} and
		{\bfseries\sffamily \color{darkgray}\ref{main32}} are false if the
		word finite-choice is removed \cite{BBSArXivICALP}. The proof of
		{\bfseries\sffamily \color{darkgray}\ref{main32}} invokes
		{\bfseries\sffamily \color{darkgray}\ref{main31}}. Flavor (and
		proofs) of {\bfseries\sffamily \color{darkgray}\ref{main32}} and
		{\bfseries\sffamily \color{darkgray}\ref{main21}} are similar, but
		both premises and conclusions are weakened in {\bfseries\sffamily
			\color{darkgray}\ref{main32}}, as emphasized.
	\end{enumerate}
	
	\
	
	\textbf{Related works.}  A large part of this paper is
	dedicated to the extension to concurrent games of the results from
	\cite{GIH10} regarding the transfer of memory from almost-surely
	winning strategies to optimal strategies in turn-based games. Note
	that the proof technique used in \cite{GIH10} is different 
	and could
	not be adapted to our more general setting. In their proof, both
	players agree on a preference over Nature states and play
	according to this preference. 
	In our proof, we slice the graph into value areas (that is, sets of
	states with the same value), and show that it is sufficient to play an
	almost-sure winning strategy in each slice; we then glue these (partial)
	strategies together to get a subgame-optimal strategy over the whole
	graph.
	
	
	The slicing technique was already used in the context of concurrent
	games in~\cite{AH06}. The authors focus on parity objectives and
	establishes a memory transfer result from limit-sure winning
	strategies to almost-optimal strategies. As an application, they show
	that, for co-Büchi objectives, since positional strategies are
	sufficient to win limit-surely, they also are to win
	almost-optimally. Their construction made heavy use of the specific
	nature of parity objectives.

	
	We also mention \cite{CH07}, where the focus is also on concurrent
	games with prefix-independent objectives. In particular, the authors
	establish a (very useful) result: if all states have positive
	values, then they all have value 1. (Note that a strengthening of
	this result is presented in this paper
	(Theorem~\ref{thm:all_not_zero_implies_almost_sure}), which also
	appears as an adaptation of a result proved in \cite{GIH10}). 
	This result is then used in another context with non-zero-sum
	games. 
	%
	
	Finally, some recent works on concurrent games have been done in
	\cite{BBSFSTTCS21,BBSCSL22,BBSArXivICALP}, where the goal is the
	following: local interactions of the two players in the player
	state are given by bi-dimensional tables; those tables can be
	abstracted as \emph{game forms}, where (output) variables are
	issues of the local interaction (possibly several issues are
	labelled by the same variable). The goal of this series of works
	is to give (intrinsic) properties of these game forms, so that,
	when used in a graph game, the existence of optimal strategies is
	ensured. For instance, in~\cite{BBSCSL22}, a property of games
	forms, called RM, is given, which ensures that, if one only uses
	RM game forms in a graph, then for every reachability objective,
	Player $\A$ will always have an optimal strategy for that
	objective. This property is a characterization of well-behaved
	game forms regarding reachability objectives since every game form
	which is not RM can be embedded into a (small) graph game in such
	a way that Player $\A$ does not have an optimal strategy. This
	line of works really differs from the target of the current
	paper.
	
	
	
	
	
	\textbf{Structure of the paper.}  Section~\ref{sec:prelim}
	presents notations, Section~\ref{sec:gf} recalls the notion of game
	forms, Section~\ref{sec:games} introduces our formalism,
	Section~\ref{sec:conditions_subgame_optimal} exhibits a necessary and
	sufficient pair of conditions for subgame optimality,
	Section~\ref{sec:transfer} shows a memory transfer from subgame
	almost-surely winning to subgame optimal in concurrent games, and
	Section~\ref{sec:fair_strategies} adapts the results of the previous
	section to the case of the existence of a subgame finite-choice
	strategy.
	
	\section{Preliminaries}
	\label{sec:prelim}
	Consider a non-empty set $Q$. We denote by $Q^*$, $Q^+$ and $Q^\omega$ the set
	of finite sequences, non-empty finite sequences and infinite sequences of
	elements of $Q$ respectively. For $n \in \N$, we denote by $Q^n$ (resp. $Q^{\leq
		n}$) the set of sequences of (resp. at most) $n$ elements of $Q$. For all $\rho
	= q_1 \cdots q_n \in Q^n$ and $i \leq n$, we denote by $\rho_i$ the element $q_i
	\in Q$ and by $\rho_{\leq i} \in Q^i$ the finite sequence $q_1 \cdots q_i$. 
	For a subset $S \subseteq Q$, we denote by $Q^* \cdot S^\omega \subseteq
	Q^\omega$ the set of infinite paths that eventually settle in $S$ and by $(Q^*
	\cdot S)^\omega \subseteq Q^\omega$ the set of infinite paths visiting
	infinitely often the set $S$.
	
	A \emph{discrete probability distribution} over a non-empty finite set $Q$ is a function $\mu: Q \rightarrow [0,1]$ such that $\sum_{x \in Q} \mu(x) = 1$. The	\emph{support} $\Supp(\mu)$ of a probability distribution $\mu: Q \rightarrow [0,1]$ is the	set of non-zeros of the distribution: $\Supp(\mu) = \{ q \in Q \mid \mu(q) \in
	(0,1] \}$. 
	The set of all distributions over 
	$Q$ is denoted $\Dist(Q)$.
	
	\section{Game forms}
	\label{sec:gf}
	
	We recall the definition of game forms -- 
	informally, bi-dimensional tables with variables -- and of games in normal
	forms -- game forms whose outcomes are values between $0$ and $1$.
	\begin{definition}[Game form and game in normal form]
		\label{def:arena_game_nf}
		A \emph{game form} (GF for short) is a tuple
		$\formNF = \langle \mathsf{Act}_\A,\mathsf{Act}_\B,\outComeNF,\outCNF \rangle$
		where $\mathsf{Act}_\A$ (resp. $\mathsf{Act}_\B$) is the non-empty finite set
		of 
		actions available to Player $\A$ (resp. $\B$), $\outComeNF$
		is a non-empty set of outcomes, and
		$\outCNF: \mathsf{Act}_\A \times \mathsf{Act}_\B \rightarrow \outComeNF$ is a
		function that associates an outcome to each pair of
		actions. When the set of outcomes $\outComeNF$ is equal to
		$[0,1]$, we say that $\formNF$ is a \emph{game in normal
			form}. For a valuation $v \in [0,1]^\outComeNF$ of the
		outcomes, the notation $\gameNF{\formNF}{v}$ refers to the
		game in normal form
		$\langle \mathsf{Act}_\A,\mathsf{Act}_\B,[0,1],v \circ \outCNF \rangle$.
	\end{definition}
	We use 
	game forms to represent interactions between two
	players. The strategies available to Player $\A$ (resp. $\B$) are
	convex combinations of actions given as the rows
	(resp. columns) of the table. In a game in normal form, Player $\A$
	tries to maximize the outcome, whereas Player $\B$ tries to minimize
	it.
	\begin{definition}[Outcome of a game in normal form]
		\label{def:outcome_game_form}
		Consider a game in normal form
		$\formNF = \langle \mathsf{Act}_\A,\mathsf{Act}_\B,[0,1],\outCNF \rangle$. The
		set $\Dist(\mathsf{Act}_\A)$ (resp. $\Dist(\mathsf{Act}_\B)$) is the
		set of strategies available to Player $\A$
		(resp. $\B$). For a pair of strategies
		$(\sigma_\A,\sigma_\B) \in \Dist(\mathsf{Act}_\A) \times
		\Dist(\mathsf{Act}_\B)$, the outcome
		$\outM_\formNF(\sigma_\A,\sigma_\B)$ in $\formNF$ of the
		strategies $(\sigma_\A,\sigma_\B)$ is defined as:
		\begin{equation*}
			\outM_\formNF(\sigma_\A,\sigma_\B) := \sum_{a \in \mathsf{Act}_\A}
			\sum_{b \in \mathsf{Act}_\B} \sigma_\A(a) \cdot \sigma_\B(b) \cdot
			\outCNF(a,b) \in [0,1]
		\end{equation*}
	\end{definition}	
	
	\begin{definition}[Value of a game in normal form and optimal strategies]
		\label{def:alternative_value_game_normal_form}
		Consider a game in normal form
		$\formNF = \langle \mathsf{Act}_\A,\mathsf{Act}_\B,[0,1],\outCNF \rangle$ and
		a
		strategy $\sigma_\A \in \Dist(\mathsf{Act}_\A)$ for Player $\A$.  The
		\emph{value} of the strategy $\sigma_\A$, denoted
		$\va_\formNF(\sigma_\A)$ is equal to:
		$\va_\formNF(\sigma_\A) := \inf_{\sigma_\B \in \Dist(\mathsf{Act}_\B)}
		\outM_{\formNF}(\sigma_\A,\sigma_\B)$, and analogously for Player
		$\B$, with a $\sup$ instead of an $\inf$. When
		$\sup_{\sigma_\A \in \Dist(\mathsf{Act}_\A)} \va_\formNF(\sigma_\A) =
		\inf_{\sigma_\B \in \Dist(\mathsf{Act}_\B)} \va_\formNF(\sigma_\B)$, it
		defines the \emph{value} of the game $\formNF$, denoted
		$\va_\formNF$.
		
		A strategy $\sigma_\A \in \Dist(\mathsf{Act}_\A)$ ensuring
		$\va_\formNF = \va_\formNF(\sigma_\A)$ is called \emph{optimal}. The
		set of all optimal strategies for Player $\A$ is denoted
		$\Opt_\A(\formNF) \subseteq \Dist(\mathsf{Act}_\A)$, and analogously for
		Player $\B$. Von Neuman's minimax theorem~\cite{vonNeuman} ensures the
		existence of
		optimal strategies (for both players). 
	\end{definition}
	In the following, strategies in games in normal forms will be called
	$\GF$-strategies, in order not to confuse them with strategies in
	concurrent (graph) games.
	
	\section{Concurrent games and optimal strategies}
	\label{sec:games}
	\subsection{Concurrent arenas and strategies}
	We introduce the definition of concurrent arenas played on a finite graph.
	\begin{definition}[Finite stochastic concurrent arena]
		A \emph{colored concurrent arena} $\Aconc$ is a tuple 
		$\AparamConc$ where $Q$ is the non-empty finite set of states, for all $q \in
		Q$, $\setA_q$ (resp. $\setB_q$) is the non-empty finite set of 
		actions available to Player $\A$ (resp. $\B$) at state $q$
		, $\distribSet$ is the finite set of Nature states, $\delta: \bigcup_{q \in Q}
		(\{ q \} \times \setA_q \times \setB_q) \rightarrow \distribSet$ is the
		transition function, $\distribFunc: \distribSet \rightarrow \Dist(Q)$ is the
		distribution function. Furthermore, $\colSet$ is the non-empty finite set of
		colors and $\colFunc: Q \rightarrow \colSet$ is the coloring function.
	\end{definition}
	
	In the following, the arena $\Aconc$ will refer to the tuple $\AparamConc$,
	unless otherwise stated. A concurrent game is obtained from a concurrent arena
	by adding a winning condition: the set of infinite paths winning for Player $\A$
	(and losing for Player $\B$).
	\begin{definition}[Finite stochastic concurrent game]
		A finite \emph{concurrent game} 
		is 
		a pair $\Games{\Aconc}{W}$ where $\Aconc$ is a finite concurrent colored arena
		and $W \subseteq \colSet^\omega$ is Borel. The set $W$ is called the
		\emph{objective}, as it corresponds to the set of colored paths winning for
		Player $\A$.
	\end{definition}
	
	In this paper, we only consider a specific kind of objectives:
	prefix-independent ones. Informally, they correspond to objectives $W$ such that
	an infinite path $\rho$ is in $W$ if and only if 
	any of its suffixes is in $W$. More formally:
	\begin{definition}[Prefix-independent objectives]
		For a non-empty finite set of colors $\colSet$ and 
		$W \subseteq \colSet^\omega$, 
		$W$ is said to be \emph{prefix-independent} (PI for short) if, for all $\rho
		\in \colSet^\omega$ and $i \geq 0$, $\rho \in W \Leftrightarrow \rho_{\geq i}
		\in W$. 
	\end{definition}
	In the following, we refer to concurrent games with prefix-independent
	objectives as PI concurrent games. Our main results will be stated for arbitrary prefix-independent objectives, but we will apply them to specific objectives, namely parity, and relevant special cases Büchi and co-Büchi.	
	\begin{definition}[Parity, Büchi, co-Büchi objectives]
		\label{def:parity_buchi_cobuchi}
		Let $\colSet \subset \N$ be a finite non-empty set of integers. Consider a
		concurrent arena $\Aconc$ with $\colSet$ as set of colors. For an infinite path
		$\rho \in Q^\omega$, we denote by $\colFunc(\rho)_\infty \subseteq \N$ the set
		of colors seen infinitely often in $\rho$: $\colFunc(\rho)_\infty := \{ n \in \N
		\mid \forall i \in \N,\; \exists j \geq i,\; \colFunc(\rho_j) = n \}$. Then, the
		\emph{parity objective} w.r.t. $\colFunc$ is the set
		$W^{\mathsf{Parity}}(\colFunc) := \{ \rho \in Q^\omega \mid \max
		\colFunc(\rho)_\infty \text{ is even } \}$. The Büchi (resp. co-Büchi) objective
		correspond to the parity objective with $\colSet := \{ 1,2 \}$ (resp. $\colSet
		:= \{ 0,1 \}$).
	\end{definition}
	
	Strategies are then defined as functions that, given the history of the game
	(i.e. the sequence of states already seen) associate 
	a distribution on the actions available to the Player.
	\begin{definition}[Strategies]
		Consider a concurrent game $\Aconc
		$. A strategy for Player $\A$  is a function $\s_\A: Q^+ \rightarrow \Dist(A)$
		with $A := \bigcup_{q \in Q} A_q$ such that, for all $\rho = q_0 \cdots q_n \in
		Q^+$, we have $\s_\A(\rho) \in \Dist(A_{q_n})$.		We denote by
		$\SetStrat{\Aconc}{\A}$ the set of all strategies in arena $\Aconc$ for Player
		$\A$. This is analogous for Player $\B$.
	\end{definition}

	We would like to define the outcome of the game given two strategies (one for
	each Player). First, we define the probability to go from a state $q$ to another
	state $q'$ given two $\GF$-strategies at a given state.
	\begin{definition}[Probability Transition]
		\label{def:mu_state}
		Consider a concurrent arena $\Aconc
		$, a state $q \in Q$ and two strategies $(\sigma_\A,\sigma_\B) \in \Dist(A_q)
		\times \Dist(B_q)$
		. Consider another state $q' \in Q$. The probability to go from $q$ to $q'$ if
		the players plays, in q, $\sigma_\A$ and $\sigma_\B$, denoted
		$\prob{q,q'}{}{}(\sigma_\A,\sigma_\B)$, is equal to:
		\begin{displaymath}
		\prob{q,q'}{}{}(\sigma_\A,\sigma_\B) = \sum_{a \in A_q} \sum_{b \in B_q}
		\sigma_\A(a) \cdot \sigma_\B(b) \cdot \distribFunc \circ \delta(q,a,b)(q')
		\end{displaymath}
	\end{definition}
	
	Let us now define the probability of occurrence of any finite path, and
	consequently of any Borel set, given two strategies.
	\begin{definition}[Probability distribution given two strategies]
		Let us consider a concurrent arena $\Aconc
		$ and $\s_\A,\s_\B \in \SetStrat{\Aconc}{\A} \times \SetStrat{\Aconc}{\B}$ two
		arbitrary strategies for Player $\A$ and $\B$. We denote by
		$\mathbb{P}^{\s_\A,\s_\B}: Q^+ \rightarrow \Dist(Q)$ the function giving the
		probability distribution over the next state of the arena given the sequence
		of states already seen. That is, for all finite path $\pi = \pi_0 \ldots \pi_n
		\in Q^+$ and $q \in Q$, we have:
		\begin{displaymath}
		\mathbb{P}^{\s_\A,\s_\B}(\pi)[q] = \prob{\pi_n,q}{}{}(\s_\A(\pi),\s_\B(\pi))
		\end{displaymath}
		
		Then, the probability of occurrence of a finite path $\pi = \pi_0 \cdots \pi_n
		\in Q^+$ from a state $q_0 \in Q$ with the pair of strategies $(\s_\A,\s_\B)$ is
		equal to $\prob{\Aconc,q_0}{\s_\A,\s_\B}(\pi) = \Pi_{i = 0}^{n-1}
		\mathbb{P}^{\s_\A,\s_\B}(\pi_{\leq i})[\pi_{i+1}]$ if $\pi_0 = q_0$ and $0$
		otherwise. 
		The probability of a cylinder set $\cyl(\pi)$ is
		$\prob{\Aconc,q_0}{\s_\A,\s_\B}[\cyl(\pi)] = \prob{\s_\A,\s_\B}{}(\pi)$ for any
		finite path $\pi \in Q^*$. This induces the probability of any Borel set in the
		usual way, 
		we denote by $\prob{\Aconc,q_0}{\s_\A,\s_\B}: \Borel(Q) \rightarrow [0,1]$ the
		corresponding probability measure. Note that the set of infinite paths in
		$Q^\omega$ whose sequence of colors correspond to a Borel set $W \subseteq
		\colSet^\omega$ is also a Borel set (as the preimage of a Borel set by a
		continuous function).
	\end{definition}
	
	Values of strategies and of the game follow and	are defined below.
	\begin{definition}[Value of strategies and of the game]
		Let $\mathcal{G} = \Games{\Aconc}{W}$ be a PI concurrent
		game 
		and consider a strategy $\s_\A \in \SetStrat{\Aconc}{\A}$
		for Player $\A$. The function
		$\MarVal{\G}[\s_\A]: Q \rightarrow [0,1]$ giving the value
		of the strategy $\s_\A$ is such that, for all $q_0 \in Q$,
		we have
		$\MarVal{\G}[\s_\A](q_0) := \inf_{\s_\B \in \SetStrat{\Aconc}{\B}}
		\prob{\Aconc,q_0}{\s_\A,\s_\B}[W]$. The function
		$\MarVal{\G}[\A]: Q \rightarrow [0,1]$ giving the value for
		Player $\A$: is such that, for all $q_0 \in Q$, we have
		$\MarVal{\G}[\A](q_0) := \sup_{\s_\A \in
			\SetStrat{\Aconc}{\A}} \MarVal{\G}[\s_\A](q_0)$. The
		function $\MarVal{\G}[\B]: Q \rightarrow [0,1]$ giving the
		value of the game for Player $\B$ is defined similarly by
		reversing the supremum and infimum.
		
		By Martin's result on the determinacy of Blackwell games \cite{martin98}, for
		all concurrent games $\G = \Games{\Aconc}{W}$, the value functions for both
		Players are equal, this defines the value function $\MarVal{\G}: Q \rightarrow
		[0,1]$ of the game: $\MarVal{\G} := \MarVal{\G}[\A] = \MarVal{\G}[\B]$.
		\label{def:determinacy}
	\end{definition}
	
	We define value areas: subsets of states 
	whose values are the same.
	\begin{definition}[Value area]
		In a PI concurrent game $\mathcal{G}
		$, 
		$V_\mathcal{G}$ refers to the set of values appearing in the game
		: $V_\mathcal{G} := \{ \MarVal{\G}[q] \mid q \in Q \}$. Furthermore, for all %
		$u \in V_\mathcal{G}$, $Q_u \subseteq Q$ refers to the set of states whose
		values are $u$ w.r.t. $\MarVal{\G}$: $Q_u := \{ q \in Q \mid \MarVal{\G}(q) = u
		\}$.
		\label{def:slice_of_game}
	\end{definition}
	
	In concurrent games, game forms appear at each state and describe the
	interactions of the players at that state. Furthermore, the valuation mapping
	each state to its value in the game can be lifted, via a convex combination,
	into a valuation of the Nature states. This, in turn, induces a natural way to
	define the game in normal form appearing at each state.
	\begin{definition}[Local interactions, Lifting valuations
		]
		In a PI concurrent game $\G$ where the valuation $\MarVal{\G}: Q \rightarrow
		[0,1]$ gives the values of the game, the lift $\LiftVal{\G}: \distribSet
		\rightarrow [0,1]$ is such that, for all 
		$d \in \distribSet$, we have $\LiftVal{\G}(d) := \sum_{q \in Q} \MarVal{\G}(q)
		\cdot \distribFunc(d)(q)$ (recall that $\distribFunc: \distribSet \rightarrow
		\Dist(Q)$ is the distribution function).
		
		Let $q \in Q$. The \emph{local interaction} at state $q$ is the game form
		$\formNF_q = \langle \setA_q,\setB_q,\distribSet,\delta(q,\cdot,\cdot)
		\rangle$. 
		The game in normal form at state $q$ is then $\formN_q :=
		\Games{\formNF_q}{\LiftVal{\G}}$.
	\end{definition}

	Note that, the values of the game in normal form $\formN_q$ 
	and of the state $q$ are equal.
	\begin{proposition}
		\label{prop:gnf}
		In a PI concurrent game $\G$, for all states $q \in Q$, we have
		$\MarVal{\G}(q) = \outM_{\formN_q}$.
	\end{proposition}
	
	
	\subsection{More on strategies}
	\label{subsec:definition_strategies}
	In this subsection, we define several kinds of strategies we will be
	interested in later on. Let us fix a PI concurrent game $\G$ for the rest of this section.
	%
	First, we consider optimal strategies, i.e. strategies realizing the value of
	the game. We also consider positively-optimal strategies, i.e. strategies whose values are positive from all states where it is possible. This is defined formally below in Definition~\ref{def:positively_optimal_strat}.
	\begin{definition}[(Positively-) optimal strategies]
		\label{def:positively_optimal_strat}
		A Player $\A$ strategy $\s_\A \in \SetStrat{\Aconc}{\A}$ is (resp.
		\emph{positively}-)\emph{optimal} from a state $q \in Q$ if $\MarVal{\G}(q) =
		\MarVal{\G}[\s_\A](q)$ (resp. if $\MarVal{\G}(q) > 0 \Rightarrow
		\MarVal{\G}[\s_\A](q) > 0$). It is (resp. positively-)optimal if this 
		holds from all states $q \in Q$.
	\end{definition}
	Note that the definition of optimal strategies we consider is sometimes
	referred to as uniform optimality, as it holds from every state of the game. 
	However, it does not say anything about what happens once some sequence of
	states have been seen. We would like now to define a notion of strategy that is
	optimal from any point that can occur after any finite sequence of states has
	been seen. This correspond to subgame optimal strategies. To define them, we
	need to introduce 
	the notion of residual strategy.
	
	\begin{definition}[Residual and Subgame Optimal Strategies]
		Consider a Player $\A$ strategy $\s_\A$. For all finite sequences $\rho \in Q^+$, the \emph{residual strategy}
		$\s_\A^\rho$ of a the strategy $\s_\A$ 
		is the strategy $\s_\A^\rho: Q^+ \rightarrow \Dist(A)$ such that, for all $\pi \in Q^+$, we have $\s_\A^\rho(\pi) := \s_\A(\rho \cdot \pi)$.
		
		The Player $\A$ strategy $\s_\A$ is \emph{subgame optimal} if, for all $\rho =	\rho' \cdot q \in Q^+$, the residual strategy $\s_\A^\rho$ is optimal from $q$, i.e. $\MarVal{\G}[\s_\A^\rho](q) = \MarVal{\G}(q)$.
	\end{definition}
	
	Note that, in particular, subgame optimal strategies are 
	optimal strategies. 
	When such strategies do exist, we want them to be as simple as possible, for
	instance we want them to be positional
	, that is that they only depend on the current state of the game.
	
	Furthermore, once a Player $\A$ strategy is fixed we obtain an (infinite)
	MDP. In such a context, $\varepsilon$-optimal strategies can
	be chosen among deterministic strategies (see for instance the
	explanation in~\cite[Thm. 1]{DBLP:journals/iandc/Chatterjee0GH15}) where deterministic strategies are such that, after any finite sequence of states, a specific action is played with probability 1. 
	Both of these notions are defined formally below in Definition~\ref{def:positional_deterministic_strategies}. 
	\begin{definition}[Positional, Deterministic strategies]
		\label{def:positional_deterministic_strategies}
		A Player $\A$ strategy $\s_\A$ is \emph{positional} if, for all states $q \in Q$ and 
		paths $\rho \in Q^+$ we have $\s_\A(\rho \cdot q) = \s_\A(q)$. 
		
		A Player $\B$ strategy $\s_\B$ is \emph{deterministic} if, for all finite sequences $\rho \cdot q \in Q^+$, there exists $b \in B_q$ such that $\s_\B(\rho \cdot q)(b) = 1$.
	\end{definition}
	
	
	\section{Necessary and sufficient condition for subgame optimality}
	\label{sec:conditions_subgame_optimal}
	In this section, we present a necessary and sufficient pair of
	conditions for a Player $\A$ strategy to be subgame optimal, formally stated in
	Theorem~\ref{thm:subgame_optimal_arbitrary_strategy}. 
	The arguments given 
	here are somewhat similar to the ones given in Section 4 of
	\cite{BBSArXivICALP}, which deals with the same question restricted to
	positional strategies.
	
	The first 
	condition is local: it specifies how a strategy behaves in the
	games in normal form at each local interaction of the game. As mentioned in
	Proposition~\ref{prop:gnf}, at each state $q$, the value of the game in normal
	form $\formN_q$ is equal to the value of the state $q$ (given by the valuation
	$\MarVal{\G} \in [0,1]^Q$). This suggests that, for all finite sequences of
	states $\rho \in Q^+$ 
	ending at that state $q$, the $\GF$-strategy $\s_\A(\rho)$ 
	needs to be optimal in the game in normal form $\formN_q$ for the residual
	strategy $\s_\A^\rho$ to be optimal from $q$. Strategies with such a property
	are called locally optimal.
	This is 
	a necessary condition for subgame optimality. (However,
	it is neither a necessary nor a sufficient
	condition for optimality, as argued in Section~\ref{sec:transfer}
	.)
	\begin{definition}[Locally optimal strategies]
		Consider a PI concurrent game $\G$. 
		%
		A Player $\A$ strategy $\s_\A$ is locally optimal if, for all $\rho = \rho'
		\cdot q \in Q^+$, the $\GF$-strategy $\s_\A(\rho)$ is optimal in the game in
		normal form $\formN_q$. That is -- recalling that $\LiftVal{\G} \in
		[0,1]^\distribSet$ lifts the valuation $\MarVal{\G} \in [0,1]^Q$ to the Nature
		states -- for all $b \in B_q$: 
		\begin{displaymath}
		\MarVal{\G}(q) \leq \outM_{\formN_q}(\s_\A(\rho),b) = \sum_{a \in A_q} \s_\A(\rho)(a) \cdot \LiftVal{\G} \circ
		\delta(q,a,b)
		\end{displaymath}
		\label{def:locally_optimal}
	\end{definition}
	
	\begin{lemma}[Proof in
		Appendix~\ref{proof:prop_sub_game_opt_implies_locally_opt}]
		\label{prop:sub_game_opt_implies_locally_opt}
		In a PI concurrent game,
		subgame optimal strategies are locally optimal.
	\end{lemma}
	Note that this was already shown for positional strategies in
	\cite{BBSArXivICALP}.
	
	Local optimality does not ensure subgame optimality in general. However, it
	does ensure that, for all Player $\B$ deterministic strategies, the game
	almost-surely eventually settles in a value area, 
	i.e. in some $Q_u$ for some $u \in V_\G$.
	\begin{lemma}[Proof in Appendix~\ref{proof:lem_locally_opt_implies_same_value}]
		\label{lem:locally_opt_implies_same_value}
		Consider a PI concurrent game $\G
		$ 
		and a Player $\A$ locally optimal strategy $\s_\A$. For all Player $\B$
		deterministic strategies, almost surely the states seen infinitely often have
		the same value
		. That is: 
		\begin{displaymath}
			\mathbb{P}^{\s_\A,\s_\B}[\bigcup_{u \in V_{\G}} Q^* \cdot
			(Q_u)^\omega] = 1
		\end{displaymath}
	\end{lemma}
	\begin{proof}[Sketch]
		First, if a state of value 1 is reached (i.e. a state in $Q_1$), then all
		states that can be seen with positive probability have value 1 (i.e. are in
		$Q_1$), since the strategy $\s_\A$ is locally optimal. Let now $u \in V_{\G}$ be
		the highest value in $V_{\G}$ that is not 1 and consider the set of infinite
		paths such that the set $Q_u$ is seen infinitely often but the game does not
		settle in it, i.e. the set $(Q^* \cdot (Q \setminus Q_u))^\omega \cap (Q^* \cdot
		Q_u)^\omega \subseteq Q^\omega$. Since the strategy $\s_\A$ is locally optimal
		(and since $V_{\G}$ is finite), one can show that there is a positive
		probability $p > 0$ such that, the conditional probability of reaching $Q_1$
		knowing that $Q_u$ is left is at least $p$. 
		Hence, if $Q_u$ is left infinitely often, almost-surely the set $Q_1$ is seen
		(and never left). It follows that the probability of the event $(Q^* \cdot (Q
		\setminus Q_u))^\omega \cap (Q^* \cdot Q_u)^\omega$ is 0. This implies that,
		almost-surely, if the set $Q_u$ is seen infinitely often, then at some point it
		is never left. The same arguments can then be used with the highest value in
		$V_{\G}$ that is less than $u$, etc. Overall, we obtain that, for all $u \in
		V_{\G}$, if a set $Q_u$ is seen infinitely often, it is eventually never left
		almost-surely. 
	\end{proof}
	
	Local optimality ensures that, at each step, the 
	expected values of the states reached 
	does not worsen (and may even improve if Player $\B$ does not play optimally). 
	By propagating this property, we obtain that, given a Player $\A$ locally
	optimal strategy and a Player $\B$ deterministic strategy, the convex
	combination of the values $u$ in $V_\G$ weighted by the probability of settling
	in the 
	value area $Q_u$, from a state $q$ is at least equal to its 
	value $\MarVal{\G}(q)$. 
	This is stated in Lemma~\ref{lem:locally_opt_implies_convex_comb} below.
	\begin{lemma}[Proof in
		Appendix~\ref{proof:lem_locally_opt_implies_convex_comb_ok}]
		\label{lem:locally_opt_implies_convex_comb}
		For a PI concurrent game $\G
		$, 
		a Player $\A$ locally optimal strategy $\s_\A$, a Player $\B$ deterministic
		strategy $\s_\B$ 
		and a state $q \in Q$: 
		\begin{displaymath}
			\MarVal{\G}(q) \leq \sum_{u \in V_{\G}} u \cdot
			\mathbb{P}^{\s_\A,\s_\B}_q[Q^* \cdot (Q_u)^\omega]
		\end{displaymath}
	\end{lemma}
	Note that if Player $\B$ plays subgame optimally, 
	then this inequality is an equality.
	\begin{proof}[Sketch]
		First, let us denote $\mathbb{P}^{\s_\A,\s_\B}_q$ by $\mathbb{P}$. 
		It can be shown by induction that, for all $i \in \N^*$, we have the property
		$\mathcal{P}(i): \MarVal{\G}(q) \leq \sum_{\pi \cdot q' \in q \cdot Q^i}
		\MarVal{\G}(q') \cdot \mathbb{P}(\pi \cdot q') =  \sum_{u \in V_\G \setminus \{
			0 \}} u \cdot \mathbb{P}[q \cdot Q^{i-1} \cdot Q_u]$. Furthermore, since by
		Lemma~\ref{lem:locally_opt_implies_same_value}, the game almost-surely settles
		in a value area, it can be shown that for $n$ large enough, the probability of
		being in $Q_u$ after $n$ steps (i.e. $\mathbb{P}[q \cdot Q^{n-1} \cdot Q_u]$) is
		arbitrarily close to the probability of eventually settling in $Q_u$ (i.e.
		$\mathbb{P}
		[Q^* \cdot (Q_u)^\omega]$). We can then apply $\mathcal{P}(n)$ to obtain the
		desired inequality.
	\end{proof}
	
	Recall that we are considering a pair of conditions to
	characterize that a strategy is subgame optimal. The first
	condition is local optimality. To summarize, we have seen that	
	the fact that a strategy is locally optimal ensures that, from
	any state $q$, the expected values of the value areas where the game settles 
	is at least $\MarVal{\G}(q)$. However, local
	optimality does not ensure anything as to the probability of
	$W$ given that the game settles in a specific value area, as
	witnessed in
	Appendix~\ref{appen:exemple_opt_loc_opt_not_sub_game_opt}. This
	is where the second condition comes into play. For the
	explanations regarding this condition, we will need
	Lemma~\ref{lem:chaterjee_value_0_1} below: a
	consequence of Levy's 0-1 Law.
	\begin{lemma}[See Appendix~\ref{appen:levy}]
		\label{lem:chaterjee_value_0_1}
		Let $\mathcal{M}$ be a countable Markov chain with a PI objective. If there is
		a $q \in Q$ such that $\MarVal{\mathcal{M}}(q) < 1$, then $\inf_{q' \in Q}
		\MarVal{\mathcal{M}}(q') = 0$.
	\end{lemma}
	Consider now %
	%
	a Player $\A$ subgame optimal strategy $\s_\A$ and a Player $\B$ deterministic
	strategy. Let us consider what happens if the game eventually settles in $Q_u$
	for some $u \in V_\G \setminus \{ 0 \}$. Assume towards a contradiction that
	there is a finite path after which the probability of $W$ given that the play
	eventually settles in $Q_u$ is less than 1. Then, there is a continuation of
	this path ending in $Q_u$ for which this probability of $W$ is less than $u$.
	Indeed, it was shown that, for a PI objective, in a countable Markov chain
	(which is what we obtain once strategies for both players are fixed), if there
	is a state with a value less than 1, then the infimum of the values in the
	Markov chain is 0 (this is what is stated in
	Lemma~\ref{lem:chaterjee_value_0_1}). Following our above
	towards-a-contradiction-assumption, there would be a finite path from which the
	Player $\A$ strategy $\s_\A$ is not optimal. This is in contradiction with the
	fact that it is subgame optimal. Hence, a second necessary condition -- 
	in addition to the local optimality assumption -- for subgame optimality is:
	from all finite paths, for all Player $\B$ deterministic strategies, for all
	positive values $u \in V_{\G} \setminus \{ 0 \}$, the probability of $W$ and
	eventually settling in $Q_u$ is equal to the probability of eventually settling
	in $Q_u$. We obtain the theorem below.

	\begin{theorem}[Proof in Appendix~\ref{proof:lem_optimal_arbitrary_strategy}]
		Consider a concurrent game $\G$ with a PI objective $W$ 
		and a Player $\A$ strategy $\s_\A \in \SetStrat{\Aconc}{\A}$. The strategy
		$\s_\A$ is subgame optimal if and only if:
		\begin{itemize}
			\item it is locally optimal;
			\item for all $\rho \in Q^+$, for all Player $\B$ deterministic strategies
			$\s_\B$, for all values $u \in V_\G \setminus \{ 0 \}$, we have
			$\mathbb{P}_\rho^{\s_\A^\rho,\s_\B^\rho}[W \cap Q^* \cdot (Q_u)^\omega] =
			\mathbb{P}_\rho^{\s_\A^\rho,\s_\B^\rho}[Q^* \cdot (Q_u)^\omega]$.
		\end{itemize}
		\label{thm:subgame_optimal_arbitrary_strategy}
	\end{theorem}
	\begin{proof}[Sketch]
		Lemma~\ref{prop:sub_game_opt_implies_locally_opt} states that local optimality
		is necessary and we have informally argued above why the second condition is
		also necessary for subgame optimality. As for the fact that they are sufficient
		conditions, this is a direct consequence of
		Lemmas~\ref{lem:locally_opt_implies_same_value}
		and~\ref{lem:locally_opt_implies_convex_comb} and the fact that deterministic
		strategies can achieve the same values as arbitrary strategies in MDPs (which we
		obtain once a Player $\A$ strategy is fixed), as cited in
		Subsection~\ref{subsec:definition_strategies}
		.
	\end{proof}
	
	One may ask what happens in the special case where the strategy $\s_\A$
	considered is positional. 
	As mentioned above, such a characterization was already presented in
	\cite{BBSArXivICALP}\footnote{The proof was only presented for a
		specific class of objectives.
	}. Overall, we obtain a similar result except that the second condition is
	replaced by what happens in the game restricted to the End Components in the
	Markov Decision Process induced by the positional strategy $\s_\A$. 
	
	\section{From subgame almost-surely winning to subgame optimality}
	\label{sec:transfer}
	
	\begin{figure}
		\begin{minipage}[b]{0.48\linewidth}
			\centering
			\includegraphics[scale=1]{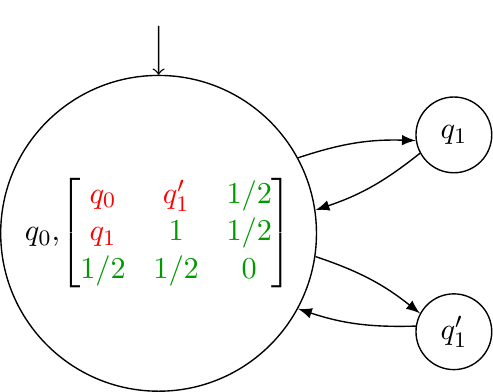}
			\caption{A co-Büchi game.}
			\label{fig:co_buchi_unfair}
		\end{minipage}
		\begin{minipage}[b]{0.48\linewidth}
			\centering
			\includegraphics[scale=1.3]{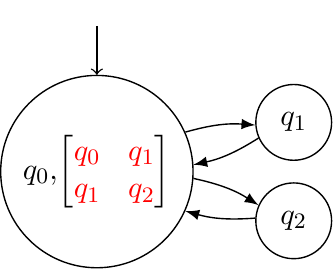}
			\caption{A parity game.}
			\label{fig:parity_unfair}
		\end{minipage}
	\end{figure}
	
	In \cite[Thm. 4.5]{GIH10}, the authors have proved a transfer result in PI
	turn-based games
	: the amount of memory sufficient to play optimally in every state of value 1
	of every game is also sufficient to play optimally in every game. 
	This result does not hold on concurrent games as is. First, although there are
	always optimal strategies in PI turn-based games (as proved in the same paper
	\cite[Thm. 4.3]{GIH10}), there are PI concurrent games without optimal
	strategies. 
	Second, infinite memory may be required to play optimally in co-Büchi
	concurrent games whereas almost-surely  winning strategies can be found among
	positional strategies in a turn-based setting. This can be seen in the game of
	Figure~\ref{fig:co_buchi_unfair} with $\colFunc(q_0) = 0$ and $\colFunc(q_1) =
	\colFunc(q_1') = 1$. The green values in the local interaction at state $q_0$
	are the values of the game if they are reached (the game ends immediately). If a
	green value is not reached, the objective of Player $\A$ is to see only finitely
	often states $q_1$ and $q_1'$. It has already been argued in
	\cite{BBSArXivICALP} that the value of this game is $1/2$ and that there is an
	optimal strategy for Player $\A$ but it requires infinite memory. 
	To play optimally, Player $\A$ must play the top row with probability $1 -
	\varepsilon_k$ and the middle row with probability $\varepsilon_k$ for
	$\varepsilon_k > 0$ that goes (fast) to 0 when $k$ goes to $\infty$ (where $k$
	denotes the number of steps). The $\varepsilon_k$ must be chosen so that, if
	Player $\B$ always plays the left column with probability 1, then the state
	$q_1$ is seen finitely often with probability 1. Furthermore, as soon as the
	state $q_1'$ is visited, Player $\A$ switches to a positional strategy playing
	the bottom row with probability $\varepsilon_k'$ small enough (where $k$ denotes
	the number of steps before the state $q_1'$ was seen) and the two top rows with
	probability $(1 - \varepsilon_k')/2$. 
	
	Hence, the transfer of memory from almost-surely winning to
	optimal does not hold in concurrent games even if it is
	assumed that optimal strategies
	exist
	. However, one can note that although the strategy described
	above is 
	optimal, it is not subgame optimal. Indeed, when the strategy
	switches, the value of the residual strategy is
	$1/2 - \varepsilon_k' <
	1/2$. 
	In fact, there is no subgame optimal strategy in that
	game. Actually, if we assume that, not only optimal but
	subgame optimal strategies exist, then the transfer of memory
	will hold.
	%
	%
	
	The aim of this section is twofold: first, we identify a necessary and
	sufficient condition for the existence of subgame optimal
	strategies\footnote{Note that this is 
		different from what we did in the previous section: there, we established a
		necessary and sufficient condition for a specific strategy to be subgame
		optimal. Here, given a game, we consider necessary and sufficient conditions on
		the game for the existence of a subgame optimal strategy.}. Second, we establish
	the above-mentioned memory transfer that relates the amount of memory to play
	subgame optimally 
	and to be almost-surely winning. Before stating the main theorem of this
	section
	, let us first introduce the definition of positionally subgame almost-surely
	winnable objective, i.e. objectives for which subgame almost-surely winning
	strategies can be found among positional strategies.
	\begin{definition}[Positionally subgame almost-surely winnable objective]
		Consider a PI objective $W \subseteq \colSet^\omega$. It is said to be a
		positionally subgame almost-surely winnable objective (\textsf{PSAW} for short)
		if the following holds: in all concurrent games $\G = \Games{\Aconc}{W}$ where
		there is a subgame almost-surely winning strategy, there is a positional one. 
	\end{definition}
	
	
	
	
	\begin{theorem}
		Consider a non-empty finite set of colors $K$ and a PI objective $\emptyset
		\subsetneq W \subseteq \colSet^\omega$. Consider a concurrent game $\G$ with
		objective $W$. Then, the three following assertions are equivalent:
		\begin{itemize}
			\item[a.] there exists a subgame optimal strategy;
			\item[b.] there exists an 
			optimal strategy that is locally optimal;
			\item[c.] there exists a 
			positively-optimal strategy that is locally optimal.
		\end{itemize}
		Furthermore, if this holds and if the objective $W$ is \textsf{PSAW}, then
		there exists a subgame optimal positional strategy.
		\label{thm:transfer_memory}
	\end{theorem}
	First, note that the equivalence is stated in terms of existence of strategies,
	not on the strategies themselves. In particular, any subgame optimal strategy is
	both 
	optimal and locally optimal, however, an 
	optimal strategy that is locally optimal is not necessarily a subgame optimal
	strategy. An example is provided in
	Appendix~\ref{appen:exemple_opt_loc_opt_not_sub_game_opt}. Second, it is
	straightforward that point \textit{a} implies point \textit{b} (from
	Theorem~\ref{thm:subgame_optimal_arbitrary_strategy}) and that point \textit{b}
	implies point \textit{c} (by definition of positively-optimal strategies). In
	the remainder of this section, we explain informally the constructions 
	leading to the proof of this theorem, i.e. to the proof that point \textit{c}
	implies point \textit{a}. The transfer of memory 
	is a direct consequence of the way this theorem is proven. The
	full proof is given in Appendix
	~\ref{proof:thm_transfer_almost_sure_optimal}. We fix a PI concurrent game
	$\G = \Games{\Aconc}{W}$ for the rest of the section.
	
	The idea is as follows. As stated in
	Theorem~\ref{thm:subgame_optimal_arbitrary_strategy}, subgame optimal strategies
	are locally optimal and win the game almost-surely if the game settles in a
	value area $Q_u$ for some positive $u \in V_{\G} \setminus \{ 0 \}$. Our idea is
	therefore to consider subgame almost-surely winning strategies in the derived
	game $\G_u$: a \textquotedblleft restriction\textquotedblright \ of the game
	$\G$ to $Q_u$ (more details will be given later). We can 
	then glue together these subgame almost-surely winning strategies -- defined
	for all $u \in V_{\G} \setminus \{ 0 \}$ -- 
	into a subgame optimal strategy.
	%
	However, there are some issues:
	\begin{enumerate}
		\item the state values in the game $\G_u$ should be all equal to 1;
		\item furthermore, there must exist a subgame almost-surely winning strategy
		in $\G_u$;
		\item this subgame almost-surely winning strategy in $\G_u$ should be locally
		optimal when considered in the whole game $\G$.
	\end{enumerate}
	Note that the method we use here is different from what the authors of
	\cite{GIH10} did to prove the transfer of memory in turn-based
	games
	. 
	
	Let us first deal with issue {\bfseries\sffamily
		\color{darkgray} 3}. One can ensure that the almost-surely winning strategies
	in the game $\G_u$ are all locally optimal in $\G$ by properly defining the game
	$\G_u$. More specifically, 
	this is done by enforcing that the only Player $\A$ possible strategies in
	$\G_u$ are locally optimal in the game $\G$. To do so, we construct the game
	$\G_u$ whose state space is $Q_u$ 
	(plus gadget states) but whose set of actions $A_{\formN_q}$, at a state $q \in
	Q_u$, is such that the set of strategies $\Dist(A_{\formN_q})$ corresponds
	exactly to the set of optimal strategies in the original game in normal form
	$\formN_q$, while keeping the set of actions $A_{\formN_q}$ for Player $\A$
	finite. This is possible thanks to
	Proposition~\ref{prop:finite_set_optimal_start_gnf} below: in every game in
	normal form $\formN_q$ at state $q \in Q_u$, there exists a finite set
	$A_{\formN_q}$ of optimal strategies such that the optimal strategies in
	$\formN_q$ are exactly the convex combinations of strategies in $A_{\formN_q}$.
	This is a well known result, argued for instance in \cite{shapley1950basic}.	
	\begin{proposition}
		\label{prop:finite_set_optimal_start_gnf}
		Consider a game in normal form $\formN = \langle \mathsf{Act}_\A,\mathsf{Act}_\B,[0,1],\delta \rangle$
		with $|\mathsf{Act}_\A| = n$ and $|\mathsf{Act}_\B| = k$. There exists a set $A_{\formN} \subseteq
		\Opt_\A(\formN)$ of optimal strategies such that $|A_{\formN}| \leq n+k$ and
		$\Dist(A_{\formN}) = \Opt_\A(\formN)$.
	\end{proposition}
	\begin{proof}[Sketch]
		One can write a system of $n+k$ inequalities (with some
		additional equalities) whose set of solutions is exactly the
		set of optimal $\GF$-strategies $\Opt_\A(\formN)$.
		The result then follows from standard system of inequalities
		arguments as the space of solutions is in fact a polytope
		with at most $n+k$ vertices.
	\end{proof}
	
	We illustrate this construction: a part of a concurrent game
	is depicted in Figure \ref{fig:local_optimal_necessary} and
	the change of the interaction of the players at state $q_0$ is
	depicted in Figures~\ref{fig:FormOriginal},~\ref{fig:FormOriginalValued},~\ref{fig:FormOriginalValuedOnlyOpt} and~\ref{fig:FormOnlyOpt}.
	\begin{figure}
		\begin{minipage}[b]{0.23\linewidth}
			\hspace*{-2.4cm}
			\vspace*{-0.1cm}
			\centering
			\includegraphics[scale=0.9]{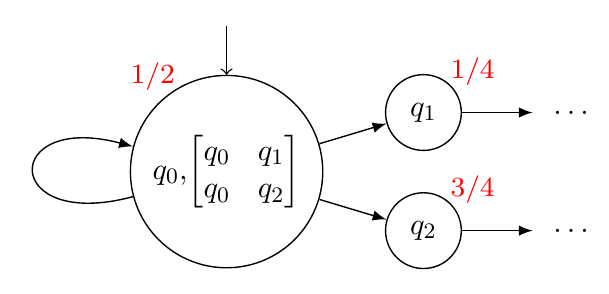}
			\caption{A concurrent game with $A_{q_0} = \{ a_1,a_2 \}$.}
			\label{fig:local_optimal_necessary}
		\end{minipage}
		\hspace*{0.1cm}
		\begin{minipage}[b]{0.15\linewidth}
			\hspace*{-1cm}
			\centering
			\includegraphics[scale=1.2]{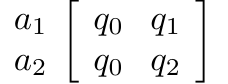}
			\caption{The local interaction $\formNF_{q_0}$ at state $q_0$.}
			\label{fig:FormOriginal}
		\end{minipage}
		\hspace*{0.1cm}
		\begin{minipage}[b]{0.15\linewidth}
			\hspace*{-0.5cm}
			\vspace*{-0.1cm}
			\centering
			\includegraphics[scale=1.1]{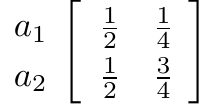}
			\caption{The game in normal form $\formN_{q_0}$
				.}
			\label{fig:FormOriginalValued}
		\end{minipage}
		\hspace*{0.1cm}
		\begin{minipage}[b]{0.19\linewidth}
			\hspace*{-0.3cm}
			\vspace*{-0.15cm}
			\centering
			\includegraphics[scale=1.1]{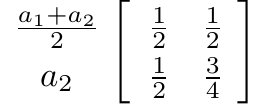}
			\caption{The game $\mathcal{F}^{\mathsf{opt},\mathsf{nf}}_{q_0}$ with only
				optimal strategies.}
			\label{fig:FormOriginalValuedOnlyOpt}
		\end{minipage}
		\hspace*{0.1cm}
		\begin{minipage}[b]{0.19\linewidth}
			\hspace*{-0.2cm}
			\vspace*{-0.15cm}
			\centering
			\includegraphics[scale=1.1]{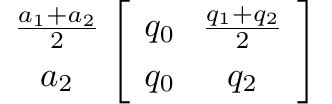}
			\caption{The game form $\mathcal{F}^{\mathsf{opt}}_{q_0}$ with only optimal
				strategies.}
			\label{fig:FormOnlyOpt}
		\end{minipage}
	\end{figure}
	
	The game $\G_u$ has the same objective $W$ as the game 
	$\G$.
	Since we want all the states to have value 1 in $\G_u$ (recall issue
	{\bfseries\sffamily
		\color{darkgray} 1}),
	we will build the game $\G_u$ such that any edge leading to a state not in
	$Q_u$ in $\G$ now leads to a PI concurrent game $\G_W$ (with the same objective
	$W$) where all states have value~1. 
	The game $\G_W$ is (for instance) a clique with all colors in $\colSet$ where
	Player $\A$ plays alone. The formal definitions of the game $\G_W =
	\Games{\Aconc_W}{W}$ and of the game $\G_u$ can be found in
	Appendix~\ref{subsec:game_value_one_everywhere}. 
	
	An illustration of this construction can be found in
	Figures~\ref{fig:decompositionGame1} and~\ref{fig:decompositionGame2}. The blue
	dotted arrows are the ones that need to be redirected when the game is changed.
	\begin{figure}
		\begin{minipage}[b]{0.48\linewidth}
			\hspace*{-1cm}
			\centering
			\includegraphics[scale=0.45]{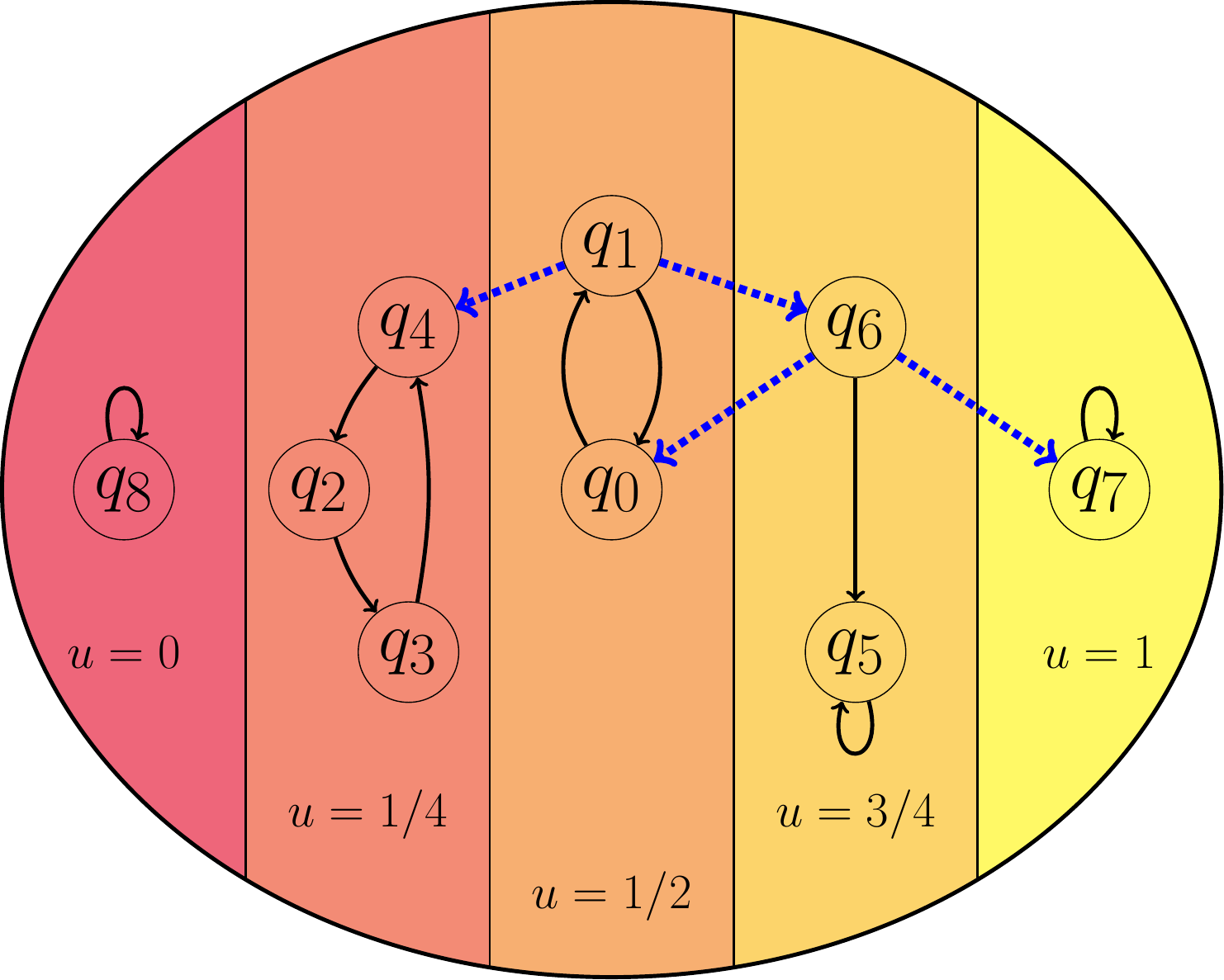}
			\caption{The depiction of a PI concurrent game with its value areas.}
			\label{fig:decompositionGame1}
		\end{minipage}
		\begin{minipage}[b]{0.48\linewidth}
			\centering
			\includegraphics[scale=0.45]{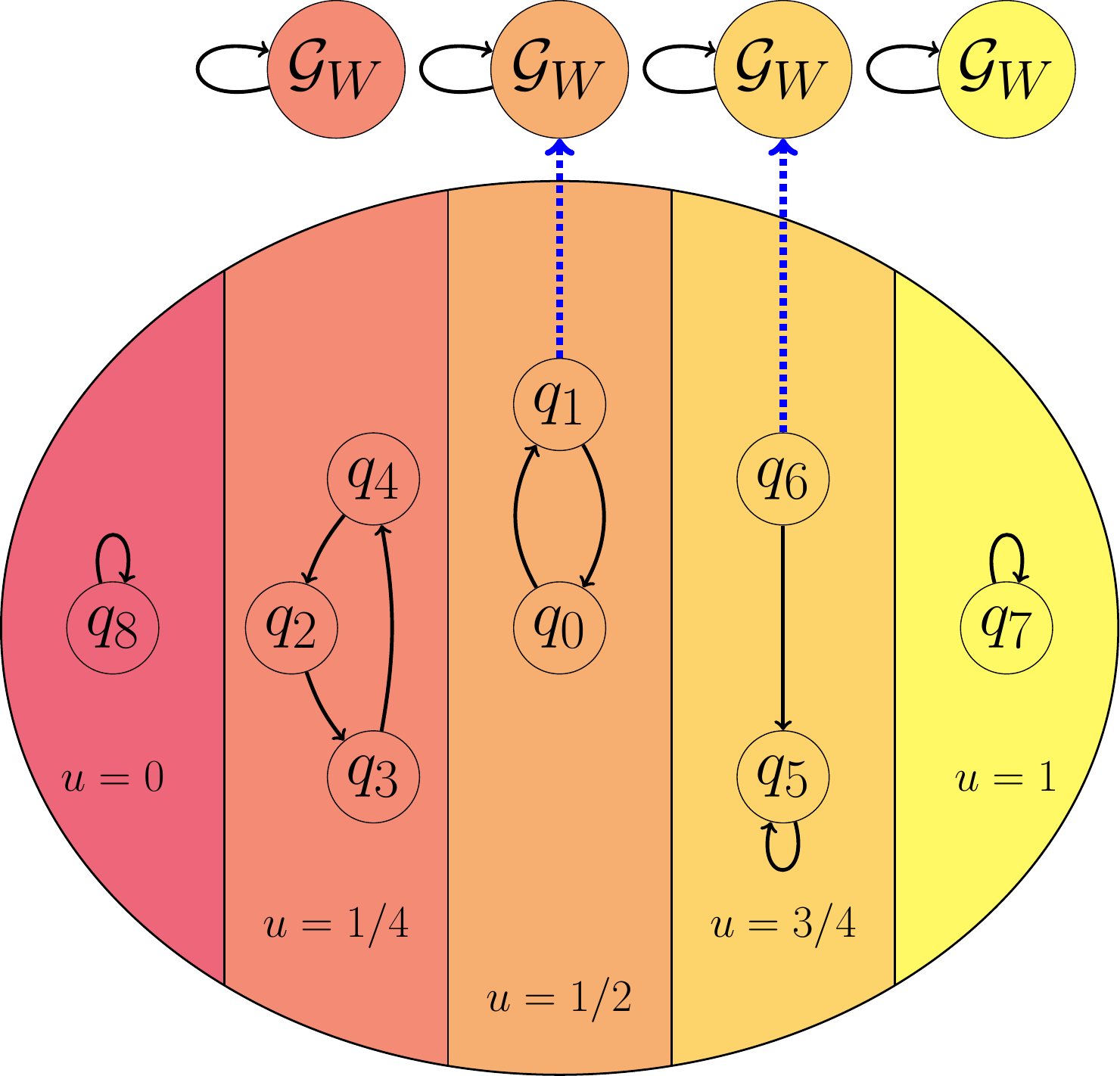}
			\caption{The PI concurrent game after the modifications described above.}
			\label{fig:decompositionGame2}
		\end{minipage}
	\end{figure}
	With such a definition, we have made some progress w.r.t. the issue
	{\bfseries\sffamily
		\color{darkgray} 1} cited previously (regarding the values being equal to 1):
	the values of all states of the game $\G_u$ are positive (for positive $u$).
	\begin{lemma}[Proof in Appendix~\ref{proof:lem_value_one_Gu}]
		Consider the game $\G_u$ for some positive $u \in V_{\G} \setminus \{ 0 \}$
		and assume that, in $\G$, there exists a 
		positively-optimal strategy that is locally optimal. 
		Then, for all states $q$ in $\G_u$, the value of the state $q$ in $\G_u$ is
		positive: $\MarVal{\G_u}(q) > 0$.
		\label{lem:value_one_Gu}
	\end{lemma}
	\begin{proof}[Sketch]
		Consider a state $q \in Q_u$ and a Player $\A$ locally optimal strategy
		$\s_\A$ in $\G$ that is positively-optimal from $q$
		. Then, the strategy $\s_\A$ (restricted to $Q_u^+$) can be seen as a strategy
		in $\G_u$ (it has to be defined in $\G_W$, but this can done straightforwardly).
		Note that this is only possible because the strategy $\s_\A$ is locally optimal
		(due to the definition of $\G_u$). 
		For a Player $\B$ strategy $\s_\B$ in $\G_u$, consider what happens with
		strategies $\s_\A$ and $\s_\B$ in both games $\G_u$ and $\G$. Either the game
		stays indefinitely in $Q_u$, and what happens in $\G_u$ and $\G$ is identical.
		Or it eventually leaves $Q_u$, leading to states of value 1 in $\G_u$. Hence,
		the value of the game $\G_u$ from $q$ with strategies $\s_\A$ and $\s_\B$ is at
		least the value of the game $\G$ from $q$ with the same strategies. Thus, the
		value of the state $q$ is positive in $\G_u$.
	\end{proof}
	
	As it turns out, Lemma~\ref{lem:value_one_Gu} suffices 
	to deal with both 
	issues {\bfseries\sffamily \color{darkgray} 1} and
	{\bfseries\sffamily \color{darkgray} 2} at the same
	time. Indeed, as stated in
	Theorem~\ref{thm:all_not_zero_implies_almost_sure} below, it
	is a general result that in a PI concurrent game, if all
	states have positive values, then all states have value 1 and
	there is a subgame almost-surely winning strategy.
	\begin{theorem}[Proof in
		Appendix~\ref{proof:thm_all_not_zero_implies_almost_sure}]
		\label{thm:all_not_zero_implies_almost_sure}
		Consider a PI concurrent game $\G$ and assume that all state values are
		greater than or equal to 
		$c > 0$, i.e. for all $q \in Q$, $\MarVal{\G}(q) \geq c$. Then, there is a
		subgame almost-surely winning strategy in $\G$.
	\end{theorem}
	\begin{remark}
		This theorem can be seen as a strengthening of Theorem
		1 from \cite{CH07}. Indeed, this Theorem 1 states that
		if all states have positive values, then they all have
		value 1 (this is then generalized to games with
		countably-many states). Theorem~\ref{thm:all_not_zero_implies_almost_sure}
		is stronger since it ensures the existence of
		(subgame) almost-surely winning strategies. Although a
		detailed proof is provided in
		Appendix~\ref{proof:thm_all_not_zero_implies_almost_sure},
		note that this theorem was already stated and proven
		in \cite{GIH10} in the context of PI turn-based
		games. Nevertheless their arguments could have been
		used \textit{verbatim} for concurrent games as well.
		In the Appendix, we give a proof using the same
		construction (namely, reset strategies) but we argue
		differently 
		why the construction proves the theorem.
	\end{remark}
	
	%
	%
	We can now glue together pieces of strategies $\s_\A^u$ defined in all games
	$\G_u$ into a single strategy $\s_\A[(\s_\A^u)_{u \in V_{\G} \setminus \{ 0
		\}}]$. Informally, the glued strategy mimics the strategy on $Q_u^+$ and
	switches strategy when a value area is left and another one is reached.
	\begin{definition}[Gluing strategies]
		Consider a PI concurrent game $\G$ 
		and for all values $u \in V_{\G} \setminus \{ 0 \}$, a strategy $\s_\A^u$ in
		the game $\G_u$. Then, we glue these strategies into the strategy
		$\s_\A[(\s_\A^u)_{u \in V_{\G} \setminus \{ 0 \}}]: Q^+ \rightarrow \Dist(A)$
		simply written $\s_\A$ such that, for all $\rho$ ending at state $q \in Q$:
		\begin{displaymath}
		\s_\A(\rho) := 
		\begin{cases}
		\s_\A^u(\pi) & \text{ if } u = \MarVal{\G}(q) > 0 \text{ for } \pi \text{ the
			longest suffix of } \rho \text{ in } Q_u^+ \\ 
		\text{is arbitrary } & \text{ if } \MarVal{\G}(q) = 0\\
		\end{cases}
		\end{displaymath}
		\label{def:glued_strat}
	\end{definition}
	As stated in Lemma~\ref{prop:glued_strategy} below, the construction described
	in Definition~\ref{def:glued_strat} transfers almost-surely winning strategies
	in $\G_u$ into a subgame optimal strategy in $\G$.
	\begin{lemma}[Proof in Appendix~\ref{proof:prop_glued_strategy}]
		\label{prop:glued_strategy}
		For all $u \in V_{\G} \setminus \{ 0 \}$, let $\s_\A^u$ be a subgame
		almost-surely winning strategy in $\G_u$. 
		The glued strategy $\s_\A[(\s_\A^u)_{u \in V_\G \setminus \{ 0 \}}]$, denoted
		$\s_\A$, is subgame optimal in $\G$.
	\end{lemma}
	\begin{proof}[Sketch]
		We apply Theorem~\ref{thm:subgame_optimal_arbitrary_strategy}. First, the
		strategy $\s_\A$ is locally optimal in all $Q_u$ for $u > 0$ by the strategy
		restriction done to define the game $\G_u$ (only optimal strategies are
		considered at each game in normal form $\formN_q$ at states $q \in Q_u$
		). Furthermore, any strategy is optimal in a game in normal form of value 0
		(which is the case of the game in normal forms of states in $Q_0$). 
		Second, if the game eventually settles in a value area $Q_u$ for some $u > 0$,
		from then on the strategy $\s_\A$ mimics the strategy $\s_\A^u$, which is
		subgame almost-surely winning in $\G_u$. Hence, the probability of $W$ given
		that the game eventually settles in $Q_u$ is 1. This holds for all $u \in V_{\G}
		\setminus \{ 0 \}$, so the second condition of
		Theorem~\ref{thm:subgame_optimal_arbitrary_strategy} holds.
	\end{proof}
	
	We now have all the ingredients to prove Theorem~\ref{thm:transfer_memory}. 
	\begin{proof}[Of Theorem~\ref{thm:transfer_memory}]
		We consider the PI concurrent game $\G$ and assume that there is a 
		positively-optimal strategy that is locally optimal. Then, by
		Lemma~\ref{lem:value_one_Gu}, for all positive values $u \in V_{\G} \setminus \{
		0 \}$, all states in $\G_u$ have positive values. It follows, by
		Theorem~\ref{thm:all_not_zero_implies_almost_sure}, that there exists a subgame
		almost-surely winning strategy in every game $\G_u$ for $u \in V_{\G} \setminus
		\{ 0 \}$. We then obtain a subgame optimal strategy by gluing these strategies
		together, 
		given by Lemma~\ref{prop:glued_strategy}.
		
		The second part of the theorem, dealing with transfer of positionality from
		subgame almost-surely winning to subgame optimal follows from the fact that if
		all strategies $\s_\A^u$ are positional for all $u \in V_{\G} \setminus \{ 0
		\}$, then so is the glued strategy $\s_\A[(\s_\A^u)_{u \in V_{\G} \setminus \{ 0
			\}}]$. 		
	\end{proof}
	
	We now apply the result of Theorem~\ref{thm:transfer_memory}
	to two specific classes of objectives: Büchi and co-Büchi
	objectives. Note that this result is already known for
	Büchi objectives, proven in \cite{BBSArXivICALP}.
	\begin{corollary}
		\label{coro:application_buchi_cobuchi}
		Consider a concurrent game with a Büchi (resp. co-Büchi)
		objective and assume that there is a 
		positively-optimal strategy that is locally optimal. Then there is a subgame
		optimal
		positional strategy.
	\end{corollary}
	Note that it is also possible to prove a memory transfer from
	subgame almost-surely winning to subgame optimal for an
	arbitrary memory skeleton, instead of only positional
	strategies. This adds only a few minor difficulties. This is
	dealt with in Appendix~\ref{appen:extension_finite_memory}.
	

	
	
	
	\textbf{Application to the turn-based setting.}
	The aim of Section~\ref{sec:transfer} was to extend an already existing result
	on turn-based games in the context of concurrent games. This required an
	adaptation of the assumptions. However, it is in fact possible to retrieve the
	original result on turn-based games from Theorem~\ref{thm:transfer_memory} in a
	fairly straightforward manner. It amounts to show that, in all finite turn-based
	games $\G$, for all values $u \in V_{\G} \setminus \{ 0 \}$, there is a locally
	optimal strategy that is positively-optimal from all states in $Q_u$. This is
	done in Appendix~\ref{appen:retrieve_result_turn_based}. 
	
	\section{Finite-choice strategies}
	\label{sec:fair_strategies}

	In this section, we introduce a new kind of strategies, namely finite-choice
	strategies. Let us first motivate why we consider such strategies. Consider
	again the co-Büchi game of Figure~\ref{fig:co_buchi_unfair}. Recall that the
	optimal strategy we described first plays the top row with increasing
	probability and the middle row with decreasing probability and then, once Player
	$\B$ plays the second column, switches to a positional strategy playing the
	bottom row with positive, yet small enough probability. Note that switching
	strategy is essential. Indeed, if Player $\A$ does not switch
	, Player $\B$ could at some point opt for the middle column and see
	indefinitely the state $q_1'$ with very high probability
	. In fact, what happens in that case is rather counter-intuitive: once Player
	$\B$ switches, there is infinitely often a positive probability to reach the
	outcome of value 1. However, the probability to ever reaching this outcome can
	be arbitrarily small, if Player $\B$ waits long enough before playing the middle
	row. This happens because the probability $\varepsilon_k$ to visit that outcome
	goes (fast) to 0 when $k$ goes to $\infty$. In fact, such an optimal strategy
	has \textquotedblleft infinite choice\textquotedblright \ in the sense that it
	may prescribe 
	infinitely many different 
	probability distribution. 
	
	In this section, we consider \emph{finite-choice strategies}, i.e. strategies
	that can use only finitely many 
	$\GF$-strategies at each state.
	
	\begin{definition}[Finite-choice strategy]
		\label{def:finite_choice_strategies}
		Let $\G$ be a concurrent game. A Player $\A$ strategy $\s_\A$ in $\G$ has
		\emph{finite choice} if, for all $q \in Q$, the set $S^{\s_\A}_q := \{ 
		\s_\A(\rho \cdot q) \mid \rho \in Q^+ \} \subseteq \Dist(A_q)$ 
		is finite.
	\end{definition}
	Note that positional (even finite-memory) and deterministic strategies are
	examples of finite-choice strategies.
	
	Interestingly, we can link finite-choice strategies with the existence of
	subgame optimal strategies. In general it does not hold that if there are
	optimal strategies, then there exists subgame optimal strategies (as
	exemplified in the game of Figure~\ref{fig:co_buchi_unfair}). However, in
	Theorem~\ref{thm:fair_strategy_uniformly_optimal} below, we state that if we
	additionally assume that the 
	optimal strategy considered has finite choice, then there is a
	subgame optimal strategy (that has also finite choice).
	\begin{theorem}[Proof in
		Appendix~\ref{proof:thm_fair_strategy_uniformly_optimal}]
		\label{thm:fair_strategy_uniformly_optimal}
		Consider a PI concurrent game $\G$. If there is a finite-choice 
		optimal strategy, 
		then there is a finite-choice subgame optimal
		strategy. 
	\end{theorem}
	\begin{proof}[Sketch]
		Consider such an 
		optimal finite-choice strategy $\s_\A$. In particular, note that there is a
		constant $c > 0$ such that for all $\rho \cdot q \in Q^+$, for all $a \in A_q$
		we have: $\s_\A(\rho \cdot q)(q) > 0 \Rightarrow \s_\A(\rho \cdot q)(q) \geq c$.
		We build a subgame optimal strategy $\s_\A'$ in the following way: for all $\rho
		= \rho' \cdot q \in Q^+$, if the residual strategy $\s_\A^\rho$ is optimal, then
		$\s_\A'(\rho) := \s_\A(\rho)$, otherwise $\s_\A'(\rho) := \s_\A(q)$ (i.e. we
		reset the strategy). Straightforwardly, the strategy $\s_\A'$ has finite choice.
		We want to apply Theorem~\ref{thm:subgame_optimal_arbitrary_strategy} to prove
		that it is subgame optimal. One can see that it is locally optimal (by the
		criterion chosen for resetting the strategy). Consider now some $\rho \in Q^+$
		ending at state $q \in Q$ and another state $q' \in Q$. Assume that the residual
		strategy $\s_\A^\rho$ is optimal but that the residual strategy $\s_\A^{\rho
			\cdot q'}$ is not. Then, similarly to why local optimality is necessary for
		subgame optimality (see
		Proposition~\ref{prop:sub_game_opt_implies_locally_opt}), 
		one can show that any Player $\B$ action $b \in B_q$ leading to $q'$ from $\rho$ with
		positive probability is such that $\MarVal{\G}(q) <
		\outM_{\formN_q}(\s_\A(\rho),b)$. Hence, there is positive probability from
		$\rho$, if Player $\B$ opts for the action $b$, to reach a state of value
		different from $u = \MarVal{\G}(q)$. And if this happens infinitely often, 
		a state of value different from $u$ will be reached
		almost-surely\footnote{This holds because the strategy $\s_\A$ has finite
			choice: the probability to see a state of different value is bounded below by
			the product of $c$ and the smallest positive probability among all Nature
			states.}. In other words, if a value area is never left, almost-surely, the
		strategy $\s_\A'$ only resets finitely often.
		
		Consider now some $\rho \in Q^+$, a Player $\B$ deterministic strategy $\s_\B$
		and a value $u \in V_\G \setminus \{ 0 \}$. From what we argued above, the
		probability of the event $Q^* \cdot (Q_u)^\omega$ (resp. $W \cap Q^* \cdot
		(Q_u)^\omega$) is the same if we intersect it with the fact that the strategy
		$\s_\A'$ only resets finitely often. Furthermore, if the strategy does not reset
		anymore from some point on, and all states have the same value $u > 0$, then it
		follows that the probability of $W$ is 1 (since $W$ is PI). We can then conclude
		by applying Theorem~\ref{thm:subgame_optimal_arbitrary_strategy}.
	\end{proof}
	
	Finite-choice strategies are interesting for another reason. In the previous
	section, we applied the memory transfer from Theorem~\ref{thm:transfer_memory}
	to the Büchi and co-Büchi objectives. We did not apply it to other objectives --
	in particular to the parity objective. Indeed, in general, contrary to the case
	of turn-based games, infinite-memory is necessary to be almost-surely winning in
	parity games. This happens in Figure~\ref{fig:parity_unfair} (already described
	in \cite{AH00}) where the objective of Player $\A$ is to see $q_1$ infinitely
	often, while seeing $q_2$ only finitely often. Let us describe a Player $\A$
	subgame almost-surely winning strategy. 
	The top row is played with probability $1 - \varepsilon_k$ and the bottom row
	is played with probability $\varepsilon_k > 0$ with $\varepsilon_k$ going to 0
	when $k$ goes to $\infty$ (the ($\varepsilon_k$) used in the game in
	Figure~\ref{fig:co_buchi_unfair} works here as well) where $k$ denotes the
	number of times the state $q_0$ is seen. Such a strategy is subgame
	almost-surely winning and does not have finite choice. In fact, it can be shown
	that all Player $\A$ finite-choice strategies have value 0 in that game. 
	
	Interestingly, the transfer of memory of Theorem~\ref{thm:transfer_memory} is
	adapted in Theorem~\ref{thm:transfer_memory_finite_choice} with the memory that
	is sufficient in turn-based games -- for those PI objectives that have a
	\textquotedblleft neutral color\textquotedblright 
	-- if we additionally assume that the subgame optimal strategy considered has
	finite choice. First, let us define what is meant by 
	\textquotedblleft neutral color\textquotedblright
	, then we define the turn-based version of \textsf{PSAW}.
	\begin{definition}[Objective with a neutral color]
		Consider a set of colors $\colSet$ and a PI objective $W \subseteq
		\colSet^\omega$. It has a \emph{neutral color} if there is some (neutral) color
		$k \in \colSet$ such that, for all $\rho = \rho_0 \cdot \rho_1 \cdots \in
		\colSet^\omega$, we have $\rho \in W \Leftrightarrow \rho_0 \cdot k \cdot \rho_1
		\cdot k \cdots \in W$. 
	\end{definition}
	
	\begin{definition}[\textsf{PASW} objective in turn-based games]
		Consider a PI objective $W \subseteq \colSet^\omega$. It is \emph{positionally
			subgame almost-surely winnable in turn-based games} (\textsf{PSAWT} for short)
		if in all turn-based games $\G = \Games{\Aconc}{W}$ where there is a subgame
		almost-surely winning strategy, there is a positional one. 
	\end{definition}
	\begin{theorem}[Proof in
		Appendix~\ref{proof:thm_transfer_memory_finite_choice}]
		Consider a \textsf{PSAWT} PI objective $W \subseteq \colSet^\omega$ with a
		neutral color 
		and a  concurrent game $\G$ with objective $W$. Assume  there
		is a subgame optimal strategy that has finite choice. Then, there is a
		positional one.
		\label{thm:transfer_memory_finite_choice}
	\end{theorem}
	\begin{proof}[Sketch]
		A finite-choice strategy $\s_\A$ 
		plays only among a finite number of $\GF$-strategies at each state. The idea
		is therefore to modify the game $\G_u$ of the previous subsection into a game
		$\G'_u$ by transforming it into a (finite) turn-based game. At each state,
		Player $\A$ chooses first her $\GF$-strategy. She can choose among only a finite
		number of them: she has at her disposal, at a state $q$, only optimal
		$\GF$-strategies in $S^{\s_\A}_q$ (recall
		Definition~\ref{def:finite_choice_strategies})
		. We consider the objective $W$ in that new arena where Player $\B$ states are
		colored with a neutral color. The existence, in $\G$, of a subgame optimal
		strategy that has finite choice ensures that all states in $\G_u'$ have positive
		values. We can then conclude as for Theorem~\ref{thm:transfer_memory}: a subgame
		optimal strategy can be obtained by gluing together subgame almost-surely
		winning strategies in the (turn-based) games $\G_u'$ (that can be chosen
		positional by assumption).
	\end{proof}
	
	As an application, one can realize that the parity, mean-payoff and generalized
	Büchi objectives have a neutral color and are \textsf{PSAWT}
	(\cite{DBLP:conf/soda/ChatterjeeJH04,liggett1969stochastic,DBLP:conf/qest/ChatterjeeAH04}).
	Hence, for these objectives, if there exists an 
	optimal strategy that has finite choice, then there is one that is positional.
	\begin{corollary}[Proof in
		Appendix~\ref{proof:coro_transfer_memory_finite_choice}]
		Consider a concurrent game $\G$ with a parity (resp. mean-payoff, resp.
		generalized Büchi) objective. Assume that there is an 
		optimal strategy that has finite choice in $\G$. Then, there is a positional
		one.
		\label{coro:transfer_memory_finite_choice}
	\end{corollary}

	\bibliographystyle{plain}
	\bibliography{biblio}

\begin{thebibliography}{10}

\bibitem{BCJ18}
Roderick Bloem, Krishnendu Chatterjee, and Barbara Jobstmann.
\newblock {\em Handbook of Model Checking}, chapter Graph games and reactive
  synthesis, pages 921--962.
\newblock Springer, 2018.

\bibitem{BBSFSTTCS21}
Benjamin Bordais, Patricia Bouyer, and St{\'{e}}phane~Le Roux.
\newblock From local to global determinacy in concurrent graph games.
\newblock In Mikolaj Bojanczyk and Chandra Chekuri, editors, {\em 41st {IARCS}
  Annual Conference on Foundations of Software Technology and Theoretical
  Computer Science, {FSTTCS} 2021, December 15-17, 2021, Virtual Conference},
  volume 213 of {\em LIPIcs}, pages 41:1--41:14. Schloss Dagstuhl -
  Leibniz-Zentrum f{\"{u}}r Informatik, 2021.

\bibitem{BBSCSLarXiv}
Benjamin Bordais, Patricia Bouyer, and St{\'{e}}phane~Le Roux.
\newblock Optimal strategies in concurrent reachability games.
\newblock {\em CoRR}, abs/2110.14724, 2021.

\bibitem{BBSCSL22}
Benjamin Bordais, Patricia Bouyer, and St{\'{e}}phane~Le Roux.
\newblock Optimal strategies in concurrent reachability games.
\newblock In Florin Manea and Alex Simpson, editors, {\em 30th {EACSL} Annual
  Conference on Computer Science Logic, {CSL} 2022, February 14-19, 2022,
  G{\"{o}}ttingen, Germany (Virtual Conference)}, volume 216 of {\em LIPIcs},
  pages 7:1--7:17. Schloss Dagstuhl - Leibniz-Zentrum f{\"{u}}r Informatik,
  2022.

\bibitem{BBSArXivICALP}
Benjamin Bordais, Patricia Bouyer, and St{\'{e}}phane~Le Roux.
\newblock Playing (almost-)optimally in concurrent b{\"{u}}chi and
  co-b{\"{u}}chi games.
\newblock {\em CoRR}, abs/2203.06966, 2022.

\bibitem{DBLP:conf/concur/Bouyer0ORV20}
Patricia Bouyer, St{\'{e}}phane~Le Roux, Youssouf Oualhadj, Mickael Randour,
  and Pierre Vandenhove.
\newblock Games where you can play optimally with arena-independent finite
  memory.
\newblock In {\em 31st International Conference on Concurrency Theory, {CONCUR}
  2020, September 1-4, 2020, Vienna, Austria (Virtual Conference)}, pages
  24:1--24:22, 2020.

\bibitem{CH07}
Krishnendu Chatterjee.
\newblock Concurrent games with tail objectives.
\newblock {\em Theor. Comput. Sci.}, 388(1-3):181--198, 2007.

\bibitem{DBLP:conf/qest/ChatterjeeAH04}
Krishnendu Chatterjee, Luca de~Alfaro, and Thomas~A. Henzinger.
\newblock Trading memory for randomness.
\newblock In {\em 1st International Conference on Quantitative Evaluation of
  Systems {(QEST} 2004), 27-30 September 2004, Enschede, The Netherlands},
  pages 206--217. {IEEE} Computer Society, 2004.

\bibitem{AH06}
Krishnendu Chatterjee, Luca de~Alfaro, and Thomas~A. Henzinger.
\newblock The complexity of quantitative concurrent parity games.
\newblock In {\em Proceedings of the Seventeenth Annual {ACM-SIAM} Symposium on
  Discrete Algorithms, {SODA} 2006, Miami, Florida, USA, January 22-26, 2006},
  pages 678--687. {ACM} Press, 2006.

\bibitem{DBLP:journals/iandc/Chatterjee0GH15}
Krishnendu Chatterjee, Laurent Doyen, Hugo Gimbert, and Thomas~A. Henzinger.
\newblock Randomness for free.
\newblock {\em Inf. Comput.}, 245:3--16, 2015.

\bibitem{DBLP:journals/iandc/ChatterjeeI15}
Krishnendu Chatterjee and Rasmus Ibsen{-}Jensen.
\newblock Qualitative analysis of concurrent mean-payoff games.
\newblock {\em Inf. Comput.}, 242:2--24, 2015.

\bibitem{DBLP:conf/soda/ChatterjeeJH04}
Krishnendu Chatterjee, Marcin Jurdzinski, and Thomas~A. Henzinger.
\newblock Quantitative stochastic parity games.
\newblock In J.~Ian Munro, editor, {\em Proceedings of the Fifteenth Annual
  {ACM-SIAM} Symposium on Discrete Algorithms, {SODA} 2004, New Orleans,
  Louisiana, USA, January 11-14, 2004}, pages 121--130. {SIAM}, 2004.

\bibitem{AH00}
Luca de~Alfaro and Thomas~A. Henzinger.
\newblock Concurrent omega-regular games.
\newblock In {\em 15th Annual {IEEE} Symposium on Logic in Computer Science,
  Santa Barbara, California, USA, June 26-29, 2000}, pages 141--154. {IEEE}
  Computer Society, 2000.

\bibitem{everett57}
Hugh Everett.
\newblock Recursive games.
\newblock {\em Annals of Mathematics Studies -- Contributions to the Theory of
  Games}, 3:67--78, 1957.

\bibitem{GIH10}
Hugo Gimbert and Florian Horn.
\newblock Solving simple stochastic tail games.
\newblock In Moses Charikar, editor, {\em Proceedings of the Twenty-First
  Annual {ACM-SIAM} Symposium on Discrete Algorithms, {SODA} 2010, Austin,
  Texas, USA, January 17-19, 2010}, pages 847--862. {SIAM}, 2010.

\bibitem{DBLP:conf/icalp/KieferMSTW20}
Stefan Kiefer, Richard Mayr, Mahsa Shirmohammadi, Patrick Totzke, and Dominik
  Wojtczak.
\newblock How to play in infinite mdps (invited talk).
\newblock In Artur Czumaj, Anuj Dawar, and Emanuela Merelli, editors, {\em 47th
  International Colloquium on Automata, Languages, and Programming, {ICALP}
  2020, July 8-11, 2020, Saarbr{\"{u}}cken, Germany (Virtual Conference)},
  volume 168 of {\em LIPIcs}, pages 3:1--3:18. Schloss Dagstuhl -
  Leibniz-Zentrum f{\"{u}}r Informatik, 2020.

\bibitem{KNPS21}
Marta Kwiatkowska, Gethin Norman, Dave Parker, and Gabriel Santos.
\newblock Automatic verification of concurrent stochastic systems.
\newblock {\em Formal Methods in System Design}, 58:188--250, 2021.

\bibitem{liggett1969stochastic}
Thomas~M Liggett and Steven~A Lippman.
\newblock Stochastic games with perfect information and time average payoff.
\newblock {\em Siam Review}, 11(4):604--607, 1969.

\bibitem{martin98}
Donald~A. Martin.
\newblock The determinacy of blackwell games.
\newblock {\em The Journal of Symbolic Logic}, 63(4):1565--1581, 1998.

\bibitem{shapley1950basic}
Lloyd~S Shapley and RN~Snow.
\newblock Basic solutions of discrete games.
\newblock {\em Contributions to the Theory of Games}, 1(24):27--27, 1950.

\bibitem{thomas02}
Wolfgang {\relax Th}omas.
\newblock Infinite games and verification.
\newblock In {\em Proc. 14th International Conference on Computer Aided
  Verification (CAV'02)}, volume 2404 of {\em Lecture Notes in Computer
  Science}, pages 58--64. Springer, 2002.
\newblock Invited Tutorial.

\bibitem{vonNeuman}
John von Neumann and Oskar Morgenstern.
\newblock {\em Theory of Games and Economic Behavior}.
\newblock Princeton Univ. Press, Princeton, 1944.

\bibitem{DBLP:conf/fossacs/Zielonka04}
Wieslaw Zielonka.
\newblock Perfect-information stochastic parity games.
\newblock In Igor Walukiewicz, editor, {\em Foundations of Software Science and
  Computation Structures, 7th International Conference, {FOSSACS} 2004, Held as
  Part of the Joint European Conferences on Theory and Practice of Software,
  {ETAPS} 2004, Barcelona, Spain, March 29 - April 2, 2004, Proceedings},
  volume 2987 of {\em Lecture Notes in Computer Science}, pages 499--513.
  Springer, 2004.

\end{thebibliography}
	
	\newpage
	\appendix	
	\section{Additional preliminaries}
	For a set $Q$ and a subset $S \subseteq Q^+$ of finite sequences of elements of
	$Q$, we denote by $S^\omega \subseteq Q^\omega$ the set of infinite sequences of
	elements of $Q$ with infinitely many prefixes in $S$: $S^\omega := \{ \pi \in
	Q^\omega \mid \forall i \in \N,\; \exists j \geq i,\; \pi_{\leq j} \in S \}$.
	
	Let us also define the notion of (countable) Markov chain. A Markov chain is a
	pair $(Q,\mathbb{P})$ where $Q$ is the set of states and $\mathbb{P}: Q \times Q
	\rightarrow [0,1]$ is the probability function such that, for all states $q \in
	Q$, we have $\sum_{q' \in Q} \mathbb{P}(q,q') = 1$.
		
	\section{Proofs from
		Section~\ref{sec:conditions_subgame_optimal}}
	\label{app:conditions_subgame_optimal}
	We recall a proposition from \cite{BBSCSLarXiv} (specifically, Proposition 42
	in \cite{BBSCSLarXiv}) that we use in this section.
	\begin{proposition}
		\label{prop:outcome_valuation}
		Consider a PI concurrent game $\G = \Games{\Aconc}{W}$
		, a state $q \in Q$ and strategies $\sigma_\A,\sigma_\B \in \Dist(A) \times
		\Dist(B)$ for both players in the game in normal form $\formNF_q$. We have the
		following relation:
		\[\sum_{q' \in Q} \mathbb{P}^{q,q'}(\sigma_\A,\sigma_\B)
		\cdot \MarVal{\G}(q') = \outM_{\formNF_q}(\sigma_\A,\sigma_\B)\]
	\end{proposition}
	
	\subsection{Proof of Lemma~\ref{prop:sub_game_opt_implies_locally_opt}}
	\label{proof:prop_sub_game_opt_implies_locally_opt}
	We show the proposition below.
	\begin{proposition}
		\label{prop:optimal_implies_locally_optimal}
		Consider a PI concurrent game $\G$ and a Player $\A$ optimal strategy $\s_\A$
		from a state $q \in Q$. Then, for all Player $\B$ actions $b \in B_q$, we have
		$\outM_{\formNF_q}(\s_\A(q),b) \geq \MarVal{\G}(q)$.
	\end{proposition}
	The proof is identical to the proof of Lemma 17 in \cite{BBSArXivICALP} in the
	case of positional strategies.
	\begin{proof}
		Assume towards a contradiction that 
		it is not the case, i.e. there is a Player $\B$ action $b \in B_q$ such that
		$\outM_{\formNF_q}(\s_\A(q),b) \leq \MarVal{\G}(q) - \varepsilon$ for some
		$\varepsilon > 0$. Let us show that the strategy $\s_\A$ is then not optimal
		from $q$. Consider a Player $\B$ strategy $\s_\B'$ such that for all $q' \in Q$,
		the value of the strategy at state $q'$ is at most $\MarVal{\G}(q') +
		\varepsilon/2$: $\MarVal{\G}[\s_\B](q') \leq \MarVal{\G}(q') + \varepsilon/2$.
		We then define a Player $\B$ strategy $\s_\B$ as follows: $\s_\B(q) := b$ and,
		for all $q' \in Q$, we have $\s_\B(q \cdot q') := \s_\B(q')$. Then, as $W$ is
		prefix-independent and by  Proposition~\ref{prop:outcome_valuation}, we obtain:
		\begin{align*}
		\mathbb{P}^{\s_\A,\s_\B}_q[W] & = \sum_{q' \in Q}
		\mathbb{P}^{\s_\A,\s_\B}(q)[q'] \cdot \mathbb{P}^{\s_\A^{q \cdot q'},\s_\B^{q
				\cdot q'}}_{q'}[W] \\
		& \leq \sum_{q' \in Q} \mathbb{P}^{\s_\A,\s_\B}(q)[q'] \cdot (\MarVal{\G}(q')
		+ \varepsilon/2) \\
		& = \sum_{q' \in Q} \mathbb{P}^{q,q'}(\s_\A(\rho),\s_\B(\rho)) \cdot
		\MarVal{\G}(q') + \varepsilon/2 \\
		& = \outM_{\formNF_q}(\s_\A(q),b) + \varepsilon/2 \\
		& \leq \MarVal{\G}[q] - \varepsilon + \varepsilon/2 = \MarVal{\G}[q] -
		\varepsilon/2
		\end{align*}
		That is, the strategy $\s_\A^\rho$ is not optimal from $q$.	
	\end{proof}
	
	The proof of Lemma~\ref{prop:sub_game_opt_implies_locally_opt} is then a direct
	consequence.
	\begin{proof}
		Assume towards a contradiction that $\s_\A$ is not locally optimal. That is,
		there is some $\rho = \rho' \cdot q \in Q^+$ and a Player $\B$ action $b \in
		B_q$ such that $\outM_{\formNF_q}(\s_\A(\rho),b) \leq \MarVal{\G}(q) -
		\varepsilon$ for some $\varepsilon > 0$. It follows by
		Proposition~\ref{prop:optimal_implies_locally_optimal} that the residual
		strategy $\s_\A^\rho$ is not optimal from $\rho$, hence the contradiction.
	\end{proof}
	
	\subsection{Proof of Lemma~\ref{lem:locally_opt_implies_same_value}}
	\label{proof:lem_locally_opt_implies_same_value}
	First, let us state a straightforward proposition about locally optimal
	strategies.
	\begin{proposition}
		\label{prop:locally_opt_next}
		In a PI concurrent game $\G$, for a Player $\A$ locally optimal strategy
		$\s_\A$ and a deterministic Player $\B$ strategy, for all finite paths $\rho =
		\rho' \cdot q \in Q^+$, we have
		$\MarVal{\G}(q) \leq \sum_{u \in V_{\G}} u \cdot
		\mathbb{P}_\rho^{\s_\A,\s_\B}[Q_u]$.
	\end{proposition}
	\begin{proof}
		We let $b := \s_\B(\rho) \in B_q$. We have:
		\begin{align*}
		\MarVal{\G}(q) & \leq \sum_{a \in A} \s_\A(\rho)(a) \cdot \LiftVal{\G} \circ
		\delta(q,a,b) = \sum_{a \in A} \s_\A(\rho)(a) \cdot (\sum_{q' \in Q}
		\distribFunc \circ \delta(q,a,b)(q') \cdot \MarVal{\G}(q')) \\
		& = \sum_{u \in V_{\G}} \sum_{a \in A} \s_\A(\rho)(a) \cdot (\sum_{q' \in Q_u}
		\distribFunc \circ \delta(q,a,b)(q') \cdot \MarVal{\G}(q')) \\
		& = \sum_{u \in V_{\G}} u \cdot \left(\sum_{q' \in Q_u} \sum_{a \in A}
		\s_\A(\rho)(a) \cdot \distribFunc \circ \delta(q,a,b) \right) = \sum_{u \in
			V_{\G}} u \cdot \mathbb{P}_\rho^{\s_\A,\s_\B}[Q_u]
		\end{align*}
	\end{proof}
	
	We can now proceed to the proof of
	Lemma~\ref{lem:locally_opt_implies_same_value}.
	\begin{proof}
		We denote by $v$ the valuation $\MarVal{\G}$. We extend the valuation $v$ to
		finite paths: $v^+: Q^+ \rightarrow [0,1]$ such that, for all $\rho \cdot q \in
		Q^+$, we have $v^+(\rho \cdot q) := v(q)$. If all states have value 0 or 1, the
		lemma straightforwardly holds. Let us now assume that there are some states of
		value between 0 and 1. Let us denote by $0 = u_0 < u_1 < \ldots < u_n = 1$ the
		states values in $V_{\G}$
		, i.e. $|V_{\G}| = n+1$. For all $i \leq n+1$, we denote by $Q_{\leq i}$ the
		set $Q_{\leq i} := v^{-1}[\{ u_0,\ldots,u_i \}]$
		. We show by induction on $k \leq n$ the following property $\mathcal{P}(k)$:
		\begin{displaymath}
		\mathbb{P}^{\s_\A,\s_\B}(\bigcup_{0 \leq i \leq k} Q^* \cdot (Q_{u_i})^\omega
		\cap Q^* \cdot (Q_{\leq k})^\omega) = \mathbb{P}^{\s_\A,\s_\B}(Q^* \cdot
		(Q_{\leq k})^\omega)
		\end{displaymath}
		This straightforwardly holds for $k = 0$. Consider now some $k \leq n$ and
		assume that $\mathcal{P}(k)$ holds. 
		We have:
		\begin{align*}
		\mathbb{P}^{\s_\A,\s_\B}(Q^* \cdot (Q_{\leq k+1})^\omega ) & =
		\mathbb{P}^{\s_\A,\s_\B}(Q^* \cdot (Q_{k+1})^\omega \cap Q^* \cdot (Q_{\leq
			k+1})^\omega ) \\
		& + \mathbb{P}^{\s_\A,\s_\B}(Q^* \cdot (Q_{\leq k})^\omega \cap Q^* \cdot
		(Q_{\leq k+1})^\omega ) \\
		& + \mathbb{P}^{\s_\A,\s_\B}((Q^* \cdot Q_{\leq k})^\omega \cap (Q^* \cdot
		Q_{k+1})^\omega \cap Q^* \cdot (Q_{\leq k+1})^\omega) 
		\end{align*}
		Let us show that the term $\mathbb{P}^{\s_\A,\s_\B}((Q^* \cdot Q_{\leq
			k})^\omega \cap (Q^* \cdot Q_{k+1})^\omega \cap Q^* \cdot (Q_{\leq k+1})^\omega)
		= 0$. This holds if $k = n$ since then, $u_{k+1} = 1$ and being locally optimal
		means that all states seen with positive probability have value 1. Assume now
		that $k < n$ and $u_{k+1} < 1$. For all $l \in \N$, we denote by $E_{\leq k}^l$
		the event $E_{\leq k}^l := Q^l \cdot (Q_{k+1} \cap Q \cdot Q_{\leq k})$ and by
		$E_{\geq k+2}^l$ the event $E_{\geq k+2}^l := Q^l \cdot (Q_{k+1} \cap Q \cdot
		Q_{\geq k+2})$. Let also $E^l := E_{\leq k}^l \cup E_{\geq k+2}^l$. We also
		denote by $E_{\leq k}^\infty := \cap_{d \in \N} \cup_{l \geq d} E_{\leq k}^l$,
		$E_{\geq k+2}^\infty := \cap_{d \in \N} \cup_{l \geq d} E_{\geq k+2}^l$ and
		$E^\infty := \cap_{d \in \N} \cup_{l \geq d} E^l$. If $
		\mathbb{P}^{\s_\A,\s_\B}(E^\infty) = 0$ then $\mathbb{P}^{\s_\A,\s_\B}((Q^*
		\cdot Q_{\leq k})^\omega \cap (Q^* \cdot Q_{k+1})^\omega \cap Q^* \cdot (Q_{\leq
			k+1})^\omega) = 0$ since $(Q^* \cdot Q_{\leq k})^\omega \cap (Q^* \cdot
		Q_{k+1})^\omega \cap Q^* \cdot (Q_{\leq k+1})^\omega \subseteq E^\infty_{\leq k}
		\subseteq E^\infty$. Let us now assume that $ \mathbb{P}^{\s_\A,\s_\B}(E^\infty)
		> 0$.
		
		Consider some $\rho = \rho' \cdot q \in Q^+ \cdot Q_{k+1}$. Since $\s_\A$ is
		locally optimal and by Proposition~\ref{prop:locally_opt_next}, we have:
		\begin{align*}
		u_{k+1} \leq \sum_{u \in V_{\G}} u \cdot \mathbb{P}_\rho^{\s_\A,\s_\B}[Q_u] &
		= u_{k+1} \cdot \mathbb{P}_\rho^{\s_\A,\s_\B}[Q_{u_{k+1}}] + \sum_{u \leq u_k} u
		\cdot \mathbb{P}_\rho^{\s_\A,\s_\B}[Q_{u}] + \sum_{u \geq u_{k+2}} u \cdot
		\mathbb{P}_\rho^{\s_\A,\s_\B}[Q_{u}] \\
		& \leq u_{k+1} \cdot \mathbb{P}_\rho^{\s_\A,\s_\B}[Q_{u_{k+1}}] + \sum_{u \leq
			u_k} u_k \cdot \mathbb{P}_\rho^{\s_\A,\s_\B}[Q_{u}] + \sum_{u \geq u_{k+2}}
		\mathbb{P}_\rho^{\s_\A,\s_\B}[Q_{u}] \\
		& = u_{k+1} \cdot \mathbb{P}_\rho^{\s_\A,\s_\B}[Q_{u_{k+1}}] + u_k \cdot
		\mathbb{P}_\rho^{\s_\A,\s_\B}[Q_{\leq k}] +
		\mathbb{P}_\rho^{\s_\A,\s_\B}[Q_{\geq k+2}] \\
		\end{align*}
		Denoting $\mathbb{P}_\rho^{\s_\A,\s_\B}[Q_{\leq k}]$ by $p_k$ and
		$\mathbb{P}_\rho^{\s_\A,\s_\B}[Q_{\geq k+2}]$ by $p_{k+2}$, we obtain:
		\begin{align*}
		& u_{k} \cdot p_k + p_{k+2} \geq u_{k+1} \cdot (p_k+p_{k+2}) \\
		\Leftrightarrow \; & p_{k+2} \cdot (1 - u_{k+1}) \geq p_k \cdot (u_{k+1} -
		u_k) \\
		\Leftrightarrow \; & p_{k+2} \geq p_k \cdot \frac{u_{k+1} - u_k}{1 - u_{k+1}}
		= p_{k+2} \cdot x
		\end{align*}
		for $x := \frac{u_{k+1} - u_k}{1 - u_{k+1}} \geq 0$. Hence, setting $c :=
		\frac{x}{1 + x} \in ]0,1[$, we have $\frac{p_{k+2}}{p_k + p_{k+2}} \geq c$. That
		is, for $l = |\rho|$, we have: 
		\begin{displaymath}
		\mathbb{P}_\rho^{\s_\A,\s_\B}(\lnot E^l_{\geq k+2} \cap E^l) \leq (1 - c)
		\cdot \mathbb{P}_\rho^{\s_\A,\s_\B}(E^l) 
		\end{displaymath}
		Then, consider the probability $\mathbb{P}^{\s_\A,\s_\B}(\lnot E^\infty_{\geq
			k+2} \mid E^\infty)$ (recall that $\mathbb{P}^{\s_\A,\s_\B}(E^\infty) > 0$).
		This is equal to:
		\begin{displaymath}
		\lim\limits_{d \rightarrow \infty} \mathbb{P}^{\s_\A,\s_\B}(\bigcap_{l \geq d}
		(\lnot E^l_{\geq k+2}) \mid E^\infty) \leq \lim\limits_{d \rightarrow \infty}
		(\lim\limits_{t \rightarrow \infty} (1-c)^t) = 0
		\end{displaymath}
		
		Furthermore:
		\begin{displaymath}
		(Q^* \cdot Q_{\leq k})^\omega \cap (Q^* \cdot Q_{k+1})^\omega \cap Q^* \cdot
		(Q_{\leq k+1})^\omega \subseteq (Q^\omega \setminus E^\infty_{\geq k+2}) \cap
		E^\infty
		\end{displaymath}
		
		It follows that:
		\begin{displaymath}
		\mathbb{P}^{\s_\A,\s_\B}((Q^* \cdot Q_{\leq k})^\omega \cap (Q^* \cdot
		Q_{k+1})^\omega \cap Q^* \cdot (Q_{\leq k+1})^\omega) \leq
		\mathbb{P}^{\s_\A,\s_\B}(\lnot E^\infty_{\geq k+2} \cap E^\infty) =
		\mathbb{P}^{\s_\A,\s_\B}(\lnot E^\infty_{\geq k+2} \mid E^\infty) \cdot
		\mathbb{P}^{\s_\A,\s_\B}(E^\infty) = 0
		\end{displaymath}
		
		Overall:
		\begin{align*}
		\mathbb{P}^{\s_\A,\s_\B}(Q^* \cdot (Q_{\leq k+1})^\omega ) & =
		\mathbb{P}^{\s_\A,\s_\B}(Q^* \cdot (Q_{k+1})^\omega \cap Q^* \cdot (Q_{\leq
			k+1})^\omega ) + \mathbb{P}^{\s_\A,\s_\B}(Q^* \cdot (Q_{\leq k})^\omega \cap Q^*
		\cdot (Q_{\leq k+1})^\omega ) \\
		& = \mathbb{P}^{\s_\A,\s_\B}(Q^* \cdot (Q_{k+1})^\omega ) +
		\mathbb{P}^{\s_\A,\s_\B}(Q^* \cdot (Q_{\leq k})^\omega ) \\
		& = \mathbb{P}^{\s_\A,\s_\B}(Q^* \cdot (Q_{k+1})^\omega ) +
		\mathbb{P}^{\s_\A,\s_\B}(\bigcup_{0 \leq i \leq k} Q^* \cdot (Q_{u_i})^\omega
		\cap Q^* \cdot (Q_{\leq k})^\omega) \\
		& = \mathbb{P}^{\s_\A,\s_\B}(\bigcup_{0 \leq i \leq k+1} Q^* \cdot
		Q_{u_i}^\omega \cap Q^* \cdot (Q_{\leq k+1})^\omega )
		\end{align*}
		Hence, $\mathcal{P}(k+1)$ holds. In fact, it does for all $k \leq n$. 	Then,
		the lemma exactly corresponds to $\mathcal{P}(n)$ since $Q_{\leq n} = Q$.
	\end{proof}
	
	\subsection{Proof of Lemma~\ref{lem:locally_opt_implies_convex_comb}}
	\label{proof:lem_locally_opt_implies_convex_comb_ok}
	In fact, we prove the more general lemma below.
	\begin{lemma}[Proof~\ref{proof:lem_locally_opt_implies_convex_comb_ok}]
		\label{lem:locally_opt_implies_convex_comb_ok}
		Consider an PI concurrent game $\G
		$
		, a Player $\A$ locally optimal strategy $\s_\A$ and a Player $\B$
		deterministic strategy $\s_\B$. For a finite sequence $\rho = \rho' \cdot q \in
		Q^+$
		:
		\begin{displaymath}
		\MarVal{\G}(q) \leq \sum_{u \in V_{\G}} u \cdot
		\mathbb{P}^{\s_\A^\rho,\s_\B^\rho}_\rho[Q^* \cdot (Q_u)^\omega]
		\end{displaymath}
	\end{lemma}
	\begin{proof}
		We consider the locally optimal strategy $\s_\A$, a deterministic Player $\B$
		strategy $\s_\B$ and the finite path $\rho \in Q^+$.	
		Let us consider the infinite Markov chain $\mathcal{M} = (\rho \cdot
		Q^+,\mathbb{P})$ with $\mathbb{P}: \rho \cdot Q^+ \times \rho \cdot Q^+
		\rightarrow [0,1]$ where, for all $\pi \in Q^+$
		, we have $\mathbb{P}(\rho \cdot \pi,\rho \cdot \pi \cdot q) :=
		\mathbb{P}^{\s_\A,\s_\B}(\rho \cdot \pi,q)$. All other probabilities are equal
		to 0. 
		The probability measure $\mathbb{P}$ is extended to finite paths starting at
		$\rho$, cylinders and arbitrary Borel sets (in particular, to $W$). 
		
		
		We additionally define inductively the function $f: \rho \cdot Q^+ \rightarrow
		\N$ in the following way: $f(\rho) := 0$ and for all $\pi \cdot q \cdot q' \in
		Q^+$, we set $f(\rho \cdot \pi \cdot q \cdot q') := f(\rho \cdot \pi \cdot q)$
		if $\MarVal{\G}(q) = \MarVal{\G}(q')$ 
		and $f(\rho \cdot q \cdot q') := f(\rho \cdot q) + 1$ otherwise. 
		This function counts the number of changes of values.
		
		Let $\varepsilon > 0$
		. By Lemma~\ref{lem:locally_opt_implies_same_value}, we have 
		\begin{displaymath}
		\lim\limits_{n \rightarrow \infty}  \mathbb{P}(Q^+ \cdot f^{-1}[\{n\}]) = 0	
		\end{displaymath}
		Consider some $N \in \N$ such that $\mathbb{P}(Q^+ \cdot f^{-1}[N+1]) \leq
		\varepsilon/2$. For all $0 \leq k \leq N$ and values $u \in V_{\G}$, we denote
		by $(Q_u,k) \subseteq \rho \cdot Q^+$ the set of finite paths ending in $Q_u$
		and whose values w.r.t. the function $f$ are $k$: $(Q_u,k) := \{ \rho \cdot \pi
		\cdot q \in Q^+ \mid q \in Q_u,\; f(\rho \cdot \pi \cdot q) = k \}$. For all $n
		\in \N$, we denote by $(Q_u,k)_n := \rho \cdot Q^n \cap (Q_u,k)$.
		
		For all $0 \leq k \leq N$, we denote by $n_k \in \N$ an index such that:
		\begin{displaymath}
		\mathbb{P}(Q^+ \cdot f^{-1}[k]) \leq \mathbb{P}(Q^{\leq n_k} \cdot f^{-1}[k])
		+ \frac{\varepsilon}{2 \cdot (N+1) \cdot N_{V_{\G}}}
		\end{displaymath} 
		where $N_{V_{\G}} := |V_{\G}|$ (this exists since $\mathbb{P}(Q^+ \cdot S) =
		\lim_{n \rightarrow \infty} \mathbb{P}(Q^{\leq n} \cdot S)$
		). In particular, this implies $\mathbb{P}(Q^{> n_k} \cdot f^{-1}[k]) =
		\mathbb{P}(Q^+ \cdot f^{-1}[k]) - \mathbb{P}(Q^{\leq n_k} \cdot f^{-1}[k]) \leq
		\varepsilon/(2 \cdot (N+1) \cdot N_V)$. 
		Now, let $n := \max_{0 \leq k \leq N} n_k$. 
		Let us show that, for all $0 \leq k \leq N$ and $u \in V_{\G}$, we have:
		\begin{equation}
		\label{eqn:close_always_value}
		\mathbb{P}[(Q_u,k)_n] = \sum_{\pi \in (Q_u,k)_n} \mathbb{P}(\pi) \leq
		\mathbb{P}(Q^+ \cdot (Q_u,k)^\omega) + \frac{\varepsilon}{2 \cdot (N+1) \cdot
			N_V}
		\end{equation}
		
		Indeed, since $n \geq n_k$, we have:
		\begin{align*}
		\mathbb{P}(Q^+ \cdot (Q_u,k)^\omega) & \geq \mathbb{P}(Q^{\leq n} \cdot
		(Q_u,k)^\omega) \\
		& = \sum_{\pi \in (Q_u,k)_n} \mathbb{P}(\pi) \cdot
		\mathbb{P}_\pi((Q_u,k)^\omega) \\
		& = \sum_{\pi \in (Q_u,k)_n} \mathbb{P}(\pi) \cdot (1 - \mathbb{P}_\pi(Q^+
		\cdot f^{-1}[k+1])) \\
		& = \sum_{\pi \in (Q_u,k)_n} \mathbb{P}(\pi) - \sum_{\pi \in (Q_u,k)_n}
		\mathbb{P}(\pi) \cdot \mathbb{P}_\pi(Q^+ \cdot  f^{-1}[k+1])) \\
		& \geq \sum_{\pi \in (Q_u,k)_n} \mathbb{P}(\pi) - \mathbb{P}_{> n}(Q^+ \cdot
		f^{-1}[k+1])) \\
		& \geq \sum_{\pi \in (Q_u,k)_n} \mathbb{P}(\pi) - \frac{\varepsilon}{2 \cdot
			(N+1) \cdot N_V} \\
		\end{align*}

		We obtain Equation~(\ref{eqn:close_always_value}). In the following, we denote
		the valuation $\MarVal{\G}$ by $v: Q \rightarrow [0,1]$. Let us show by
		induction on $i$ the following property $\mathcal{P}(i)$: \textquotedblleft
		$\sum_{\pi \in \rho \cdot Q^i} \mathbb{P}(\pi) \cdot v(\pi) \geq
		v(\rho)$\textquotedblright where $v(\pi) \in [0,1]$ refers to $v(q)$ for $q \in
		Q$ the last state of $\pi$
		. The property $\mathcal{P}(0)$ straightforwardly holds. Assume now that this
		property holds for some $i \in \N$. We have, by
		Proposition~\ref{prop:outcome_valuation} and since $\s_\A$ is locally optimal:
		\begin{align*}
		\sum_{\pi \in \rho \cdot Q^{i+1}} \mathbb{P}(\pi) \cdot v(\pi) & = \sum_{\pi =
			\rho \cdot q_1 \cdots q_i \cdot q_{i+1} \in \rho \cdot Q^{i+1}} \mathbb{P}(\pi)
		\cdot v(q_{i+1}) \\
		& = \sum_{\pi' = \rho \cdot q_0 \cdots q_i \in \rho \cdot Q^{i}} \sum_{q_{i+1}
			\in Q} \mathbb{P}(\pi' \cdot q_{i+1}) \cdot v(q_{i+1}) \\
		& = \sum_{\pi' = \rho \cdot q_1 \cdots q_i \in \rho \cdot Q^{i}} \sum_{q_{i+1}
			\in Q} \mathbb{P}(\pi') \cdot \mathbb{P}(\pi',\pi' \cdot q_{i+1}) \cdot
		v(q_{i+1}) \\
		& = \sum_{\pi' = \rho \cdot q_0 \cdots q_i \in \rho \cdot Q^{i}}
		\mathbb{P}(\pi') \cdot \sum_{q_{i+1} \in Q} \mathbb{P}(\pi',\pi' \cdot q_{i+1})
		\cdot v(q_{i+1}) \\
		& = \sum_{\pi' = \rho \cdot q_1 \cdots q_i \in \rho \cdot Q^{i}}
		\mathbb{P}(\pi') \cdot \sum_{q_{i+1} \in Q}
		\mathbb{P}^{\s_\A,\s_\B}(\pi')[q_{i+1}] \cdot v(q_{i+1}) \\
		& = \sum_{\pi' = \rho \cdot q_1 \cdots q_i \in \rho \cdot Q^{i}}
		\mathbb{P}(\pi') \cdot \sum_{q_{i+1} \in Q}
		\mathbb{P}^{q_i,q_{i+1}}(\s_\A(\pi'),\s_\B(\pi')) \cdot v(q_{i+1}) \\
		& = \sum_{\pi' = \rho \cdot q_1 \cdots q_i \in \rho \cdot Q^{i}}
		\mathbb{P}(\pi') \cdot \outM_{\formNF_{q_i}}(\s_\A(\pi'),\s_\B(\pi')) \\
		& \geq \sum_{\pi' = \rho \cdot q_1 \cdots q_i \in \rho \cdot Q^{i}}
		\mathbb{P}(\pi') \cdot v(q_i) \\
		& \geq v(\rho)
		\end{align*}
		Hence, the property $\mathcal{P}(i)$ is ensured for all $i \in \mathbb{N}$, in
		particular $\mathcal{P}(n)$ holds. It follows that:
		
		\begin{align*}
		\hspace*{-1.3cm}
		\sum_{u \in V_{\G}} u \cdot \mathbb{P}(Q^+ \cdot (Q_u)^\omega) & \geq \sum_{u
			\in V_{\G}} u \cdot \mathbb{P}(\bigcup_{k = 0}^N Q^+ \cdot (Q_u,k)^\omega) &
		\text{ since }\bigcup_{k = 0}^N Q^+ \cdot (Q_u,k)^\omega \subseteq Q^+ \cdot
		Q_u^\omega\\
		& \geq \sum_{k = 0}^N \sum_{u \in V_{\G}} u \cdot \mathbb{P}(Q^+ \cdot
		(Q_u,k)^\omega) & \text{ since } (Q_u,k)^\omega \cap (Q_u,j)^\omega = \emptyset
		\text{ if }k \neq j\\
		& \geq \sum_{k = 0}^N \sum_{u \in V_{\G}} u \cdot (\sum_{\pi \in (Q_u,k)_n}
		\mathbb{P}(\pi) - \frac{\varepsilon}{2 \cdot (N+1) \cdot N_{V_\G}}) & \text{ by
			Equation~(\ref{eqn:close_always_value}) }\\
		& \geq \sum_{k = 0}^N \sum_{u \in V} \sum_{\pi \in (Q_u,k)_n} \mathbb{P}(\pi)
		\cdot v(\pi) - \frac{\varepsilon}{2} & \text{ since }u = v(\pi), \forall \pi \in
		(Q_u,k)_n \\
		& = \sum_{k = 0}^N \sum_{\pi \in \rho \cdot Q^n \cap f^{-1}[k]}
		\mathbb{P}(\pi) \cdot v(\pi) - \frac{\varepsilon}{2} & \text{ since }\rho \cdot
		Q^n \cap f^{-1}[k] = \uplus_{u \in V_{\G}} (Q_u,k)_n \\
		& \geq \sum_{\pi \in \rho \cdot Q^n} \mathbb{P}(\pi) \cdot v(\pi) -
		\varepsilon & \text{ since }\mathbb{P}(Q^+ \cdot f^{-1}[N+1]) \leq
		\frac{\varepsilon}{2}\\
		& \geq v(\rho) - \varepsilon & \text{ by }\mathcal{P}(n)
		\end{align*}
		As this holds for all $\varepsilon > 0$, we obtain: $\sum_{u \in V_{\G}} u
		\cdot \mathbb{P}(Q^+ \cdot (Q_u)^\omega) \geq v(\rho) = \MarVal{\G}(q)$ (for $q$
		the last state of $\rho$).
	\end{proof}
	Lemma~\ref{lem:locally_opt_implies_convex_comb} is then a direct consequence.
	
	\subsection{Levy's 0-1 Law on Markov chains}
	\label{appen:levy}
	Let us state an adaptation of Levy's 0-1 Law to the context of infinite Markov
	chains to prefix-independent objectives. This comes from Theorem 5  in
	\cite{DBLP:conf/icalp/KieferMSTW20}:
	\begin{theorem}[Levy's 0-1 Law for prefix-independent objectives]
		Consider a countable Markov chain $\mathcal{M}$ on a set of states $Q$ with a
		probability $\mathbb{P}$. For all prefix-independent objectives $W \subseteq
		Q^\omega$, the sets $W$ and $\{ \rho \in Q^\omega \mid \lim_{n \rightarrow
			\infty} \mathbb{P}_{\rho_n}(W) = 1 \}$ are equal up to a null set. This also
		holds for $Q^\omega \setminus W$: the sets $Q^\omega \setminus W$ and $\{ \rho
		\in Q^\omega \mid \lim_{n \rightarrow \infty} \mathbb{P}_{\rho_n}(W) = 0 \}$ are
		equal up to a null set.
		\label{thm:levy}
	\end{theorem}
	
	Lemma~\ref{lem:chaterjee_value_0_1} is then a direct consequence of this
	theorem. Note that Lemma~\ref{lem:chaterjee_value_0_1} also comes from Lemma 2
	in \cite{CH07}.
	
	\subsection{Proof of Theorem~\ref{thm:subgame_optimal_arbitrary_strategy}}
	\label{proof:lem_optimal_arbitrary_strategy}
	\begin{proof}
		The first property is necessary by
		Lemma~\ref{prop:sub_game_opt_implies_locally_opt}. Let us show that the second
		one is also necessary. Consider a Player $\A$ strategy $\s_\A$ such that there
		is a finite paths $\rho \in Q^+$, a deterministic Player $\B$ strategy $\s_\B$
		and a positive value $u \in V_{\G} \setminus \{ 0 \}$ such that
		$\mathbb{P}_\rho^{\s_\A^\rho,\s_\B^\rho}[W \cap Q^* \cdot (Q_u)^\omega] <
		\mathbb{P}_\rho^{\s_\A^\rho,\s_\B^\rho}[Q^* \cdot (Q_u)^\omega]$. Consider the
		countable Markov chain induced by both strategies $\s_\A^\rho$ and $\s_\B^\rho$
		whose initial state in $\rho$. We denote by $\mathbb{P}$ the probability
		function in that Markov chain. Consider the prefix-independent objective $W_u :=
		W \cup (Q^* \cdot (Q \setminus Q_u))^\omega$. We have:
		\begin{align*}
		\mathbb{P}[W_u] & = \mathbb{P}[W_u \cap Q^* \cdot (Q_u)^\omega] +
		\mathbb{P}[W_u \cap (Q^* \cdot (Q \setminus Q_u))^\omega] \\
		& = \mathbb{P}[W \cap Q^* \cdot (Q_u)^\omega] + \mathbb{P}[(Q^* \cdot (Q
		\setminus Q_u))^\omega] \\
		& < \mathbb{P}[Q^* \cdot (Q_u)^\omega] + \mathbb{P}[(Q^* \cdot (Q \setminus
		Q_u))^\omega] \\
		& = 1
		\end{align*}
		Let us now show that there is a finite paths whose last state is in $Q_u$ and
		from which the probability of $W_u$ is less than $u/2$. By
		Lemma~\ref{lem:chaterjee_value_0_1}, there is state $\pi \in Q^*$ (which
		corresponds to a finite path) in the Markov chain such that $\mathbb{P}_\pi[W_u]
		< u/2$. If $\pi \in Q^* \cdot Q_u$, we have shown what we want.	Assume now
		towards a contradiction that it is not the case and that for all paths $\pi' \in
		Q^* \cdot Q_u$, we have $\mathbb{P}_{\pi \cdot \pi'}[W_u] \geq u/2$. Then:
		\begin{align*}
		\mathbb{P}_\pi[W_u] & = \mathbb{P}_\pi[W_u \cap (Q \setminus Q_u)^* \cdot Q_u
		\cdot Q^\omega] + \mathbb{P}_\pi[W_u \cap (Q \setminus Q_u)^\omega] \\
		& = \sum_{\pi' \in (Q \setminus Q_u)^* \cdot Q_u} \mathbb{P}_{\pi}(\pi') \cdot
		\mathbb{P}_{\pi \cdot \pi'}[W_u] + \mathbb{P}_\pi[(Q \setminus Q_u)^\omega] \\
		& \geq \sum_{\pi' \in (Q \setminus Q_u)^* \cdot Q_u} \mathbb{P}_{\pi}(\pi')
		\cdot \frac{u}{2} + \mathbb{P}_\pi[(Q \setminus Q_u)^\omega] \\
		& = \mathbb{P}_\pi[(Q \setminus Q_u)^* \cdot Q_u \cdot Q^\omega] \cdot
		\frac{u}{2} + \mathbb{P}_\pi[(Q \setminus Q_u)^\omega] \\
		& \geq \frac{u}{2}
		\end{align*}
		This is a contradiction since $\mathbb{P}_\pi[W_u] < u/2$. Hence, there exists
		some path $\pi \in Q^* \cdot Q_u$ such that $\mathbb{P}_\pi[W_u] < u/2$. Since
		$W \subseteq W_u$, it follows that $\mathbb{P}_\pi[W] \leq \mathbb{P}_\pi[W_u] <
		u/2$. That is, the residual strategy $\s_\A^{\rho \cdot \pi}$ is not optimal
		from the last state of $\rho \cdot \pi$, that is the Player $\A$ strategy
		$\s_\A$ is not subgame optimal.
		
		Let us now show that these conditions are sufficient. Let $\rho \in Q^+$. Let
		us show that the residual strategy $\s_\A^\rho$ is optimal from $\rho$. Consider
		a Player $\B$ deterministic strategy $\s_\B$. We have, by
		Lemma~\ref{lem:locally_opt_implies_same_value}:
		\begin{align*}
		\mathbb{P}^{\s_\A^\rho,\s_\B^\rho}_\rho[W] & =
		\mathbb{P}^{\s_\A^\rho,\s_\B^\rho}_\rho[W \cap ( \bigcup_{u \in V_{\G}} Q^+
		\cdot (Q_u)^\omega)] = \sum_{u \in V_{\G}}
		\mathbb{P}^{\s_\A^\rho,\s_\B^\rho}_\rho[W \cap Q^+ \cdot (Q_u)^\omega] \\
		& \geq \sum_{u \in V_{\G} \setminus \{ 0 \}} u \cdot
		\mathbb{P}^{\s_\A^\rho,\s_\B^\rho}_\rho[W \cap Q^+ \cdot (Q_u)^\omega] 
		\\
		& \geq \sum_{u \in V_{\G} \setminus \{ 0 \}} u \cdot
		\mathbb{P}^{\s_\A^\rho,\s_\B^\rho}_\rho[Q^+ \cdot (Q_u)^\omega] + 0 \cdot
		\mathbb{P}^{\s_\A^\rho,\s_\B^\rho}_\rho[Q^+ \cdot (Q_0)^\omega]\\
		& = \sum_{u \in V_{\G}} u \cdot \mathbb{P}^{\s_\A^\rho,\s_\B^\rho}_\rho[Q^+
		\cdot (Q_u)^\omega]  \geq \MarVal{\G}(q)					
		\end{align*}
		Where the last inequality comes from
		Lemma~\ref{lem:locally_opt_implies_convex_comb_ok}.	It follows that the Player
		$\A$ residual strategy $\s_\A^\rho$ is optimal from $\rho$
		.
	\end{proof}
	
	\section{Proof from Section~\ref{sec:transfer}}
	\label{proof:thm_transfer_almost_sure_optimal}
	
	\subsection{(Positively) optimal strategy that is locally optimal but not
		subgame optimal}
	\label{appen:exemple_opt_loc_opt_not_sub_game_opt}
	\begin{figure}[t]
		\centering
		\includegraphics[scale=1]{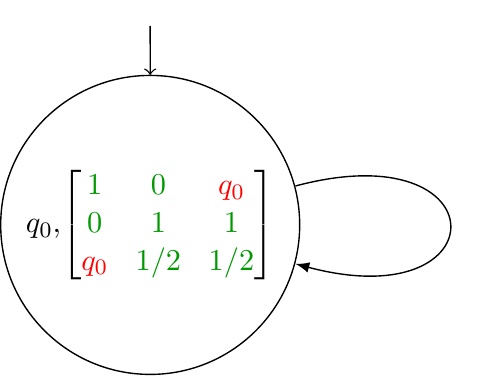}
		\caption{A reachability game.}
		\label{fig:reach_opt_not_sub_game_opt}
	\end{figure}
	
	Consider the game of Figure~\ref{fig:reach_opt_not_sub_game_opt}: it is a
	reachability game, that is if it loops indefinitely on $q_0$, the value is 0.
	The value of the state $q_0$ is $1/2$, it is achieved by a Player $\A$
	positional strategy playing the two top rows with probability $1/2$ and by a
	Player $\B$ positional strategy playing the two columns with probability $1/2$. 
	
	However, denoting $a_1,a_2$ and $a_3$ the three actions available to Player
	$\A$ at state $q_0$ from top to bottom, consider the following Player $\A$
	strategy $\s_\A$ such that $\s_\A(q_0)(a_1) = \s_\A(q_0)(a_2) := 1/2$ and
	$\s_\A(q_0^{n+1})(a_3) := 1$ for all $n \geq 1$. Then, this strategy is locally
	optimal and it is optimal. Indeed, if the game loops at least once on $q_0$,
	then there was the same probability to loop on $q_0$ and to reach outcome 1.
	Hence, the mean of the values is at least $1/2$ which is the value of the state
	$q_0$. However, it is not subgame optimal since after the game loops once on
	$q_0$, then Player $\B$ can ensure value 0 by playing indefinitely the left
	column with probability 1. 
	
	Similarly, consider a Player $\A$ strategy $\s_\A$ such that $\s_\A(q_0)(a_1) =
	\s_\A(q_0)(a_2) := 1/4$, $\s_\A(q_0)(a_3) := 1/2$ and $\s_\A(q_0^{n+1})(a_3) :=
	1$ for all $n \geq 1$. Then, this strategy is locally optimal and
	positively-optimal, however it is not optimal since the value of this strategy
	is $1/4$, which can be enforced by Player $\B$ by playing indefinitely the left
	column with probability 1.
	
	\subsection{Formal definition of the new games}
	\label{subsec:game_value_one_everywhere}
	\subsubsection{Game where all states have value 1}
	\begin{proposition}
		\label{prop:G_W}
		Consider a PI objective $\emptyset \subsetneq W \subseteq \colSet^\omega$%
		. Then, there is a concurrent game $\G_W = \langle \Aconc_W,W \rangle$ where
		all states have value $1$.
	\end{proposition}
	\begin{proof}
		We consider the concurrent arena $\Aconc_W = \AparamConc$ where Player $\A$
		can choose alone the sequence of states seen among states of colors in
		$\colSet$, that is:
		\begin{itemize}
			\item $Q := \colSet$, for all $q \in Q$, we have $A_q := \colSet$ and $B_q :=
			\{ b \}$;
			\item $\distribSet := \{ d_q \mid q \in Q \}$ and for all $q \in Q$, we have
			$\distribFunc(d_q)(q) = 1$;
			\item For all $q \in Q$ and $q' \in A_q = \colSet$, we have $\delta(q,q',b)
			:= d_{q'}$;
			\item for all $q \in Q$, we have $\colFunc(q) := q$.
		\end{itemize}
		This game is in fact turn-based since only Player $\A$'s actions affect the
		sequence of states. She can decide exactly the sequence of states -- and colors
		-- that is seen. Hence, she can follow exactly a path $\rho \in W$ (except for
		the first color, which does not matter since $W$ is prefix-independent), hence
		the values of states are 1.
	\end{proof}
	
	\subsubsection{The definition of the game $\G_u$}
	\begin{definition}
		\label{def:game_g_u}
		For a positive value $u \in V_{\G} \setminus \{ 0 \}$, we define the game
		$\G_u = \Games{\Aconc_u}{W}$ with $\Aconc_u = \langle Q'_u,(A'_q)_{q \in
			Q'_u},(B'_q)_{q \in Q'_u},\distribSet,\delta',\distribFunc',\colSet,\colFunc
		\rangle$ where we refer to the arena $\Aconc_W$ as a sink state to which there
		are some transitions:
		\begin{itemize}
			\item $Q := Q_u \cup \Aconc_W$;
			\item For all $q \in Q_u$, we have $A'_q := A_{\formN_q}$ and $B'_q := B_q$;
			\item $\distribSet' := \{ d_{q,\sigma_q,b} \mid q \in Q_u,\; \sigma_q \in
			A_q, b \in B \}$;
			\item For all $q \in Q_u$ and $\sigma_q \in A_q = S_{\formN_q}$ we have:
			$\delta'(q,\sigma_q,b) := d_{q,\sigma_q,b}$;
			\item For all $q \in Q_u,\; \sigma_q = \sum_{a \in A} \sigma_q(a) \cdot
			\mathds{1}_{\{ a \}} \in A'_q$, $b \in B'_q$ and $q' \in Q$, we have:
			\begin{displaymath}
			\distribFunc'(d_{q,\sigma_q,b})(q') := 
			\begin{cases}
			\sum_{a \in A} \sigma_q(a) \cdot \distribFunc(\delta(q,a,b))(q') & \text{ if
			} q' \in Q_u\\
			\sum_{a \in A} \sigma_q(a) \cdot \distribFunc(\delta(q,a,b))[Q \setminus Q_u]
			& \text{ if } q' = \Aconc_W\\
			\end{cases}
			\end{displaymath}
		\end{itemize}
	\end{definition}

	

	
	\subsection{Proof of Lemma~\ref{lem:value_one_Gu}}
	\label{proof:lem_value_one_Gu}
	In fact we show a stronger result.
	\begin{lemma}
		Consider some $u \in V_{\G} \setminus \{ 0 \}$ and the game $\G_u$. Consider
		also the objective (that is not prefix-independent) $W_u := W \cup Q^* \cdot (Q
		\setminus Q_u) \cdot Q^\omega \subseteq Q^\omega$ and assume that there exists,
		in $\G$ a locally optimal strategy that is positively-optimal w.r.t. the
		objective $W_u$ from all states in $Q_u$. Then, the values of all states in
		$\G_u$ is positive: $\MarVal{\G_u}(q) > 0$, for all $q$ in $\G_u$.
		\label{lem:value_one_qualitatively_Gu}
	\end{lemma}
	The proof of this lemma is quite straightforward but it is quite long. It is
	due to how the game $\G_u$ is formally defined.
	\begin{proof}
		First, note that all the values of all states in $\G_u$ w.r.t. the objective
		$W_u$ are at least $u$ since $W \subseteq W_u$. Let us now consider such a
		locally optimal strategy $\s_\A$ that is positively-optimal w.r.t. the objective
		$W_u$ from all states in $Q_u$ and let $r = \min_{q \in Q_u}
		\MarVal{\Games{\Aconc}{W_u}}[\s_\A](q) > 0$ be the minimum of the values of the
		strategy $\s_\A$ in $Q_u$ w.r.t. the objective $Q_u$. 
		Let us assume that one state $q \in Q_u$ has value 0 in the game $\G_u$
		(w.r.t. to the objective $W$) and consider a Player $\B$ strategy $\s_\B$ whose
		value in the game $\G_u$ is at most $r/2 > 0$ from the state $q$: i.e.
		$\MarVal{\G_u}[\s_\B](q) \leq r/2$ w.r.t. $W$. Note that the Player $\B$
		strategy $\s_\B$ is a strategy in $\Aconc$ and in $\Aconc_u$. 
		For all $\rho = \rho' \cdot q \in (Q_u)^+$, we can write $\s_\A(\rho)$ as a
		convex combination of elements in $A'_q$: $\s_\A(\rho) = \sum_{\sigma \in A'_q}
		\mu_{\sigma}(\rho) \cdot \sigma$. Let us now define the Player $\A$ strategy
		$\s_\A'$ in the game $\G_u$ that mimics the strategy $\s_\A$, i.e. for all $\rho
		= \rho' \cdot q \in (Q_u)^+$ and $\sigma \in A'_q$, we set:
		$\s_\A'(\rho)(\sigma) := \mu_{\sigma}(\rho)$. When the game enters $\Aconc_W$,
		the strategy $\s_\A'$ switches to a sure strategy. Hence, since the objective
		$W$ is prefix-independent, for all $\rho \in (Q_u)^* \cdot \Aconc_W$, we have
		$\MarVal{\G_u}(\s_\A')(\rho) = 1$. We denote by $\mathbb{P}^{\s_\A,\s_\B}$ the
		probability measure induced by strategies $\s_\A,\s_\B$ in $\G$ from state $q$
		and by $\mathbb{P}^{\s_\A',\s_\B,u}$ the probability measure induced by
		$\s_\A',\s_\B$ in $\G_u$ from state $q$. 
		
		Let us now show by induction on $n$ the property $\mathcal{P}(n)$:
		for all $\rho \in (Q_u)^n$, we have $\mathbb{P}^{\s_\A,\s_\B}(\rho) =
		\mathbb{P}^{\s_\A',\s_\B,u}(\rho)$ and $\mathbb{P}^{\s_\A,\s_\B}(Q_u^{\leq n}
		\cdot (Q \setminus Q_u)) = \mathbb{P}^{\s_\A',\s_\B,u}(Q_u^{\leq n} \cdot
		\Aconc_W)$. This straightforwardly holds for $n = 1$. Assume now that this holds
		for some $n-1 \geq 0$. Consider some $\rho = \rho' \cdot q = \rho'' \cdot q'
		\cdot q \in (Q_u)^n \cdot Q$. We have:
		\begin{align*}
		\mathbb{P}^{\s_\A,\s_\B}(\rho) & = \mathbb{P}^{\s_\A,\s_\B}(\rho') \cdot
		\mathbb{P}^{\s_\A,\s_\B}_{\rho'}(q) \\
		& = \mathbb{P}^{\s_\A,\s_\B,u}(\rho') \cdot \left( \sum_{a \in A} \sum_{b \in
			B} \s_\A(\rho')(a) \cdot \s_\B(\rho')(b) \cdot \distribFunc(\delta(q',a,b))(q)
		\right) \\
		& = \mathbb{P}^{\s_\A,\s_\B,u}(\rho') \cdot \sum_{b \in B}  \s_\B(\rho')(b)
		\cdot \left( \sum_{a \in A} \s_\A(\rho')(a) \cdot 
		\distribFunc(\delta(q',a,b))(q) \right) \\
		& = \mathbb{P}^{\s_\A,\s_\B,u}(\rho') \cdot \sum_{b \in B}  \s_\B(\rho')(b)
		\cdot \left( \sum_{a \in A} \sum_{\sigma \in A'_q} \mu_{\sigma}(\rho') \cdot
		\sigma(a) \cdot  \distribFunc(\delta(q',a,b))(q) \right) \\
		& = \mathbb{P}^{\s_\A,\s_\B,u}(\rho') \cdot \sum_{b \in B}  \s_\B(\rho')(b)
		\cdot \left( \sum_{\sigma \in A'_q} \mu_{\sigma}(\rho') \cdot \sum_{a \in A}
		\sigma(a) \cdot  \distribFunc(\delta(q',a,b))(q) \right) 
		\end{align*}
		If $q \in Q_u$, we obtain:
		\begin{align*}
		\mathbb{P}^{\s_\A,\s_\B}(\rho) & = \mathbb{P}^{\s_\A,\s_\B,u}(\rho') \cdot
		\sum_{b \in B}  \s_\B(\rho')(b) \cdot \sum_{\sigma \in A'_q} \mu_{\sigma}(\rho')
		\cdot \distribFunc'(d_{q'},\sigma_q,b)(q) \\
		& = \mathbb{P}^{\s_\A,\s_\B,u}(\rho') \cdot \sum_{b \in B}  \s_\B(\rho')(b)
		\cdot \sum_{\sigma \in A'_q} \s_\A'(\rho')(\sigma) \cdot
		\distribFunc'(d_{q'},\sigma,b)(q) \\
		& = \mathbb{P}^{\s_\A,\s_\B,u}(\rho') \cdot \sum_{\sigma \in A'_q}  \sum_{b
			\in B} \s_\A'(\rho')(\sigma) \cdot \s_\B(\rho')(b) \cdot
		\distribFunc'(\delta'(q',\sigma,b))(q) \\
		& = \mathbb{P}^{\s_\A,\s_\B,u}(\rho') \cdot
		\mathbb{P}^{\s_\A',\s_\B,u}_{\rho'}(q) \\
		& = \mathbb{P}^{\s_\A',\s_\B,u}(\rho)
		\end{align*}
		Furthermore, we have:
		\begin{align*}
		\mathbb{P}^{\s_\A,\s_\B}[Q_u^n \cdot (Q \setminus Q_u)] & = \sum_{\rho' \in
			(Q_u)^n} \mathbb{P}^{\s_\A,\s_\B}(\rho') \cdot
		\mathbb{P}^{\s_\A,\s_\B}_{\rho'}[Q \setminus Q_u] \\
		& = \sum_{\rho' \in (Q_u)^n} \mathbb{P}^{\s_\A,\s_\B}(\rho') \cdot \sum_{q \in
			Q \setminus Q_u} \mathbb{P}^{\s_\A,\s_\B}_{\rho'}(q) \\
		& = \sum_{\rho' \in (Q_u)^n} \mathbb{P}^{\s_\A,\s_\B}(\rho') \cdot \sum_{q \in
			Q \setminus Q_u} \sum_{b \in B}  \s_\B(\rho')(b) \cdot \left( \sum_{\sigma \in
			A_q} \mu_{\sigma}(\rho') \cdot \sum_{a \in A} \sigma(a) \cdot 
		\distribFunc(\delta(q',a,b))(q) \right) \\
		& = \sum_{\rho' \in (Q_u)^n} \mathbb{P}^{\s_\A,\s_\B}(\rho') \cdot \sum_{b \in
			B}  \s_\B(\rho')(b) \cdot \left( \sum_{\sigma \in A_q} \mu_{\sigma}(\rho') \cdot
		\sum_{a \in A} \sigma(a) \cdot  \sum_{q \in Q \setminus Q_u} 
		\distribFunc(\delta(q',a,b))(q) \right) \\
		& = \sum_{\rho' \in (Q_u)^n} \mathbb{P}^{\s_\A,\s_\B}(\rho') \cdot \sum_{b \in
			B}  \s_\B(\rho')(b) \cdot \left( \sum_{\sigma \in A_q} \mu_{\sigma}(\rho') \cdot
		\sum_{a \in A} \sigma(a) \cdot \distribFunc(\delta(q',a,b))[Q \setminus Q_u]
		\right) \\
		& = \sum_{\rho' \in (Q_u)^n} \mathbb{P}^{\s_\A,\s_\B}(\rho') \cdot \sum_{b \in
			B}  \s_\B(\rho')(b) \cdot \left( \sum_{\sigma \in A_q} \mu_{\sigma}(\rho') \cdot
		\distribFunc'(d_{q',\sigma,b})[\Aconc_W] \right) \\
		& = \sum_{\rho' \in (Q_u)^n} \mathbb{P}^{\s_\A,\s_\B}(\rho') \cdot
		\sum_{\sigma \in A'_q} \sum_{b \in B} \s_\A'(\rho')(\sigma) \cdot
		\s_\B(\rho')(b) \cdot \distribFunc'(\delta(q',\sigma,b))[\Aconc_W]  \\
		& = \sum_{\rho' \in (Q_u)^n} \mathbb{P}^{\s_\A,\s_\B}(\rho') \cdot
		\mathbb{P}^{\s_\A',\s_\B,u}_{\rho'}[\Aconc_W]  \\
		& = \mathbb{P}^{\s_\A',\s_\B,u}[Q_u^n \cdot \Aconc_W]
		\end{align*}
		
		It follows that:
		\begin{align*}
		\mathbb{P}^{\s_\A,\s_\B}[(Q_u)^{\leq n} \cdot (Q \setminus Q_u)] & =
		\mathbb{P}^{\s_\A,\s_\B}[(Q_u)^{\leq n-1} \cdot (Q \setminus Q_u)] +
		\mathbb{P}^{\s_\A,\s_\B}[(Q_u)^n \cdot (Q \setminus Q_u)] \\
		& = \mathbb{P}^{\s_\A',\s_\B,u}[(Q_u)^{\leq n-1} \Aconc_W] +
		\mathbb{P}^{\s_\A',\s_\B,u}[(Q_u)^n \cdot \Aconc_W] \\
		& = \mathbb{P}^{\s_\A',\s_\B,u}[(Q_u)^{\leq n} \cdot \Aconc_W]
		\end{align*}	
		Hence, $\mathcal{P}(n)$ holds. In fact, this property holds for all $n \in
		\N$. From this, we obtain:
		\begin{itemize}
			\item $\mathbb{P}^{\s_\A,\s_\B}[(Q_u)^* \cdot (Q \setminus Q_u)] =
			\mathbb{P}^{\s_\A',\s_\B,u}[(Q_u)^* \cdot \Aconc_W]$;
			\item $\mathbb{P}^{\s_\A,\s_\B}[W \cap (Q_u)^\omega] =
			\mathbb{P}^{\s_\A',\s_\B,u}[W \cap (Q_u)^\omega]$.
		\end{itemize}
		That is:
		\begin{align*}
		\mathbb{P}^{\s_\A,\s_\B}[W_u] & = \mathbb{P}^{\s_\A,\s_\B}[W_u \cap
		(Q_u)^\omega] + \mathbb{P}^{\s_\A,\s_\B}[W_u \cap (Q_u)^* \cdot (Q \setminus
		Q_u)] \\ 
		& = \mathbb{P}^{\s_\A,\s_\B}[W \cap (Q_u)^\omega] +
		\mathbb{P}^{\s_\A,\s_\B}[(Q_u)^* \cdot (Q \setminus Q_u)] \\ 
		& = \mathbb{P}^{\s_\A',\s_\B,u}[W \cap (Q_u)^\omega] +
		\mathbb{P}^{\s_\A',\s_\B,u}[(Q_u)^* \cdot \Aconc_W] \\
		& = \mathbb{P}^{\s_\A',\s_\B,u}[W \cap (Q_u)^\omega] +
		\mathbb{P}^{\s_\A',\s_\B,u}[W \cap (Q_u)^* \cdot \Aconc_W] \\
		& = \mathbb{P}^{\s_\A',\s_\B,u}[W] \leq r/2 < r
		\end{align*}
		The inequality follow by the choice of the Player $\B$ strategy $\s_\B$. This
		result is in contradiction with the definition of $r$: the minimum of the values
		of the states in $Q_u$ w.r.t. the strategy $\s_\A$ and the objective $W_u$. In
		fact, there is no such state $q \in Q_u$ whose value in $\G_u$ is 0. 
		The lemma follows.
	\end{proof}
	
	\subsection{Proof of Lemma~\ref{prop:glued_strategy}}
	\label{proof:prop_glued_strategy}
	\begin{proof}
		We denote the strategy $\s_\A[(\s_\A^u)_{u \in V_{\G} \setminus \{ 0 \}}]$ by
		$\s_\A$ and we apply Theorem~\ref{thm:subgame_optimal_arbitrary_strategy} to
		show that it is subgame optimal. The strategy $\s_\A$ is locally optimal at all
		states in $Q_u$ for all positive $u \in V_{\G} \setminus \{ 0 \}$, by definition
		of the games $\G_u$. Furthermore, since the values of all states $q \in Q_0$ is
		0, the values of the game in normal form in $Q_0$ are 0. Hence, for all state $q
		\in Q_0$, any strategy in the game in normal form $\formNF_q$ is optimal in game
		in normal form $\formNF_q$. That is, the strategy $\s_\A$ is locally optimal.
		
		Let us now show that it ensures the second property of
		Theorem~\ref{thm:subgame_optimal_arbitrary_strategy}. Let $u \in V_{\G}
		\setminus \{ 0 \}$. Consider a deterministic Player $\B$ strategy $\s_\B$.
		Consider a finite path $\rho \in Q^+$ and $\pi = \pi' \cdot q \in Q^* \cdot
		Q_u$. Assume that $\mathbb{P}_{\rho \cdot \pi}^{\s_\A,\s_\B}((Q_u)^\omega) > 0$.
		Then, we have $\mathbb{P}_{\rho \cdot \pi}^{\s_\A,\s_\B}(W \mid (Q_u)^\omega) =
		1$ since the strategy $\s_\A$ then behaves like $\s_\A^u$ and this holds for
		$\s_\A^u$ (as it is subgame almost-surely winning in $\G_u$ and the objective
		$W$ is prefix-independent). That is, $\mathbb{P}_{\rho \cdot
			\pi}^{\s_\A,\s_\B}(W \cap (Q_u)^\omega) = \mathbb{P}_{\rho \cdot
			\pi}^{\s_\A,\s_\B}((Q_u)^\omega)$. As this holds for all $\pi \in Q^* \cdot
		Q_u$, it follows that $\mathbb{P}_{\rho}^{\s_\A,\s_\B}(W \cap Q^* \cdot
		(Q_u)^\omega) = \mathbb{P}_{\rho}^{\s_\A,\s_\B}(Q^* \cdot (Q_u)^\omega)$. 
		
		We can then apply Theorem~\ref{thm:subgame_optimal_arbitrary_strategy} to
		obtain the theorem.
	\end{proof}

	\subsection{Extension to finite memory}
	\label{appen:extension_finite_memory}
	We first recall what is a memory skeleton (see, for
	instance,~\cite{DBLP:conf/concur/Bouyer0ORV20}) and we see how it
	can implement finite-memory strategies. For a set of colors $\colSet$
	and a set of states $Q$, a memory skeleton on $\colSet$ is a triple
	$\mathcal{M} = \langle M,\minit,\mu \rangle$, where $M$ is a non-empty
	set called the memory, $\minit \in M$ is the initial state of the
	memory and $\mu: M \times \colSet \rightarrow M$ is the update
	function. Note that the update function $\mu$ can be extended inductively into
	a function $\mu^*: M \times \colSet^* \rightarrow M$ in the following way: for
	all $m \in M$, $\mu^*(m,\epsilon) := m$ and for all $\rho \cdot k \in
	\colSet^+$, $\mu^*(m,\rho \cdot k) := \mu(\mu^*(m,\rho),k)$. Given some set of
	actions $A_q$ for each $q \in Q$, an action map with memory $M$ is a function
	$\lambda: M \times Q \rightarrow \Dist(\cup_{q \in Q} A_q)$  such that for all
	$q \in Q$ and $m \in M$ we have $\lambda(m,q) \in \Dist(A_q)$. In fact, a memory
	skeleton and an action map implement a strategy.
	
	\begin{definition}[Implementation of finite memory strategies]
		Consider a concurrent arena $\Aconc$ on a set of colors $\colSet$. A memory
		skeleton	$\mathcal{M} = \langle M,\minit,\mu \rangle$ on $\colSet$ and an action
		map $\lambda: M \times Q \rightarrow \Dist(\cup_{q \in Q} A_q)$ implement the
		strategy $\s_\A: Q^+ \rightarrow \Dist(\cup_{q \in A_q} A_q)$ that is
		defined by, for all $\rho = \rho' \cdot q \in Q^+$, $\s_\A(\rho) :=
		\lambda(\mu^*(\minit,\colFunc(\rho)),q) \in \Dist(A_q)$.
		
		A strategy $\s_\A$ is \emph{finite memory} if there exists a memory
		skeleton $\mathcal{M} = \langle M,\minit,\mu \rangle$, with $M$
		finite, and an action map $\lambda$ implementing $\s_\A$. If $M$ is a
		singleton, we retrieve the definition of positional strategies.
	\end{definition}
	
	We can extend the definition positionaly almost-surely winnable objective to
	objective winnable with a fixed memory skeleton.
	\begin{definition}[$\mathcal{M}$-almost-surely winnable objective]
		Consider an Borel prefix-independent objective $W \subseteq \colSet^\omega$ on
		a set of colors $\colSet$ and a memory skeleton $\mathcal{M}$. It is
		$\mathcal{M}$-\emph{almost-surely winnable} ($\mathcal{M}$-\textsf{ASW} for
		short) if in all finite games $\G = \Games{\Aconc}{W}$ on the set of colors
		$\colSet$, if there is a subgame almost-surely winning strategy, then there is
		one that can be implemented with $\mathcal{M}$.
	\end{definition}
	
	We now adapt Theorem~\ref{thm:transfer_memory} to the transfer of
	finite-memory.
	\begin{theorem}
		Consider a non-empty finite set of colors $K$ and a PI objective $\emptyset
		\subsetneq W \subseteq \colSet^\omega$. Consider a PI concurrent game $\G$ with
		objective $W$. For some memory skeleton $\mathcal{M}$, if the objective $W$ is
		$\mathcal{M}$-\textsf{ASW} and if there is a subgame optimal strategy in $\G$
		then there is one that can be implemented by $\mathcal{M}$.
		\label{thm:transfer_finite_memory}
	\end{theorem}
	
	Let us define the set of relevant memory states of a given memory skeleton
	$\mathcal{M}$, that is the subset of memory states that could occur from
	$\minit$ given any finite sequence of colors. That is:
	\begin{definition}[Relevant memory states]
		Consider a memory skeleton $\mathcal{M} = \langle M,\minit,\mu \rangle$ on a
		set of colors $\colSet$. The set of \emph{relevant memory states} $M_r \subseteq
		M$ is defined as $M_r := \{ m \in M \mid \exists \rho \in \colSet^*,\; m =
		\mu^*(\minit,\rho) \}$. For all such relevant states $m \in M_r$, we denote by
		$\mathcal{M}^m$ the memory skeleton $\mathcal{M}^m := \langle M,m,\mu \rangle$.
	\end{definition}
	
	Then, we have the following proposition:
	\begin{proposition}
		Consider a set of colors $\colSet$, a memory skeleton $\mathcal{M} = \langle
		M,\minit,\mu \rangle$ on $\colSet$ and the corresponding set of relevant memory
		states $M_r$. Consider also a $\mathcal{M}$-\textsf{ASW} objective $W \subseteq
		K^\omega$. Then, for all finite games $\G = \Games{\Aconc}{W}$ on the set of
		colors $\colSet$ where there is a subgame almost-surely winning strategy, there
		is an action map $\lambda: M \times Q \rightarrow \Dist(\cup_{q \in Q} A_q)$
		such that, for all relevant states $m \in M_r$, the strategy implemented by
		$\mathcal{M}^m$ and $\lambda$ is subgame almost-surely winning in $\G$.
		\label{prop:sub_game_optimal_from_every_memory_state}
	\end{proposition}
	\begin{proof}
		Consider a set of colors $\colSet$, a memory skeleton $\mathcal{M} = \langle
		M,\minit,\mu \rangle$ on $\colSet$ and the corresponding set of relevant memory
		states $M_r$. Consider also a $\mathcal{M}$-\textsf{ASW} objectives $W \subseteq
		K^\omega$. Consider also a finite game $\G = \Games{\Aconc}{W}$ on the set of
		colors $\colSet$ where there is a subgame almost-surely winning strategy. Let us
		build a game $\G^\B$ where Player $\B$ first plays for finite number of rounds
		to 'choose' the memory state of the Player $\A$ strategy, and then the game
		proceeds as in $\G$. More specifically, for $n := |M_r|$ and denoting $Q$ the
		set of states in $\G$, we consider the finite arena $\Aconc^\B = \langle
		Q^\B,(A_q)_{q \in Q^\B},(B_q)_{q \in
			Q^\B},\distribSet,\delta,\distribFunc,\colSet,\colFunc \rangle$ defined in the
		following way:
		\begin{itemize}
			\item $Q^\B := \colSet \times \{ 0,\ldots,n \}$ the set of states where only
			Player $\B$ will play;
			\item For all $q = (k,i) \in Q^\B$, we have $A'_q := \{ a \}$ for an
			arbitrary action $a$ and $B_q := \mathsf{Cont}_q \uplus \mathsf{Exit}_q$ with
			$\mathsf{Cont}_q := \{ b_{k,i+1} \mid k \in \colSet \}$ if $i < n$ and
			$\mathsf{Cont}_q := \emptyset$ otherwise. Furthermore, $\mathsf{Exit}_q := \{
			b_{q'} \mid q' \in Q \}$;
			\item $\distribSet := \{ d_{k,i+1} \mid k \in \colSet,\; 0 \leq i \leq n-1 \}
			\uplus \{ d_{q'} \mid q' \in Q \}$;
			\item For all $q = (k,i) \in Q^\B$ and $k \in \colSet$ we have:
			$\delta(q,a,b_{k,i+1}) := d_{k,i+1}$ and for all $q' \in Q$, we have
			$\delta(q,a,b_{q'}) := d_{q'}$;
			\item For all $k \in \colSet$ and $0 \leq i \leq n-1$, we have
			$\distribFunc(d_{k,i+1})((k,i+1)) := 1$ and for all $q \in Q$
			$\distribFunc(d_{q})(q) := 1$;
			\item For all $k \in \colSet$ and $0 \leq i \leq n$, we have $\colFunc((k,i))
			:= k$.
		\end{itemize}
		The arena $\Aconc^\B$ is plugged in before the arena $\Aconc$ (thus obtaining
		the arena $\Aconc^\B \cdot \Aconc$) and eventually enters this arena $\Aconc$
		via the Nature states $d_{q}$ for $q \in Q$ (which have to be chosen at some
		point: there are the only successors of the states $(k,n)$ for all $k \in
		\colSet$). We claim that there is still a subgame almost-surely winning strategy
		in the game $\G^\B := \langle \Aconc^\B \cdot \Aconc,W \rangle$. This is due to
		the fact that it was the case of the game $\G$,that the objective $W$ is
		prefix-independent and that the gae can only stay in $\Aconc^\B$ for a finite
		number of rounds. Hence, since the objective $W$ is
		$\mathcal{M}$-$\mathsf{ASW}$, there is an action map $\lambda: M \times Q^\B
		\cup Q \rightarrow \Dist(\cup_{q \in Q^\B \cup Q} A_q)$ such that the strategy
		$\s_\A$ implemented by $\mathsf{M}$ and $\lambda$ is subgame almost-surely
		winning in $\G^B$. In particular, Player $\B$ can play in the arena $\Aconc^\B$
		such that, when she leaves it to 'start the game $\G$', it can be done with the
		memory state of the Player $\A$ strategy $\s_\A$ is an arbitrary relevant memory
		state $m \in M_r$. (Note that indeed Player $\B$ can reach any relevant memory
		state because the memory skeleton $\langle M_r,\minit,\mu \rangle$ can be seen
		as a strongly connected graph of size $|M_r| = n$.) That is, for all relevant
		memory states $m \in M_r$ such that there exists a finite sequence $\rho \in
		\colSet^k$ for $k \leq n+1$ of colors such that $m = \mu^*(\minit,\rho)$, the
		strategy implemented by $\mathcal{M}^m$ and $\lambda$ is in fact the residual
		strategy $\s_\A^{(\rho_0,0) \cdots (\rho_k,k)}$, and it is also subgame optimal.
		
	\end{proof}
	
	We can now proceed to the proof of Theorem~\ref{thm:transfer_finite_memory}.
	\begin{proof}
		Consider a non-empty finite set of colors $K$ and a PI objective $\emptyset
		\subsetneq W \subseteq \colSet^\omega$. Consider a concurrent game $\G$ with
		objective $W$. Assume that the objective $W$ is $\mathcal{M}$-\textsf{ASW} for
		some memory skeleton $\mathcal{M}$ and that there is a subgame optimal strategy
		in $\G$.
		We consider the same construction than for the proof of
		Theorem~\ref{thm:transfer_memory}. However, for all $u \in V_{\G} \setminus \{ 0
		\}$, since there is a subgame almost-surely winning strategy in the game $\G_u$,
		by Proposition~\ref{prop:sub_game_optimal_from_every_memory_state}, there is an
		action map $\lambda_u: M \times Q_u \rightarrow \Dist(\cup_{q \in Q_u} A_q)$
		such that, for all relevant states $m \in M_r$, the strategy implemented by
		$\mathcal{M}^m$ and $\lambda$ is subgame almost-surely winning in $\G_u$. We
		then glue these strategies into the an action map $\lambda: M \times Q
		\rightarrow \Dist(\cup_{q \in Q} A_q)$ such that, for all $\rho = \rho' \cdot q
		\in Q^+$, $\lambda(m,q) := \lambda_{\MarVal{\G}(q)}(m,q)$, i.e. the action map
		chosen depends on the value area. We then claim that the strategy $\s_\A$
		implemented by the memory skeleton $\mathcal{M}$ and the action map $\lambda$ is
		subgame optimal. As for the proof of Lemma~\ref{prop:glued_strategy}, this
		strategy is locally optimal. Furthermore, again as for the proof of
		Lemma~\ref{prop:glued_strategy}, let us show that it ensures the second property
		of Theorem~\ref{thm:subgame_optimal_arbitrary_strategy}. Let $u \in V_{\G}
		\setminus \{ 0 \}$. Consider a deterministic Player $\B$ strategy $\s_\B$.
		Consider a finite path $\rho \in Q^+$ and $\pi = \pi' \cdot q \in Q^* \cdot
		Q_u$. Assume that $\mathbb{P}_{\rho \cdot \pi}^{\s_\A,\s_\B}((Q_u)^\omega) > 0$.
		Then, we have $\mathbb{P}_{\rho \cdot \pi}^{\s_\A,\s_\B}(W \mid (Q_u)^\omega) =
		1$ since the strategy $\s_\A$ then behaves like $\s_\A^u$ (implemented by
		$\mathcal{M}$ and $\lambda_u$) and this holds for $\s_\A^u$ (as it is subgame
		almost-surely winning in $\G_u$ regardless of the starting relevant memory state
		and since the objective $W$ is prefix-independent). 
	\end{proof}
	
	\subsection{Retrieving the original result from \cite{GIH10}}
	\label{appen:retrieve_result_turn_based}

	Let us first show that 
	positively-optimal strategies that are locally optimal always exists (for a
	slightly modified objective) in a turn-based setting.
	\begin{proposition}
		Consider a turn-based game $\G$ (i.e. it is such that the set of states can be
		partitioned into Player $\A$ states $q$ where $B_q$ is a singleton and Player
		$\B$ states $q$ where $A_q$ is a singleton). Consider the valuation $v: Q
		\rightarrow [0,1]$ giving the value of the game from all states w.r.t. action
		strategies. Consider some value $u \in V_\G \setminus \{ 0 \}$ (w.r.t. the
		valuation $v$) and consider the objective (first defined in
		Lemma~\ref{lem:value_one_qualitatively_Gu}) $W_u := W \cup Q^* \cdot (Q
		\setminus Q_u) \cdot Q^\omega \subseteq Q^\omega$ (where $Q_u$ is defined w.r.t.
		the valuation $v$). Then, there is a locally optimal action strategy whose value
		is positive from all states in $Q_u$.
	\end{proposition}
	\begin{proof}
		Consider a turn-based game $\G$ and let $Q = Q_\A \uplus Q_\B$ where for all
		$q \in Q_\A$ (resp. $Q_\B$) we have $B_q$ (resp. $A_q$) a singleton. Consider
		some value $u \in V_\G \setminus \{ 0 \}$ and assume towards a contradiction
		that there is a state $q \in Q_u$ from which there is no locally optimal action
		strategy whose value is positive from $q$. Let $\distribSet_{< u} \subseteq
		\distribSet$ be the set of Nature states whose values is less than $u$ and let
		$x < u$ be the highest values of these Nature states. Consider the game $\G' =
		\langle \Aconc',W_u \rangle$ where the arena $\Aconc'$ is a restriction of the
		arena $\Aconc$ where the Player $\A$ actions leading to Nature states in
		$\distribSet_{< u}$ are removed (i.e. Player $\A$ can only play locally optimal
		strategies at states in $Q_u$). By assumption, there is a state $q \in Q_u$
		whose value is 0 in that game $\G'$. Let $\s_\B$ be a Player $\B$ strategy in
		that game $\G'$ whose value is $(u-x)/4 > 0$ from $q$. Now consider a Player
		$\A$ deterministic strategy $\s_\A$ (recall that in turn-based games
		deterministic strategies acheive the same values than arbitrary strategies, see
		for instance the explanations in Theorem 1 from
		\cite{DBLP:journals/iandc/Chatterjee0GH15}). Let us denote by $\mathsf{NLO}$
		(for non-locally-optimal) the set of finite paths ending in $Q_u$ after which
		the Player $\A$ deterministic strategy $\s_\A$ is not locally optimal:
		$\mathsf{NLO} := \{ \rho = \rho' \cdot q' \in Q^* \cdot Q_u \mid
		\delta(q,\s_\A(\rho),b) \in \distribSet_{< u} \}$ for $b$ the only Player $\B$
		action in $B_{q'}$. We can now define the Player $\B$ strategy $\s_\B'$ in the
		following way: it mimics the strategy $\s_\B$ while the play stays in $Q_u$ and
		no finite paths in $\mathsf{NLO}$ has been reached, when such a path is reached,
		$\s_\B'$ switches to strategy of value $(x + u)/2$ (which is possible by
		definition of $\mathsf{NLO}$). Consider now a Player $\A$ strategy that does
		exactly $\s_\A$ until a finite paths in $\mathsf{NLO} \cup Q^* \cdot (Q
		\setminus Q_u)$ is reached and in that case, it switches to an arbitrary locally
		optimal strategy (it is therefore a strategy in the game $\G'$). With these
		choices, $\s_\A$ and $\s_\A'$ coincide on paths in $(Q_u \setminus
		\mathsf{NLO})^*$. The strategies $\s_\B$ and $\s_\B'$ also coincide on these
		paths. Considering what happens with strategies $\s_\A'$ and $\s_\B$ -- recall
		$\s_\B$ has value at most $u-x/4$ against locally optimal strategies for the
		objective $W_u$:
		\begin{align*}
		\mathbb{P}^{\s_\A',\s_\B}_q[W_u] & = \mathbb{P}^{\s_\A,\s_\B}_q[W_u \cap (Q_u
		\setminus \mathsf{NLO})^\omega] + \mathbb{P}^{\s_\A',\s_\B}_q[W_u \cap (Q_u
		\setminus \mathsf{NLO})^* \cdot (Q \setminus Q_u)] +
		\mathbb{P}^{\s_\A',\s_\B}_q[W_u \cap Q_u^* \cdot \mathsf{NLO}] \\
		& = \mathbb{P}^{\s_\A',\s_\B}_q[W \cap (Q_u \setminus \mathsf{NLO})^\omega] +
		\mathbb{P}^{\s_\A',\s_\B}_q[(Q_u \setminus \mathsf{NLO})^* \cdot (Q \setminus
		Q_u)] + \mathbb{P}^{\s_\A',\s_\B}_q[W_u \cap Q_u^* \cdot \mathsf{NLO}] \\
		& = \mathbb{P}^{\s_\A,\s_\B'}_q[W \cap (Q_u \setminus \mathsf{NLO})^\omega] +
		\mathbb{P}^{\s_\A,\s_\B'}_q[(Q_u \setminus \mathsf{NLO})^* \cdot (Q \setminus
		Q_u)] + \mathbb{P}^{\s_\A',\s_\B}_q[W_u \cap Q_u^* \cdot \mathsf{NLO}] \\
		& \leq \frac{u-x}{4}
		\end{align*}
		It follows that $\mathbb{P}^{\s_\A,\s_\B'}_q[W \cap (Q_u \setminus
		\mathsf{NLO})^\omega] + \mathbb{P}^{\s_\A,\s_\B'}_q[(Q_u \setminus
		\mathsf{NLO})^* \cdot (Q \setminus Q_u)] \leq \frac{u-x}{4}$. Furthermore,by
		choice of the strategy $\s_\B'$, we have $\mathbb{P}^{\s_\A,\s_\B'}_q[W \cap
		Q_u^* \cdot \mathsf{NLO}] \leq \frac{u+x}{2}$. We obtain:
		\begin{align*}
		\mathbb{P}^{\s_\A,\s_\B'}_q[W] & = \mathbb{P}^{\s_\A,\s_\B}_q[W \cap (Q_u
		\setminus \mathsf{NLO})^\omega] + \mathbb{P}^{\s_\A,\s_\B'}_q[W \cap (Q_u
		\setminus \mathsf{NLO})^* \cdot (Q \setminus Q_u)] +
		\mathbb{P}^{\s_\A,\s_\B'}_q[W \cap Q_u^* \cdot \mathsf{NLO}] \\
		& \leq \mathbb{P}^{\s_\A,\s_\B}_q[W \cap (Q_u \setminus \mathsf{NLO})^\omega]
		+ \mathbb{P}^{\s_\A,\s_\B'}_q[(Q_u \setminus \mathsf{NLO})^* \cdot (Q \setminus
		Q_u)] + \frac{u+x}{2} \\
		& \leq \frac{u-x}{4} + \frac{u+x}{2} \\
		& = u - \frac{u-x}{4} < u
		\end{align*}
		Hence, all Player $\A$ deterministic strategies have values at most $u -
		\frac{u-x}{4} < u$. This is a contradiction with the fact that the state $q$ has
		value $u$. 
	\end{proof}
	
	We can now deduce that subgame optimal strategies always exist in turn-based
	games. Indeed, it was shown in Lemma~\ref{lem:value_one_qualitatively_Gu} that
	if, for all values $u \in V_{\G} \setminus \{ 0 \}$, there are locally optimal
	strategies whose value is positive from all states in $Q_u$ w.r.t. the objective
	$W_u$ (which we will call assumption $H_u$), then all the states in the games
	$\G_u$ (from the proof of Theorem~\ref{thm:transfer_memory}) is positive. Hence,
	with a straightforward adaptation of the proof of
	Theorem~\ref{thm:transfer_memory}, one can show that this assumption $H_u$
	implies the existence of subgame optimal strategy. It follows that there always
	are subgame optimal strategies in turn-based games. 
	
	Furthermore, note that in that case there are subgame optimal deterministic
	strategies. This is due to fact that
	Theorem~\ref{thm:all_not_zero_implies_almost_sure}, when stated in turn-based
	games, ensures the existence of subgame almost-surely winning deterministic
	strategies. This is already proved in \cite{GIH10}. Our adaptation of the proof
	(which can be found in
	Appendix~\ref{proof:thm_all_not_zero_implies_almost_sure}) could also show it,
	one has just to realize that $\varepsilon$-optimal strategies can be found among
	deterministic strategies in turn-based games.
	
	\section{Proof of Theorem~\ref{thm:all_not_zero_implies_almost_sure}}
	\label{proof:thm_all_not_zero_implies_almost_sure}
	
	
	
	Before proceeding to the proof of
	Theorem~\ref{thm:all_not_zero_implies_almost_sure}, we need some additional
	notations and a very useful theorem. 
	
	\paragraph{Additional Notations}
	Consider a finite set $Q$. For all $\rho \in Q^*$, we denote by $\cyl(\rho)
	\subseteq Q$ the set $\cyl(\rho) := \{ \rho \cdot \pi \mid \pi \in Q^* \}$.
	Then, for a subset $S$ of $Q^*$, we denote by $\mathsf{Suff}(S) \subseteq Q^*$
	the set of suffixes of element of $S$,
	that is $\mathsf{Suff}(S) := \cup_{\rho \in S} \cyl(\rho)$. 
	We say that a set $S \in Q^*$ is \emph{suffix-closed} if $S =
	\mathsf{Suff}(S)$. Furthermore, for all suffix-closed sets $S$ of $Q^+$, there
	is a unique set $\prefUniq{S} \subseteq S$ such that $S$ can be written as the
	following disjoint union: $S = \uplus_{\rho \in \prefUniq{S}} \cyl(\rho)$ (the
	set $\prefUniq{S}$ can be defined as $\prefUniq{S} := \{ \rho \in S \mid \rho
	\text{ has no prefix in } S\}$). 
	Finally, for all finite paths $\rho = q_0 \cdots q_n \in Q^+$, we denote by
	$\Head(\rho) \in Q$ the last state of $\rho$, i.e. $\Head(\rho) := q_n$.
	
	Consider an PI concurrent game $\G = \Games{\Aconc}{W}$, and a Player $\A$
	strategy $\s_\A$. 
	Let us now define, for $u \in [0,1]$, an operator $\mathsf{ValInf}_{u}:
	\SetStrat{\Aconc}{\A} \times \mathcal{P}(Q^+) \rightarrow \mathcal{P}(Q^+)$  on
	strategies that, given a strategy $\s \in \SetStrat{\Aconc}{\A}$ and a set of
	finite paths of interest that is prefix-closed $S \in \mathcal{P}(Q^+)$,
	associates the set of finite paths $\rho$ whose value, w.r.t. the residual
	strategy $\s^{\rho}$ is less than $\frac{1}{2}$. That is,
	$\mathsf{ValInf}_{u}(\s,S) := \{ \rho \in S \mid \MarVal{\G}[\s^\rho](\rho) < u
	\}$. 
	
	
	\begin{proof}
		For all states $q \in Q$, we denote by $\s_q: Q^+ \rightarrow \Dist(A) \in
		\SetStrat{\Aconc}{\A}$ a Player $\A$ strategy whose value from the state $q$ is
		at least $c \cdot (1 - c/4)$ (i.e. 
		$\MarVal{\G}[\s_q](q) \geq c \cdot (1 - c/4)$). Now, let $q_0 \in Q$ be a
		state and let us exhibit an almost-surely winning strategy from $q_0$. We build
		inductively a sequence of strategies $(\s_n)_{n \in \N} \in
		(\SetStrat{\Aconc}{\A})^\N$ and a sequence of sets of finite paths $(S_n)_{n \in
			\N} \in (\mathcal{P}(Q^+))^\N$ ensuring the following properties, for all $n \in
		\N$:
		\begin{itemize}
			\item[(a)] $S_{n}$ is suffix-closed;
			\item[(b)] $S_{n} \subseteq S_{n-1}$ and $S_n \cap \prefUniq{S_{n-1}} =
			\emptyset$;
			\item[(c)] for all $k < n$, $\s_n$ coincides with $\s_k$ on $Q^+ \setminus
			S_{k+1}$;
			\item[(d)] for all $\rho \in \prefUniq{S_n}$, the value of the residual
			strategy $\s_n$ in $\rho$ is at least $c \cdot (1 - c/4)$, i.e.
			$\MarVal{\G}[\s_n^\rho](\rho) \geq c \cdot (1 - c/4)$
		\end{itemize}
		
		
		Initially, we set $\s_0 := \s_{q_0}$ and $S_0 := q_0 \cdot Q^+$. Note that it
		is indeed suffix-closed, thus satisfying property $(a)$ and, since
		$\prefUniq{S_0} = \{ q_0 \}$, property $(d)$ is also ensured. Then, assume that
		for some $n \geq 1$, for all $k \leq n-1$, $s_k \in \SetStrat{\Aconc}{\A}$ and
		$S_k \subseteq Q^+$ are defined and properties $(a)-(d)$ hold. Let us define
		$s_{n}$ and $S_{n}$. If $S_{n-1} = \emptyset$, then $S_{n} := S_{n-1}$ and
		$\s_{n} := \s_{n-1}$ (and properties $(a)-(d)$ are ensured). 
		
		Now assume that $S_{n-1} \neq \emptyset$. Let us first define $S_n \subseteq
		Q^*$ and show that it satisfies properties $(a)-(b)$. Let $V_{n-1} =
		\mathsf{ValInf}_{c \cdot \frac{1}{2}}(\s_{n-1},S_{n-1}) \subseteq S_{n-1}$ be
		the set of finite paths of $S_{n-1}$ where the value of the residual strategy of
		$s_{n-1}$ is less than $c \cdot \frac{1}{2}$. Then, we define $S_{n} \subseteq
		Q^+$ as the set of suffixes of $V_{n-1}$: $S_{n} := \mathsf{Suff}(V_{n-1})$. In
		particular, note that $\prefUniq{S_{n}} = \prefUniq{V_{n-1}}$. 
		Clearly, $S_n$ is suffix-closed, hence property $(a)$ is ensured. 
		
		Consider now property $(b)$. We have $S_n = \mathsf{Suff}(V_{n-1}) \subseteq
		\mathsf{Suff}(S_{n-1}) = S_{n-1}$ since $V_{n-1} \subseteq S_{n-1}$ and
		$S_{n-1}$ is suffix-closed. 
		In addition, consider a path $\rho \in \prefUniq{S_{n-1}} \neq \emptyset$%
		. Assume towards a contradiction that $\rho \in S_{n} =
		\mathsf{Suff}(V_{n-1})$. Then, there is some $\rho' \in V_{n-1} \subseteq
		S_{n-1}$ such that $\rho \in \cyl(\rho')$. Since $\rho' \in S_{n-1}$, there
		would be some $\rho'' \in \prefUniq{S_{n-1}}$ such that $\rho' \in
		\cyl(\rho'')$, and hence $\rho \in \cyl(\rho'')$. We obtain that $\rho \in
		\cyl(\rho) \cap \cyl(\rho'')$ with $\rho,\rho'' \in \prefUniq{S_{n-1}}$. By
		definition, this implies $\rho = \rho'' $ and $\rho' \in \cyl(\rho)$. That is,
		$\rho = \rho' \in V_{n-1}$. However, by property $(d)$ at index $n-1$, we have
		$\MarVal{\G}[\s_{n-1}^\rho](\Head(\rho)) \geq c \cdot (1 - c/4) > c/2$ and hence
		$\rho \notin V_{n-1}$. Hence the contradiction. In fact, $\rho \notin S_{n}$ and
		property $(b)$ is ensured.
		
		
		
		Let us now define the strategy $\s_n: Q^+ \rightarrow \Dist(A)$. We set:
		\begin{itemize}
			\item $\restr{\s_n}{Q^+ \setminus S_n} := \s_{n-1}$;
			\item for all $\rho \in \prefUniq{V_{n-1}} = \prefUniq{S_n}$, we have
			$\s_n^\rho := \s_\mathsf{head(\rho)}$.
		\end{itemize} 
		Let us show that this definition ensures properties $(c)-(d)$. Let $k \leq
		n-1$. Then, $\s_k$ coincides with $\s_{n-1}$ on $Q^+ \setminus S_{k+1} \subseteq
		Q^+ \setminus S_n$. Hence, $\s_k$ also coincides with $\s_n$ on $Q^+ \setminus
		S_{k+1}$. Furthermore, $\s_{n-1}$ coincides with $\s_{n}$ on $Q^+ \setminus
		S_{n}$. Hence, property $(c)$ holds. Finally, since $W$ is a prefix-independent
		objective, for all $\rho \in \prefUniq{S_n}$ and $q := \Head(\rho) \in Q$, we
		have  $\MarVal{\G}[\s_n^\rho](\Head(\rho)) = \MarVal{\G}[\s_q](q) \geq c \cdot
		(1 - c/4)$. That is, property $(d)$ holds. 
		
		
		
		This concludes the inductive definitions of the sequences $(\s_n)_{n \in \N}
		\in (\SetStrat{\Aconc}{\A})^\N$ and $(S_n)_{n \in \N} \in
		(\mathcal{P}(Q^+))^\N$. We can then define the strategy $\s_\A$ as the limit of
		the strategies $\s_n$. That is, if at some point $S_n = \emptyset$  and $\s_n =
		\s_{n + k}$ for all $k \geq 0$, then we set $\s_\A := \s_n$. Otherwise, note
		that by property $(b)$ we have $\cap_{n \in \N} S_n = \emptyset$. Indeed,
		consider a finite path $\rho \in \cap_{n \in \N} S_n \subseteq Q^*$. For all $n
		\in \N$, there a unique finite path $\rho_n \in \prefUniq{S_n}$ such that $\rho
		\in \cyl(\rho_n)$. Furthermore, for all $k < l \in \N$, we have $\rho_k \neq
		\rho_{l}$, otherwise we would have $\prefUniq{S_k} \ni \rho_k \in \prefUniq{S_l}
		\subseteq S_l \subseteq S_{k+1}$ (since $l \geq k+1$) which does not hold.
		Hence, $\rho$ has infinitely many different finite prefixes, which is not
		possible. 
		It follows that, for all $\rho \in Q^+$, there is exactly one index $n \in \N$
		such that $\rho \in S_n \setminus S_{n-1}$.
		We then define the strategy $\s_\A$ in the following way:
		\begin{displaymath}
		\forall n \in \N,\; \restr{\s_\A}{S_n \setminus S_{n-1}} := \s_n
		\end{displaymath}
		
		With property $(c)$, this definition ensures that:
		\begin{equation}
		\forall n \in \N,\; \restr{\s_\A}{Q^+ \setminus S_{n+1}} = \s_n
		\label{eqn:s_as_sn}
		\end{equation}
		Note that this also holds in the first case where $S_n = \emptyset$ for some
		$n \in \N$.
		
		We claim that this Player $\A$ strategy is almost-sure. Consider a Player $\B$
		strategy $\s_\B$. Define the concurrent game $\G'$ that is obtained from $\G$ by
		unfolding the state space, thus obtaining the countable state space $q_0 \cdot
		Q^*$ (the set of strategies is unchanged). 
		Let us show that the value of the game $\G'$ with strategies $\s_\A$ and
		$\s_\B$ from all states (which corresponds to finite paths in $\G$) is at least
		$c \cdot \frac{1}{2} > 0$. This would imply that it is in fact $1$ from all
		states, by Lemma~\ref{lem:chaterjee_value_0_1}.
		
		For all sets $S \subseteq q_0 \cdot Q^*$, we denote by $\diamondsuit S$ the
		event specifying that the set $S$ is reached.
		
		Let $n \in \N$ and $\rho \in S_n \setminus S_{n+1}$. 
		First, let us consider a Player $\B$ strategy $\tilde{\s}_\B$ that coincides
		with $\s_\B$ on $Q^+ \setminus S_{n+1}$ and such that, for all paths $\rho' \in
		\prefUniq{S_{n+1}} \cap \cyl(\rho)$
		, we have $\prob{\rho'}{\s_n^{\rho'},\s_\B^{\rho'}}[W] \leq c \cdot
		\frac{1}{2}$. Note that this is possible since $\rho' \in \prefUniq{S_{n+1}} =
		\prefUniq{V_n} \subseteq V_n = \mathsf{ValInf}_{c \cdot
			\frac{1}{2}}(\s_{n},S_{n})$, which means that $\MarVal{\G}[\s_n^\rho](\rho') < c
		\cdot \frac{1}{2}$. With this definition, if
		$\mathbb{P}^\rho_{\s_n^\rho,\tilde{\s}_\B^\rho}[\diamondsuit S_{n+1}] > 0$, we
		have:
		\begin{align*}
		\mathbb{P}^\rho_{\s_n^\rho,\tilde{\s}_\B^\rho}[W \mid \diamondsuit S_{n+1}] &
		= \frac{\sum_{\pi \in \cyl(\rho) \cap \prefUniq{S_{n+1}}}
			\mathbb{P}^{\pi}_{\s_n^{\pi},\tilde{\s}_\B^{\pi}}[W] \cdot
			\mathbb{P}^{\rho}_{\s_n^{\rho},\tilde{\s}_\B^{\rho}}[
			\pi]}{\mathbb{P}^\rho_{\s_n^\rho,\tilde{\s}_\B^\rho}[\diamondsuit S_{n+1}]} \\ &
		\leq c \cdot \frac{1}{2} \frac{\sum_{\pi \in \cyl(\rho) \cap \prefUniq{S_{n+1}}}			
			\mathbb{P}^{\rho}_{\s_n^{\rho},\tilde{\s}_\B^{\rho}}[\pi]}{\mathbb{P}^\rho_{\s_n^\rho,\tilde{\s}_\B^\rho}[\diamondsuit
			S_{n+1}]} \\
		& = c \cdot \frac{1}{2} \cdot
		\frac{\mathbb{P}^\rho_{\s_n^\rho,\tilde{\s}_\B^\rho}[\diamondsuit
			S_{n+1}]}{\mathbb{P}^\rho_{\s_n^\rho,\tilde{\s}_\B^\rho}[\diamondsuit S_{n+1}]}
		= c \cdot \frac{1}{2}
		\end{align*}
		Furthermore, since $\tilde{\s}_\B$ and $\s_\B$ coincide on $Q^+ \setminus
		S_{n+1}$:
		\begin{align*}
		\mathbb{P}^\rho_{\s_n^\rho,\tilde{\s}_\B^\rho}[W] & =
		\mathbb{P}^\rho_{\s_n^\rho,\tilde{\s}_\B^\rho}[W \mid \diamondsuit S_{n+1}]
		\cdot \mathbb{P}^\rho_{\s_n^\rho,\tilde{\s}_\B^\rho}[\diamondsuit S_{n+1}] +
		\mathbb{P}^\rho_{\s_n^\rho,\tilde{\s}_\B^\rho}[W \mid \lnot \diamondsuit
		S_{n+1}] \cdot \mathbb{P}^q_{\s_n^\rho,\tilde{\s}_\B^\rho}[\lnot \diamondsuit
		S_{n+1}] \\
		& \leq c \cdot \frac{1}{2} \cdot
		\mathbb{P}^\rho_{\s_n^\rho,\s_\B^\rho}[\diamondsuit S_{n+1}] +
		\mathbb{P}^\rho_{\s_n^\rho,\s_\B^\rho}[W \mid \lnot \diamondsuit S_{n+1}] \cdot
		\mathbb{P}^\rho_{\s_n^\rho,\s_\B^\rho}[\lnot \diamondsuit S_{n+1}] \\
		& = c \cdot \frac{1}{2} \cdot p_1 + p_2 \cdot (1-p_1)
		\end{align*}
		for $p_1 := \mathbb{P}^\rho_{\s_n^\rho,\s_\B^\rho}[\diamondsuit S_{n+1}]$ and
		$p_2 := \mathbb{P}^\rho_{\s_n^\rho,\s_\B^\rho}[W \mid \lnot \diamondsuit
		S_{n+1}]$. We obtain:
		\begin{equation}
		\mathbb{P}^\rho_{\s_n^\rho,\tilde{\s}_\B^\rho}[W] \leq \frac{c}{2} \cdot p_1 +
		p_2 \cdot (1-p_1)
		\label{eqn:ineq}
		\end{equation}
		
		This inequality holds for all $\rho \in S_n \setminus S_{n+1}$ and $n \in \N$.

		Now, in the case where $\rho \in \prefUniq{S_n}$, let us show that
		$\mathbb{P}^\rho_{\s_\A^\rho,\s_\B^\rho}[W] \geq c \cdot \frac{1}{2}$. Indeed,
		we have by property $(d)$: $\MarVal{\G}[\s_n^\rho](\rho) \geq c \cdot (1 -
		c/4)$. This implies $\mathbb{P}^\rho_{\s_n^\rho,\tilde{\s}_\B^\rho}[W] \geq c
		\cdot (1 - c/4)$. Therefore:
		\begin{align*}
		c \cdot (1 - \frac{c}{4}) \leq \frac{c}{2} \cdot p_1 + p_2 \cdot (1 - p_1)
		\end{align*}
		Hence, $p_2 > \frac{c}{2}$ (since $1/2 < 1 - c/4$) and:
		\begin{align*}
		p_1 \leq \frac{p_2 - c \cdot (1 - \frac{c}{4})}{p_2 - \frac{c}{2}} 
		\end{align*}
		That is:
		\begin{align*}
		p_2 \cdot (1 - p_1) \geq p_2 \cdot \frac{p_2 - \frac{c}{2} - p_2 + c \cdot (1
			- \frac{c}{4})}{p_2 - \frac{c}{2}} = p_2 \cdot \frac{c \cdot (\frac{1}{2} -
			\frac{c}{4})}{p_2 - \frac{c}{2}} = \frac{c}{2} \cdot \frac{p_2 - p_2 \cdot
			\frac{c}{2}}{p_2 - \frac{c}{2}} \geq \frac{c}{2}
		\end{align*}
		
		
		We can now consider the probability
		$\mathbb{P}^\rho_{\s_\A^\rho,\s_\B^\rho}[W]$ of satisfying $W$ given strategies
		$\s_\A$ and $\tilde{\s}_\B$. Note that $\s_\A$ coincides with $\s_n$ on $S_n
		\setminus S_{n+1}$ and in particular on $\cyl(\rho) \setminus S_{n+1}$. Hence,
		we have:
		\begin{align*}
		\mathbb{P}^\rho_{\s_\A^\rho,\s_\B^\rho}[W] & =
		\mathbb{P}^\rho_{\s_\A^\rho,\s_\B^\rho}[W \mid \diamondsuit S_{n+1}] \cdot
		\mathbb{P}^\rho_{\s_\A^\rho,\s_\B^\rho}[\diamondsuit S_{n+1}] +
		\mathbb{P}^\rho_{\s_\A^\rho,\s_\B^\rho}[W \mid \lnot \diamondsuit S_{n+1}] \cdot
		\mathbb{P}^\rho_{\s_\A^\rho,\s_\B^\rho}[\lnot \diamondsuit S_{n+1}] \\
		& \geq \mathbb{P}^\rho_{\s_\A^\rho,\s_\B^\rho}[W \mid \lnot \diamondsuit
		S_{n+1}] \cdot \mathbb{P}^\rho_{\s_\A^\rho,\s_\B^\rho}[\lnot \diamondsuit
		S_{n+1}] \\
		& = \mathbb{P}^\rho_{\s_n^\rho,\s_\B^\rho}[W \mid \lnot \diamondsuit S_{n+1}]
		\cdot \mathbb{P}^\rho_{\s_n^\rho,\s_\B^\rho}[\lnot \diamondsuit S_{n+1}] \\ & =
		p_2 \cdot (1 - p_1) \geq \frac{c}{2}
		\end{align*}
		This holds for all $\rho \in \prefUniq{S_n}$ and for all $n \in \N$. 
		
		Consider now some arbitrary $\rho \in S_n \setminus S_{n+1}$. In this case,
		since $\rho \notin \mathsf{ValInf}_{c \cdot \frac{1}{2}}(\s_{n},S_{n}) = V_{n}
		\subseteq S_{n+1}$, we have $\MarVal{\G}[\s_n^\rho](\rho) \geq \frac{c}{2}$ and
		$\mathbb{P}^\rho_{\s_n^\rho,\tilde{\s}_\B^\rho}[W] \geq \frac{c}{2}$. Hence,
		with Equation~\ref{eqn:ineq} we have:
		\begin{displaymath}
		\frac{c}{2} \leq \mathbb{P}^\rho_{\s_n^\rho,\tilde{\s}_\B^\rho}[W] \leq
		\frac{c}{2} \cdot p_1 + p_2 \cdot (1-p_1)
		\end{displaymath}
		Hence, assuming $1 - p_1 = \mathbb{P}^\rho_{\s_n^\rho,\s_\B^\rho}[\lnot
		\diamondsuit S_{n+1}] > 0$, we obtain that $p_2  =
		\mathbb{P}^\rho_{\s_n^\rho,\s_\B^\rho}[W \mid \lnot \diamondsuit S_{n+1}] \geq
		\frac{c}{2}$. Furthermore, we have shown that for all $\rho' \in
		\prefUniq{S_{n+1}}$, we have 
		$\mathbb{P}^{\rho'}_{\s_\A^{\rho'},\s_\B^{\rho'}}[W] \geq \frac{c}{2}$. Hence,
		assuming that $\mathbb{P}^\rho_{\s_\A^\rho,\s_\B^\rho}[\diamondsuit S_{n+1}]
		\neq 0$, we have: 
		\begin{align*}
		\mathbb{P}^\rho_{\s_\A^\rho,\s_\B^\rho}[W \mid \diamondsuit S_{n+1}] & =
		\frac{1}{\mathbb{P}^\rho_{\s_\A^\rho,\s_\B^\rho}[\diamondsuit S_{n+1}]} \cdot
		\sum_{\rho' \in \prefUniq{S_{n+1}}}
		\mathbb{P}^{\rho'}_{\s_\A^{\rho'},\s_\B^{\rho'}}[W] \cdot
		\mathbb{P}^\rho_{\s_\A^\rho,\s_\B^\rho}[\rho'] \\
		& \geq \frac{1}{\mathbb{P}^\rho_{\s_\A^\rho,\s_\B^\rho}[\diamondsuit S_{n+1}]}
		\cdot \sum_{\rho' \in \prefUniq{S_{n+1}}} \frac{c}{2} \cdot
		\mathbb{P}^\rho_{\s_\A^\rho,\s_\B^\rho}[\rho'] \\
		& = \frac{c}{2} \cdot
		\frac{1}{\mathbb{P}^\rho_{\s_\A^\rho,\s_\B^\rho}[\diamondsuit S_{n+1}]} \cdot
		\sum_{\rho' \in \prefUniq{S_{n+1}}} \cdot
		\mathbb{P}^\rho_{\s_\A^\rho,\s_\B^\rho}[\rho'] \\
		& = \frac{c}{2} \cdot
		\frac{1}{\mathbb{P}^\rho_{\s_\A^\rho,\s_\B^\rho}[\diamondsuit S_{n+1}]} \cdot
		\mathbb{P}^\rho_{\s_\A^\rho,\s_\B^\rho}[\diamondsuit S_{n+1}] \\
		& = \frac{c}{2}
		\end{align*}
		
		Then, assuming that $0<  \mathbb{P}^\rho_{\s_\A^\rho,\s_\B^\rho}[\diamondsuit
		S_{n+1}] < 1$, we have:
		\begin{align*}
		\mathbb{P}^\rho_{\s_\A^\rho,\s_\B^\rho}[W] & =
		\mathbb{P}^\rho_{\s_\A^\rho,\s_\B^\rho}[W \mid \diamondsuit S_{n+1}] \cdot
		\mathbb{P}^\rho_{\s_\A^\rho,\s_\B^\rho}[\diamondsuit S_{n+1}] +
		\mathbb{P}^\rho_{\s_\A^\rho,\s_\B^\rho}[W \mid \lnot \diamondsuit S_{n+1}] \cdot
		\mathbb{P}^\rho_{\s_\A^\rho,\s_\B^\rho}[\lnot \diamondsuit S_{n+1}] \\
		& \geq \frac{c}{2} \cdot \mathbb{P}^\rho_{\s_\A^\rho,\s_\B^\rho}[\diamondsuit
		S_{n+1}] + \mathbb{P}^\rho_{\s_\A^\rho,\s_\B^\rho}[W \mid \lnot \diamondsuit
		S_{n+1}] \cdot \mathbb{P}^\rho_{\s_\A^\rho,\s_\B^\rho}[\lnot \diamondsuit
		S_{n+1}] \\
		& = \frac{c}{2} \cdot \mathbb{P}^\rho_{\s_n^\rho,\s_\B^\rho}[\diamondsuit
		S_{n+1}] + \mathbb{P}^\rho_{\s_n^\rho,\s_\B^\rho}[W \mid \lnot \diamondsuit
		S_{n+1}] \cdot \mathbb{P}^\rho_{\s_n^\rho,\s_\B^\rho}[\lnot \diamondsuit
		S_{n+1}] \\
		& = \frac{c}{2} \cdot \mathbb{P}^\rho_{\s_n^\rho,\s_\B^\rho}[\diamondsuit
		S_{n+1}] + p_2 \cdot \mathbb{P}^\rho_{\s_n^\rho,\s_\B^\rho}[\lnot \diamondsuit
		S_{n+1}] \\
		& \geq \frac{c}{2} \cdot \mathbb{P}^\rho_{\s_n^\rho,\s_\B^\rho}[\diamondsuit
		S_{n+1}] + \frac{c}{2} \cdot \mathbb{P}^\rho_{\s_n^\rho,\s_\B^\rho}[\lnot
		\diamondsuit S_{n+1}] \\
		& = \frac{c}{2}
		\end{align*}
		Note that this also holds if the probability
		$\mathbb{P}^\rho_{\s_\A^\rho,\s_\B^\rho}[\diamondsuit S_{n+1}]
		=\mathbb{P}^\rho_{\s_n^\rho,\s_\B^\rho}[\diamondsuit S_{n+1}]$ is either equal
		to 0 or to 1.
		
		Overall, we obtain that for all $\rho \in Q^+$, we have $\MarVal{\G'}(\rho) =
		\mathbb{P}^{\rho}_{\s_\A^\rho,\s_\B^\rho}[W] \geq \frac{1}{2}$. Hence, $\inf_{q
			\in Q} \MarVal{\G'}(q) \geq \frac{1}{2}$, which implies, by
		Lemma~\ref{lem:chaterjee_value_0_1}, that, for all $\rho \in q_0 \cdot Q^*$, we
		have $1 = \MarVal{\G'}(q) = \mathbb{P}^{\G,\rho}_{\s_\A^\rho,\s_\B^\rho}[W]$. in
		particular, $\mathbb{P}^{\G}_{\s_\A,\s_\B}[W] = 1$. As this holds for all Player
		$\B$ strategy $\s_\B$, it follows that the strategy $\s_\A$ is almost-sure. We
		can then do the same from all states $q \in Q$ to obtain a strategy 
		almost-sure. Furthermore, we have shown that the value of the strategy from
		all finite paths is at least $c/2 > 0$. In fact, this implies that the strategy
		is subgame almost-sure. 
	\end{proof}
	
	\section{Finite-choice strategies}
	
	
	\subsection{Proof of Theorem~\ref{thm:fair_strategy_uniformly_optimal}}
	\label{proof:thm_fair_strategy_uniformly_optimal}	
	We show Theorem~\ref{thm:fair_strategy_uniformly_optimal} for more general
	strategies, namely, positively bounded strategies. 
	\begin{definition}[Positively bounded strategy]
		A Player $\A$ strategy $\s_\A$ is \emph{positively bounded} (p.b. for short)
		if there is a constant $c > 0$ such that, for all $\rho \cdot q \in Q^+$, for
		all $a \in A_q$ we have: $\s_\A(\rho \cdot q)(q) > 0 \Rightarrow \s_\A(\rho
		\cdot q)(q) \geq c$.
	\end{definition}
	
	We show the theorem below.
	\begin{theorem}
		\label{thm:positively_bounded_strategy_uniformly_optimal}
		Consider a PI concurrent game $\G$. Assume that there is an 
		optimal strategy that is positively bounded. Then, there is a subgame optimal
		strategy that is positively bounded.
	\end{theorem}
	\begin{proof}
		Let us denote by $\s_\mathsf{pb}$ an 
		optimal 
		positively bounded Player $\A$ strategy. Let us build inductively a Player
		$\A$ subgame optimal strategy $\s_\A$. It is defined as follows: for all finite
		paths $\rho = \rho' \cdot q \in Q^+$, we set $\s_\A(\rho)$ to:
		\begin{displaymath}
		\s_\A(\rho) := \begin{cases}
		\s_{\mathsf{pb}}(\rho) & \text{ if }\s_{\mathsf{pb}} \text{ is optimal from
		}\rho\text{, i.e. } \MarVal{\G}(\s_{\mathsf{pb}}^\rho)[q] = \MarVal{\G}[q]\\
		\s_{\mathsf{pb}}(q) & \text{ otherwise }\\
		\end{cases}
		\end{displaymath}
		Since the strategy $\s_{\mathsf{pb}}$ is p.b.
		, it follows that the strategy $\s_\A$ also is.	Let us show that it is subgame
		optimal by applying Theorem~\ref{thm:subgame_optimal_arbitrary_strategy}. 
		
		Let $\rho = \rho' \cdot q \in Q^+$. In all cases, the strategy $\s_\A$
		coincides with the strategy $\s_{\mathsf{pb}}$ (either at $\rho$ or $q$) which
		is optimal from $q$, hence, by
		Proposition~\ref{prop:optimal_implies_locally_optimal}, we have that, for all
		Player $\B$ action $b \in B_q$: $\outM_{\formNF_q}(\s_\A(q),b) \geq
		\MarVal{\G}(q)$. It follows that the strategy $\s_\A$ is locally optimal. 
		
		Let us now show that it ensures the second property. Let $\rho \in Q^+$ and
		let us denote by $\s_\rho$ the residual strategy $\s_\A^\rho$. Consider a Player
		$\B$ deterministic strategy $\s_\B$ and some value $u \in V_{\G} \setminus \{ 0
		\}$. We introduce two notations:
		\begin{itemize}
			\item we denote by $\mathsf{Exit}_u \subseteq Q^+$ the set of finite paths
			ending in $Q_u$ with a positive probability to exit this value area:
			$\mathsf{Exit}_u := \{ \pi \cdot Q^* \cdot Q_u \mid
			\mathbb{P}^{\s_\rho,\s_\B}_{\pi}[Q \setminus Q_u] > 0\}$. 
			\item we also denote by $\mathsf{Deviate} \subseteq Q^+$ the set of finite
			paths where the strategies $\s_\A^\mathsf{pb}$ is not optimal: $\mathsf{Deviate}
			:= \{ \pi = \pi' \cdot q \in Q^+ \mid \MarVal{\G}[\s_{\mathsf{pb}}^{\rho \cdot
				\pi}](q) <  \MarVal{\G}(q) \}$.
		\end{itemize}
		Let us show the following facts:
		\begin{itemize}
			\item[(a).] $\mathbb{P}^{\s_\rho,\s_\B}_{\rho}[Q^* \cdot (Q_u)^\omega \cap
			(Q^* \cdot \mathsf{Exit}_u)^\omega]
			= 0$;
			\item[(b).] $\mathbb{P}^{\s_\rho,\s_\B}[Q^* \cdot (Q_u)^\omega \cap Q^* \cdot
			(Q \setminus \mathsf{Exit}_u)^\omega] \leq \mathbb{P}^{\s_\rho,\s_\B}[Q^* \cdot
			(Q_u)^\omega 
			\cap Q^* \cdot (Q \setminus \mathsf{Deviate})^\omega]$
			;
			\item[(c).] $\mathbb{P}^{\s_\rho,\s_\B}_{\rho}[Q^* \cdot (Q_u \setminus
			\mathsf{Deviate})^\omega] = \mathbb{P}^{\s_\rho,\s_\B}_{\rho}[W \cap Q^* \cdot
			(Q_u \setminus \mathsf{Deviate})^\omega]$.
		\end{itemize}
		If we assume that all these facts hold, then we obtain:
		\begin{align*}
		\mathbb{P}^{\s_\rho,\s_\B}_{\rho}[Q^* \cdot (Q_u)^\omega] & =
		\mathbb{P}^{\s_\rho,\s_\B}_{\rho}[Q^* \cdot (Q_u)^\omega \cap Q^* \cdot (Q
		\setminus \mathsf{Exit}_u)^\omega] & \text{ by fact (a)} \\
		& \leq \mathbb{P}^{\s_\rho,\s_\B}_{\rho}[Q^* \cdot (Q_u)^\omega \cap Q^* \cdot
		(Q \setminus \mathsf{Deviate})^\omega] & \text{ by fact (b)} \\
		& = \mathbb{P}^{\s_\rho,\s_\B}_{\rho}[Q^* \cdot (Q_u \setminus
		\mathsf{Deviate})^\omega]  \\
		& = \mathbb{P}^{\s_\rho,\s_\B}_{\rho}[W \cap Q^* \cdot (Q_u \setminus
		\mathsf{Deviate})^\omega] & \text{ by fact (c)} \\
		& \leq \mathbb{P}^{\s_\rho,\s_\B}_{\rho}[W \cap Q^* \cdot (Q_u)^\omega]  \\
		& \leq \mathbb{P}^{\s_{\rho},\s_\B}_{\rho}[Q^* \cdot (Q_u)^\omega] 
		\end{align*}
		In fact, all these inequalities are equalities. We can then apply
		Theorem~\ref{thm:subgame_optimal_arbitrary_strategy} to conclude. 
		Let us now show all these facts one by one. 
		\begin{itemize}
			\item[(a).] 
			Consider some $\pi = \pi' \cdot q \in \mathsf{Exit}_u$. We have
			$\mathbb{P}^{\s^\pi_\rho,\s_\B^\pi}_{\rho \cdot \pi}[Q \setminus Q_u] > 0$. Let
			$b := \s_\B(\pi)$ (recall that $\s_\B$ is a deterministic strategy) and let
			$A_{Q \setminus Q_u} := \{ a \in A_q \mid \distribFunc \circ \delta(q,a,b)[Q
			\setminus Q_u] > 0 \}$. Then, $\s^\pi_\rho[A_{Q \setminus Q_u}] > 0$ hence
			$\s^\pi_\rho[A_{Q \setminus Q_u}] \geq c$ for some fixed $c > 0$ (since
			$\s_\rho$ is p.b.). Furthermore, let $x := \min_{d \in \distribSet} \min_{q \in
				\Supp(\distribFunc(d))} \distribFunc(d)(q) > 0$. It follows that
			$\mathbb{P}^{\s^\pi_\rho,\s_\B^\pi}_{\rho \cdot \pi}[Q \setminus Q_u] \geq c
			\cdot x$. In fact, this holds for all $\pi \in \mathsf{Exit}_u$. Hence, for all
			$\pi \in Q^*$, we have $\mathbb{P}^{\s^\pi_\rho,\s_\B^\pi}_{\rho \cdot
				\pi}[(Q_u)^\omega \mid (Q^* \cdot \mathsf{Exit}_u)^\omega] \leq \lim_{n
				\rightarrow \infty} (1 - c \cdot x)^n = 0$. It follows that
			$\mathbb{P}^{\s_\rho,\s_\B}_{\rho}[Q^* \cdot (Q_u)^\omega \cap (Q^* \cdot
			\mathsf{Exit}_u)^\omega] = 0$.
			\item[(b).] Let us show that $\mathbb{P}^{\s_\rho,\s_\B}[Q^* \cdot
			(Q_u)^\omega \cap Q^* \cdot (Q \setminus \mathsf{Exit}_u)^\omega \cap (Q^* \cdot
			\mathsf{Deviate})^\omega] = 0$.	Let $\theta \in Q^* \cdot (Q_u)^\omega \cap Q^*
			\cdot (Q \setminus \mathsf{Exit}_u)^\omega$. Let $n \in \mathbb{N}$ be an index
			such that $\theta_{\geq n} \in (Q_u \setminus \mathsf{Exit}_u)^\omega$. 
			Consider, assuming it exists, the least index $i \geq n+1$ such that
			$\theta_i \in \mathsf{Deviate}$. That is, $\MarVal{\G}[\s_{\mathsf{pb}}^{\rho
				\cdot \theta_{\leq i}}](\theta_i) < \MarVal{\G}(\theta_i)$ and
			$\MarVal{\G}[\s_{\mathsf{pb}}^{\rho \cdot \theta_{\leq i-1}}](\theta_{i-1}) =
			\MarVal{\G}(\theta_{i-1})$. With a straightforward adaptation of
			Proposition~\ref{prop:optimal_implies_locally_optimal}, if
			$\mathbb{P}^{\s_\rho,\s_\B}_{\theta_{\leq i-1}}[\theta_i] > 0$, for $b :=
			\s_\B(\rho \cdot \theta_{\leq i-1})$ (recall that $\s_\B$ is deterministic), we
			have $\outM_{\formNF_{\theta_{i-1}}}(\s_\rho(\theta_{\leq i-1}),b) >
			\MarVal{\G}(\theta_{i-1}) = u$. Hence, at $\theta_{\leq i-1}$, there is a
			non-zero probability to reach a state of value different from $u$, i.e.
			$\mathbb{P}^{\s_\rho,\s_\B}_{\theta_{\leq i-1}}[Q \setminus Q_u] > 0$. That is,
			$\theta_{i-1} \in \mathsf{Exit}_u$.
			That is a path -- with a positive probability to occur -- that does not visit
			$\mathsf{Exit}_u$ does not visit $\mathsf{Deviate}$ as well. Hence,
			almost-surely, a path visiting $\mathsf{Exit}_u$ only finitely often visits
			$\mathsf{Deviate}$ only finitely often. It follows that
			$\mathbb{P}^{\s_\rho,\s_\B}[Q^* \cdot (Q_u)^\omega \cap Q^* \cdot (Q \setminus
			\mathsf{Exit}_u)^\omega \cap (Q^* \cdot \mathsf{Deviate})^\omega] = 0$. That is:
			$\mathbb{P}^{\s_\rho,\s_\B}[Q^* \cdot (Q_u)^\omega \cap Q^* \cdot (Q \setminus
			\mathsf{Exit}_u)^\omega] = \mathbb{P}^{\s_\rho,\s_\B}[Q^* \cdot (Q_u)^\omega
			\cap Q^* \cdot (Q \setminus \mathsf{Exit}_u)^\omega \cap Q^* \cdot (Q \setminus
			\mathsf{Deviate})^\omega] \leq \mathbb{P}^{\s_\rho,\s_\B}[Q^* \cdot (Q_u)^\omega
			\cap Q^* \cdot (Q \setminus \mathsf{Deviate})^\omega]$.
			\item[(c).] We proceed similarly to how we proved the necessary conditions of
			Theorem~\ref{thm:subgame_optimal_arbitrary_strategy} with an additional
			difficulty to conclude. Indeed, assume towards a contradiction that
			$\mathbb{P}^{\s_\rho,\s_\B}_{\rho}[Q^* \cdot (Q_u \setminus
			\mathsf{Deviate})^\omega] > \mathbb{P}^{\s_\rho,\s_\B}_{\rho}[W \cap Q^* \cdot
			(Q_u \setminus \mathsf{Deviate})^\omega]$. Consider the countable Markov chain
			induced by both strategies $\s_\rho$ and $\s_\B$ whose initial state is $\rho$.
			We denote by $\mathbb{P}$ the probability function in that Markov chain.
			Consider the prefix-independent objective $W_u := W \cup (Q^* \cdot (Q \setminus
			Q_u \cup \mathsf{Deviate}))^\omega$. We have:
			\begin{align*}
			\mathbb{P}[W_u] & = \mathbb{P}[W_u \cap Q^* \cdot (Q_u \setminus
			\mathsf{Deviate})^\omega] + \mathbb{P}[W_u \cap (Q^* \cdot (Q \setminus Q_u \cup
			\mathsf{Deviate})^\omega] \\
			& = \mathbb{P}[W \cap Q^* \cdot (Q_u \setminus \mathsf{Deviate})^\omega] +
			\mathbb{P}[(Q^* \cdot (Q \setminus Q_u \cup \mathsf{Deviate}))^\omega] \\
			& < \mathbb{P}[Q^* \cdot (Q_u \setminus \mathsf{Deviate})^\omega] +
			\mathbb{P}[(Q^* \cdot (Q \setminus Q_u \cup \mathsf{Deviate}))^\omega] \\
			& = 1
			\end{align*}
			Let us now show that there is a finite paths whose last state is in $Q_u
			\setminus \mathsf{Deviate}$ and from which the probability of $W_u$ is less than
			$u/2$. By Lemma~\ref{lem:chaterjee_value_0_1}, since $W_u$ is prefix
			independent, there is state $\pi \in Q^*$ (which corresponds to a finite path)
			in the Markov chain such that $\mathbb{P}_\pi[W_u] < u/2$. 
			
			Now, assume towards a contradiction that, for all $\pi' \in \pi \cdot Q^*
			\cdot (Q_u \setminus \mathsf{Deviate})$, we have either $\mathbb{P}_{\pi'}[W_u]
			\geq \frac{u}{2}$ or $\mathbb{P}_{\pi'}[Q^* \cdot (Q \setminus Q_u \cup
			\mathsf{Deviate})] > \frac{u}{2}$. Let us denote by $Q_1 := \{ \pi' \in \pi
			\cdot Q^* \cdot (Q_u \setminus \mathsf{Deviate}) \mid \mathbb{P}_{\pi'}[W_u]
			\geq \frac{u}{2} \}$ and by $Q_2 := \pi \cdot Q^* \cdot (Q_u \setminus
			\mathsf{Deviate}) \setminus Q_1$. By definition, $ \pi \cdot Q^* \cdot (Q_u
			\setminus \mathsf{Deviate}) = Q_1 \uplus Q_2$ and by assumption, for all $\pi'
			\in Q_2$, we have $\mathbb{P}_{\pi'}[Q^* \cdot (Q \setminus Q_u \cup
			\mathsf{Deviate})] > u/2$. It follows that $\mathbb{P}_{\pi}[Q^* \cdot (Q_u
			\setminus \mathsf{Deviate})^\omega \cap (Q^* \cdot Q_2)^\omega] = 0$. That is,
			$\mathbb{P}_{\pi}[Q^* \cdot (Q_u \setminus \mathsf{Deviate})^\omega] =
			\mathbb{P}_{\pi}[Q^* \cdot (Q_1 \uplus Q_2)^\omega] = \mathbb{P}_{\pi}[Q^* \cdot
			(Q_1)^\omega]$. Furthermore, $\mathbb{P}_\pi[W_u \cap Q^* \cdot (Q_1)^\omega] =
			\mathbb{P}_\pi[Q^* \cdot (Q_1)^\omega]$ by definition of $Q_1$ (since at some
			point, only states of with probability at least $u/2$ of $W_u$ are seen) and
			Theorem~\ref{thm:levy}. Then:
			\begin{align*}
			\mathbb{P}_\pi[W_u] & = \mathbb{P}_\pi[W_u \cap Q^* \cdot (Q_u \setminus
			\mathsf{Deviate})^\omega] + \mathbb{P}_\pi[W_u \cap (Q^* \cdot (Q \setminus Q_u
			\cup \mathsf{Deviate}))^\omega] \\
			& = \mathbb{P}_\pi[W_u \cap Q^* \cdot (Q_1)^\omega] + \mathbb{P}_\pi[(Q^*
			\cdot (Q \setminus Q_u \cup \mathsf{Deviate}))^\omega] \\
			& = \mathbb{P}_\pi[Q^* \cdot (Q_1)^\omega] + \mathbb{P}_\pi[(Q^* \cdot (Q
			\setminus Q_u \cup \mathsf{Deviate}))^\omega] \\
			& = \mathbb{P}_\pi[Q^* \cdot (Q_u \setminus \mathsf{Deviate})^\omega] +
			\mathbb{P}_\pi[(Q^* \cdot (Q \setminus Q_u \cup \mathsf{Deviate}))^\omega] \\
			& = 1
			\end{align*}
			This is a contradiction with the fact that $\mathbb{P}_\pi[W_u] \leq u/2$. In
			fact, there is some $\pi' \in \pi \cdot Q^* \cdot (Q_u \setminus
			\mathsf{Deviate})$ such that we have $\mathbb{P}_{\pi'}[W_u] < \frac{u}{2}$ and
			$\mathbb{P}_{\pi'}[Q^* \cdot (Q \setminus Q_u \cup \mathsf{Deviate})] \leq u/2$.
			Let $\pi' = \pi'' \cdot q$. Since $\pi' \notin \mathsf{Deviate}$, we have
			$\MarVal{\G}[\s_{\mathsf{pb}}^{\rho \cdot \pi'}](q) = \MarVal{\G}(q) = u$. It
			follows that:
			\begin{align*}
			u \leq \mathbb{P}_{\rho \cdot \pi'}^{\s_{\mathsf{pb}}^{\rho \cdot
					\pi'},\s_\B}[W] & \leq \mathbb{P}_{\rho \cdot \pi'}^{\s_{\mathsf{pb}}^{\rho
					\cdot \pi'},\s_\B}[W \mid (Q_u \setminus \mathsf{Deviate})^\omega] \cdot
			\mathbb{P}_{\rho \cdot \pi'}^{\s_{\mathsf{pb}}^{\rho \cdot \pi'},\s_\B}[(Q_u
			\setminus \mathsf{Deviate})^\omega] \\
			& + \mathbb{P}_{\rho \cdot \pi'}^{\s_{\mathsf{pb}}^{\rho \cdot
					\pi'},\s_\B}[Q^* \cdot (Q \setminus Q_u \cup \mathsf{Deviate})] \\
			& \leq \mathbb{P}_{\rho \cdot \pi'}^{\s_{\mathsf{pb}}^{\rho \cdot
					\pi'},\s_\B}[W \mid (Q_u \setminus \mathsf{Deviate})^\omega] \cdot (1 -
			\frac{u}{2}) + \frac{u}{2}
			\end{align*}
			Note that we can indeed relate these probabilities with the previous ones
			(expressed with $\mathbb{P}$) with the strategy $\s_\rho$ -- instead of
			$\s_{\mathsf{pb}}^{\rho \cdot \pi'}$ -- since these two strategies coincide
			outside of $\mathsf{Deviate}$. We obtain:
			\begin{displaymath}
			\mathbb{P}_{\rho \cdot \pi'}^{\s_{\mathsf{pb}}^{\rho \cdot \pi'},\s_\B}[W
			\mid (Q_u \setminus \mathsf{Deviate})^\omega] = \mathbb{P}_{\pi'}[W \mid (Q_u
			\setminus \mathsf{Deviate})^\omega] \geq \frac{u}{2 - u}
			\end{displaymath}
			We can then conclude that:
			\begin{align*}
			\mathbb{P}_{\pi'}[W_u] & \geq \mathbb{P}_{\pi'}[W_u \cap (Q_u \setminus
			\mathsf{Deviate})^\omega] \\
			& = \mathbb{P}_{\pi'}[W \mid (Q_u \setminus \mathsf{Deviate})^\omega] \cdot
			\mathbb{P}_{\pi'}[(Q_u \setminus \mathsf{Deviate})^\omega] \\
			& \geq \frac{u}{2 - u} \cdot (1 - \frac{u}{2}) \\
			& = \frac{u}{2}
			\end{align*}
			This is a contradiction with the fact that $\mathbb{P}_{\pi'}[W_u] <
			\frac{u}{2}$. In fact, our assumption $\mathbb{P}^{\s_\rho,\s_\B}_{\rho}[Q^*
			\cdot (Q_u \setminus \mathsf{Deviate})^\omega] >
			\mathbb{P}^{\s_\rho,\s_\B}_{\rho}[W \cap Q^* \cdot (Q_u \setminus
			\mathsf{Deviate})^\omega]$ does not hold. That is, we have
			$\mathbb{P}^{\s_\rho,\s_\B}_{\rho}[Q^* \cdot (Q_u \setminus
			\mathsf{Deviate})^\omega] = \mathbb{P}^{\s_\rho,\s_\B}_{\rho}[W \cap Q^* \cdot
			(Q_u \setminus \mathsf{Deviate})^\omega]$.
		\end{itemize}
	\end{proof}
	The proof of Theorem~\ref{thm:fair_strategy_uniformly_optimal} can then be done
	in a similar way with the additional remark that if $\s_{\mathsf{fc}}$ is a
	finite-choice strategy, then the strategy $\s_\A$ defined in this proof also is.
	
	\subsection{Proof of Theorem~\ref{thm:transfer_memory_finite_choice}}
	\label{proof:thm_transfer_memory_finite_choice}
	First note that $W_\mathsf{proj}$ is prefix-independent. Indeed, for $\rho \in
	W_\mathsf{proj}$. We have $\rho = k^n \cdot \pi_0 \cdot k \cdot \pi_1 \cdots$
	for $n \in \{ 0,1 \}$ and $\pi \in W$. Then, for all $i \in \N$, we have
	$\rho_{\geq i} = k^{n'} \cdot \pi_j \cdot k \cdot \pi_{j+1} \cdots$ for some $n'
	\in \{ 0,1 \}$ and $j \in \N$. Then, since $W$ is prfix independent, then
	$\pi_{\geq j} \in W$.
	
	Now, for a finite choice strategy $\s_\A$ and a value $u \in V_{\G} \setminus
	\{ 0 \}$, let us define the game $\G^{\mathsf{tb}}_u$:
	\begin{definition}
		\label{def:game_g_u_tb}
		For a positive value $u \in V_{\G} \setminus \{ 0 \}$, we define the game
		$\G^\mathsf{tb}_u = \Games{\Aconc_u}{W_{\mathsf{proj}}}$ with
		$\Aconc^\mathsf{tb}_u = \langle Q^\A_u \uplus Q^\B_u \uplus \Aconc_W,(A'_q)_{q
			\in Q'_u},(B'_q)_{q \in
			Q'_u},\distribSet,\delta',\distribFunc',\colSet,\colFunc' \rangle$ where we
		refer to the arena $\Aconc_W$ as a sink state to which there are some
		transitions:
		\begin{itemize}
			\item $Q_u^\A := Q_u$ the set of Player $\A$'s states;
			\item For all $q \in Q^\A_u$, we have $A'_q := S_q$ and $B'_q := \{
			b_{\mathsf{id}} \}$ for a new fresh action $b_{\mathsf{id}}$;
			\item $Q_u^\B := \{ (q,\sigma) \mid q \in Q_u,\; \sigma \in A'_q \}$ the set
			of Player $\B$'s states;
			\item For all $q \in Q^\B_u$, we have $A'_q := \{ a_{\mathsf{id}} \}$ for a
			new fresh action $a_{\mathsf{id}}$ and $B'_q := B_q$;
			\item $\distribSet' := \{ d_{q,\sigma} \mid q \in Q_u,\; \sigma \in A'_q \}
			\uplus \{ d_{q,\sigma,b} \mid q \in Q_u,\; \sigma_q \in A'_q, b \in B \}$;
			\item For all $q \in Q_u$ and $\sigma \in A'_q$ we have:
			$\delta'(q,\sigma,b_{\mathsf{id}}) := d_{q,\sigma_q}$;
			\item For all $q \in Q_u$, $\sigma \in A'_q$ and $b \in B_q$, we have:
			$\delta'((q,\sigma),a_\mathsf{id},b) := d_{q,\sigma,b}$;
			\item For all $q \in Q_u,\; \sigma \in A'_q$, we have
			$\distribFunc'(d_{q,\sigma})((q,\sigma)) := 1$. Furthermore, for all $b \in
			B'_q$ and $q' \in Q$, we have:
			\begin{displaymath}
			\distribFunc'(d_{q,\sigma_q,b})(q') := 
			\begin{cases}
			\sum_{a \in A} \sigma_q(a) \cdot \distribFunc(\delta(q,a,b))(q') & \text{ if
			} q' \in Q_u\\
			\sum_{a \in A} \sigma_q(a) \cdot \distribFunc(\delta(q,a,b))[Q \setminus Q_u]
			& \text{ if } q' = \Aconc_W\\
			\end{cases}
			\end{displaymath}
			\item Finally, for all $q \in Q_u$, we have $\colFunc'(q) := \colFunc(q)$ and
			for all $\sigma \in A'_q$: $\colFunc'(q,\sigma) := k$ for some arbitrary color
			$k \in \colSet$.
		\end{itemize}
	\end{definition}	
	
	Let us now proceed to the proof of
	Theorem~\ref{thm:transfer_memory_finite_choice}. 
	\begin{proof}
		As for the proof of Lemma~\ref{lem:value_one_Gu}, we can show that if there is
		a subgame optimal strategy that has finite choice, then for all $u \in V_{\G}
		\setminus \{ 0 \}$, all the states in the game $\G^\mathsf{tb}_u$ have a
		positive value. Indeed, if Player $\A$ plays the subgame optimal strategy that
		has finite choice in the turn-based game $\G_u^\mathsf{tb}$, we obtain the same
		MDPs (modulo intermediate states colored with $k$) than the MDPs obtained in the
		concurrent game $\G$ (restricted to $Q_u$) where Player $\A$ plays the same
		subgame optimal strategy that has finite choice. Furthermore, all the games
		$\G^\mathsf{tb}_u$ are finite and turn-based. We can then apply the same proof
		than for Theorem~\ref{thm:transfer_memory}: by
		Theorem~\ref{thm:all_not_zero_implies_almost_sure}, there exists a subgame
		almost-surely winning strategy in all games $\G_u$ for $u \in V_{\G} \setminus
		\{ 0 \}$. We then obtain a subgame optimal strategy by gluing these strategies
		into one, this is given by Lemma~\ref{prop:glued_strategy}. Then again, if all
		strategies $\s_\A^u$ are positional for all $u \in V_{\G} \setminus \{ 0 \}$,
		then so is the glued strategy $\s_\A[(\s_\A^u)_{u \in V_{\G} \setminus \{ 0
			\}}]$. 		
	\end{proof}
	
	\subsection{Proof of Corollary~\ref{coro:transfer_memory_finite_choice}}
	\label{proof:coro_transfer_memory_finite_choice}
	Let us formally define these objectives and argue that they have aneutral color
	and that they are \textsf{PSAWT}. 
	
	Consider first the parity objective. It is formally defined in
	Definition~\ref{def:parity_buchi_cobuchi}. The least color (which is an integer
	for the parity objective) in $\colSet$ is straightforwardly a neutral color
	(since we consider the maximum of the colors seen infinitely often).
	Furthermore, the parity objective is \textsf{PSAWT}, as shown for instance in
	\cite{DBLP:conf/soda/ChatterjeeJH04,DBLP:conf/fossacs/Zielonka04}.
	
	Let us now define the mean-payoff objective. 
	\begin{definition}[Mean-payoff]
		Let $\colSet := Q \cap [0,1]$ and $m \in \colSet$. The \emph{mean-payoff}
		objective $W_\mathsf{MP(m)}$ w.r.t. $m$ is $W_\mathsf{MP(m)} := \{ \rho \in
		\colSet^\omega \mid \limsup_{n \rightarrow \infty} \frac{1}{n+1} \sum_{i = 0}^n
		\rho_i \geq m \}$.
	\end{definition}
	Given a mean-payoff objective $W_\mathsf{MP(m)}$ for some $m \in Q \cap [0,1]$,
	one can see that $m$ is in fact a neutral color for $W_\mathsf{MP(m)}$.
	Furthermore, it is also \textsf{PSAWT}, as proved in
	\cite{liggett1969stochastic}. 
	
	Finally, consider the generalized Büchi objective. 
	\begin{definition}[Generalized Büchi]
		Let $\colSet \subseteq \N$ be a finite subset of integers. A \emph{generalized
			Büchi} objective is an intersection of Büchi objectives on $\colSet$. 
	\end{definition}
	This objective has a neutral color, up to adding a fresh color that does not
	appear in any of the intersected Büchi objectives. Furthermore, it is also
	\textsf{PSAWT}, as proved in \cite{DBLP:conf/qest/ChatterjeeAH04} as a sub-class
	of upward-closed Muller objectives. 
	
	The proof of Corollary~\ref{coro:transfer_memory_finite_choice} is then direct.
	\begin{proof}
		By Theorem~\ref{thm:fair_strategy_uniformly_optimal}, there is subgame optimal
		strategy in $\G$ that is finite choice. Since the objective considered has a
		neutral color and is \textsf{PSAWT}, it follows that there is a positional
		subgame optimal strategy by Theorem~\ref{thm:transfer_memory_finite_choice}.
	\end{proof}
\end{document}